\documentclass[a4paper,oneside,12pt,final]{article}
\pdfoutput=1

\usepackage{amsfonts}
\usepackage{amsmath}
\usepackage{amssymb}
\usepackage{amsthm}
\usepackage[mathscr]{eucal}
\usepackage{bbold}
\usepackage{braket}
\usepackage{subcaption}
\usepackage{color}
\usepackage{dsfont}
\usepackage{framed}
\usepackage{mathtools}
\usepackage{physics}
\usepackage[normalem]{ulem}
\usepackage{showlabels}
\usepackage{tensor}
\usepackage{thmtools}
\usepackage{thm-restate}
\usepackage{ulem}
\usepackage{tikz}
\usetikzlibrary{math} 
\usetikzlibrary{calc}
\usepackage{setspace}
\usepackage{graphicx}

\usepackage{longtable}

\usepackage{enumitem}

\usepackage{subcaption}
\usepackage{float}
\usepackage{afterpage}
\usepackage{jheppub}
\usepackage{cleveref}

\definecolor{SHCcolor}{rgb}{0.9,0.1,0.1}

\definecolor{VHcolor}{rgb}{0.7,0.3,0.9}

\definecolor{MRacolor}{rgb}{0.1,0.7,0.1}

\definecolor{MRocolor}{rgb}{0.1,0,1}

\definecolor{maxcolor}{rgb}{1,0.03,0}
\definecolor{dblue}{rgb}{0.25,0.03,0.8}
\definecolor{vecolor}{rgb}{0.7,0.3,0.9}
\definecolor{sergiocolor}{rgb}{0.1,0.7,0.1}
\definecolor{outcolor}{rgb}{1,0.55,0}
\definecolor{tocolor}{rgb}{0.2,0.5,0.2}

\newcommand{\ignore}[1]{}

\newtheorem*{claim}{Claim}

\newtheorem{nthm}{Theorem}[]
\newtheorem{nlemma}{Lemma}

\newtheorem{nconj}{Conjecture}

\newtheorem{ncor}{Corollary}
\newtheorem{ndefi}{Definition}

\renewcommand{\[}{\left[}

\renewcommand{\emptyset}{\varnothing}

\renewcommand{\implies}{\quad\Rightarrow\quad}
\newcommand{\st}{\, : \,}

\newcommand{\R}{\mathbb{R}}

\newcommand{\N}{{\sf{N}}}
\newcommand{\M}{{\sf{M}}}
\newcommand{\D}{{\sf{D}}}
\newcommand{\E}{{\sf{E}}}
\newcommand{\I}{\mathscr{I}}
\newcommand{\J}{\mathscr{J}}
\newcommand{\K}{\mathscr{K}}
\newcommand{\cut}{\mathscr{C}}
\newcommand{\conf}{\mathcal{C}}

\newcommand{\mU}{U^*}
\newcommand{\tgm}[1]{G_{#1}}
\newcommand{\htgm}[1]{\hat{G}_{#1}}
\newcommand{\gm}[1]{\widetilde{G_{#1}}}
\newcommand{\mcs}[1]{\breve{#1}} 

\newcommand{\es}{\varnothing}

\makeatletter\def\@fpheader{~}\makeatother
\newcommand{\er}{\mathscr{R}} 
\newcommand{\erssa}{\er_{_\text{SSA}}} 
\newcommand{\erq}{\er_{_\text{Q}}} 
\newcommand{\erh}{\er_{_\text{H}}} 

\tikzset{cutedgestyle/.style={very thick,color=cutedgecolor}}  
\tikzset{edgestyle/.style={very thick,color=edgecolor}}  
\tikzset{edgeweightstyle/.style={color=edgecolor!60!black}}  
\tikzset{ugofastyle/.style={->,densely dotted, thick,color=ugofacolor}}     
\tikzset{ugobastyle/.style={<-, densely dotted,thick,color=ugobacolor}}     
\tikzset{bfstyle/.style={<->,densely dotted, thick,color=ugofacolor}}     
\definecolor{Acolor}{rgb}{0.8,0.5,0.5}  
\definecolor{Bcolor}{rgb}{0.5,0.8,0.5}  
\definecolor{Ccolor}{rgb}{0.55,0.55,0.85}  
\definecolor{Dcolor}{rgb}{0.9,0.8,0.}  
\definecolor{Ecolor}{rgb}{0.7,0.5,0.8}  
\definecolor{Fcolor}{rgb}{0.5,0.35,0.2}  
\definecolor{Ocolor}{rgb}{0,0,0}  
\definecolor{bvcolor}{rgb}{0.6,0.6,0.6}  
\definecolor{cutvertcolor}{rgb}{0,0,0}  
\definecolor{cutedgecolor}{rgb}{1,0,0}  
\definecolor{edgecolor}{rgb}{0.8,0.8,0.5}  
\definecolor{ugofacolor}{rgb}{0.3,0.3,0.7}  
\definecolor{ugobacolor}{rgb}{0.7,0.4,0.7}  

\definecolor{edgeweightcolor}{rgb}{1,0,0}  

\makeatletter
\newcommand{\bigplus}{%
  \DOTSB\mathop{\mathpalette\mattos@bigplus\relax}\slimits@
}
\newcommand\mattos@bigplus[2]{%
  \vcenter{\hbox{%
    \sbox\z@{$#1\sum$}%
    \resizebox{!}{0.9\dimexpr\ht\z@+\dp\z@}{\raisebox{\depth}{$\m@th#1+$}}%
  }}%
  \vphantom{\sum}%
}
\makeatother

\title{The holographic entropy cone from marginal independence}
 
\author[a]{Sergio Hern\'andez-Cuenca}
\emailAdd{sergiohc@ucsb.edu}
 
\author[b]{Veronika E. Hubeny}
\emailAdd{veronika@physics.ucdavis.edu}
 
\author[c]{Massimiliano Rota}
\emailAdd{m.rota@uva.nl}
 
\affiliation[a]{Department of Physics, University of California, Santa Barbara, CA 93106, USA}
\affiliation[b]{Center for Quantum Mathematics and Physics (QMAP)\\ 
Department of Physics \& Astronomy, University of California, Davis CA, USA}
\affiliation[c]{Institute for Theoretical Physics, University of Amsterdam,
Science Park 904, Postbus 94485, 1090 GL Amsterdam, The Netherlands}

\abstract{
The holographic entropy cone characterizes the relations between entanglement entropies for a spatial partitioning of the boundary spacetime of a holographic CFT in any state describing a classical bulk geometry.
We argue that the holographic entropy cone, for an arbitrary number of parties, can be reconstructed from more fundamental data determined solely by subadditivity of quantum entropy. We formulate certain conjectures about graph models of holographic entanglement, for which we provide strong evidence, and rigorously prove that they all imply that such a reconstruction is possible. Our conjectures (except only for the weakest) further imply that the necessary data is remarkably simple. In essence, all one needs to know to reconstruct the holographic entropy cone, is a certain subset of the extreme rays of this simpler ``subadditivity cone'', namely those which can be realized in holography.
This recasting of the bewildering entanglement structure of geometric states into primal building blocks paves the way to distilling the essence of holography for the emergence of a classical bulk spacetime.
}

\begin{document}
 
 
\maketitle

\section{Introduction}
\label{sec:intro}

In recent years there has been significant progress in understanding how in the gauge/gravity duality \cite{Maldacena:1997re,Gubser:1998bc,Witten:1998qj} the semiclassical bulk physics is encoded in the boundary theory. 
Drawing on \cite{Almheiri:2014lwa,Jafferis:2015del}, it was shown
\cite{Dong:2016eik,Cotler:2017erl} that if one has access to a subsystem of the boundary, one can then reconstruct local operators inside a bulk region known as the entanglement wedge \cite{Headrick:2014cta}. However, this reconstruction of bulk operators assumes the knowledge of the classical geometric background.  To obtain a more complete understanding of the holographic encoding of the bulk physics, one would like to characterise which states of the boundary field theory are dual to classical geometries, and for all such cases, decode the bulk metric directly from the boundary data.  

Intuition from the scale/radius duality suggests that the deeper in the bulk we wish to see, the more non-local (in a suitable sense)  the CFT probe we need.  One particularly convenient class of boundary observables are built from lightlike objects.  Perhaps the simplest and most accessible ones are the `bulk-cone singularities' \cite{Hubeny:2006yu} 
whereby endpoints of null geodesics through the bulk (which can be used to reconstruct the conformal metric in the bulk region probed by such geodesics, cf.\ e.g.\ \cite{Hammersley:2006cp,Folkestad:2021kyz}) are visible directly in the singularities of boundary correlators.
This can be generalized to the so-called `bulk-point singularities' of Landau diagrams corresponding to $n$-particle scattering \cite{Maldacena:2015iua}, used for the construction of the `light-cone cuts' \cite{Engelhardt:2016wgb,Engelhardt:2016crc,Hernandez-Cuenca:2020ppu}, which offer a much more elegant extraction from much more complicated and harder-to-access CFT data.
However to reach even deeper, into causally inaccessible regions, one needs a CFT probe implemented by a spacelike  construct in the bulk.  One particularly natural such geometrical object is a codimension-$2$ extremal surface.\footnote{\, From the purely geometrical level, the motivation is that higher-dimensional surfaces probe deeper (when comparing amongst various-dimensional bulk extremal surfaces anchored within a fixed-radius region of the boundary of a given asymptotically-AdS spacetime) \cite{Hubeny:2012ry}. }
Knowing the proper areas for a family of such surfaces can be inverted to extract the bulk metric (including the conformal factor), again within the bulk regions reached by such surfaces.\footnote{\, Early proof of principle was carried out in e.g.\ \cite{Bilson:2010ff} and an argument for uniqueness of the recovered metric in 4-dimensions was given in \cite{Bao:2019bib}, which also reviews further approaches to bulk metric reconstruction.
}

Having given a bulk motivation for codimension-$2$ extremal surfaces as providing a particularly natural and deep probe of the bulk geometry, it is intriguing to note that even from the CFT viewpoint, such surfaces arise very naturally, in the context of holographic entanglement entropy.  
Here a key role was played by the celebrated RT/HRT formula \cite{Ryu_2006,Hubeny_2007} (collectively shortened to HRRT), which computes the von Neumann entropy of boundary subsystems in terms of the area of certain bulk surfaces. This ``geometrisation'' of correlations \cite{VanRaamsdonk:2010pw}, which more recently has also been observed for other information quantities (such as for example \cite{Takayanagi:2017knl,Dutta:2019gen}), seems to indicate that from the boundary point of view, a possible characterisation of bulk \textit{geometric states} could be formulated in terms of certain features
of their entanglement structure. Motivated by this, we would like to use the HRRT formula to extract as much information as possible about the entanglement structure of geometric states in holography.

In general, one way to investigate the entanglement structure of a given class of states  
utilizes an analysis of constraints, which typically take the form of inequalities satisfied by certain information quantities. Since we are interested in the implications of the HRRT formula, here we will focus on inequalities satisfied by linear combinations of von Neumann entropies of various boundary subsystems.

Two important examples of these classes of inequalities are \textit{subadditivity} (SA) and \textit{strong subadditivity} (SSA), the saturation of either of which has a clear implication for the entanglement structure of 
a given density matrix. The saturation of SA, or equivalently the vanishing of the mutual information, is associated with the absence of any form of correlation between 
a specified pair of subsystems, and the corresponding factorisation of the density matrix. As we will see, this fact will play a central role in our construction. The saturation of SSA on the other hand is associated with a particular entanglement structure, commonly known as quantum Markov chain, which is central for the theory of quantum error correction and recovery maps \cite{Hayden:2004}.\footnote{\, See \cite{Casini:2017roe} for a discussion of this property in QFT.}$^,$\footnote{\, The fact that saturation of SSA will not play a fundamental role in our characterization of geometric states is related to the fact that by MMI (see below) such saturation can only be achieved when SA is also saturated (see \cite{Hayden:2011ag} for more details).}

It is important to notice that while the saturation of SA and SSA is associated with specific entanglement structures, the inequalities themselves do not characterise any particular structure because they are satisfied by all quantum states. For restricted classes of states, however, the von Neumann entropies of various subsystems might satisfy additional inequalities. This can be easily seen for example for the case of classical probability distributions, where the von Neumann entropy reduces to the Shannon entropy, which in addition to the inequalities mentioned above also satisfies monotonicity.

The fact that in holography the entanglement structure of geometric states is constrained, beyond that of arbitrary quantum states, became evident with the work of \cite{Hayden:2011ag}, which proved an inequality known as \textit{monogamy of mutual information} (MMI).\footnote{\, It is straightforward to see that outside the holographic setting,
this inequality can be violated, for example, by a $4$-party GHZ state.} It is then interesting to ask 
what are all
other such  inequalities, and a systematic search was initiated in \cite{Bao:2015bfa}. This work introduced the notion of the \textit{holographic entropy cone} (HEC) and proved that this cone is polyhedral for an arbitrary number of parties $\N$, implying that for any $\N$ there exists only a finite number of non-redundant\footnote{\, An inequality is non-redundant if it is not implied by other inequalities.} inequalities which define the \emph{facets}. It also derived a set of new inequalities for five parties, which was later proved to be the complete set in \cite{HernandezCuenca:2019wgh}. For more than five parties, it was shown in \cite{Hayden:2016cfa} that the holographic entropy cone is indeed contained in the quantum one \cite{pippenger2003inequalities} for any $\N$, but the detailed structure of the HEC remains mostly unknown.\footnote{\, A new family of inequalities for every odd $\N$ was also found in \cite{Bao:2015bfa}, and argued to be non-redundant with respect to SA and SSA and among themselves. However it remains unclear  whether they genuinely are all facets of the HEC for every $\N$. Computational efforts to construct the complete HEC for $\N=6$ are also ongoing \cite{n6wip}. At the time of writing, more than $4122$ orbits of extreme rays and $182$ orbits of facets have been found.}

One limitation of \cite{Bao:2015bfa} is that, while it offered a tool that can be used to prove if a given inequality is valid (via so-called contraction maps), it did not provide a constructive way of deriving such candidate inequalities in the first place, or even determine if they correspond to facets at all.\footnote{\, Indeed, the proof-by-contraction method is unable to ascertain that a given inequality is \emph{not} valid. In other words, while the method provides a sufficient condition for an inequality to be valid, it remains an open question whether it is also necessary \cite{Avis:2021xnz}.} A step forward in this direction was accomplished by \cite{Avis:2021xnz}, in the form of two systematic algorithms for the construction of the HEC for any number of parties. These algorithms were devised so as to recursively converge towards tighter inequalities, terminating with the finding of all facets of the HEC as one of its outputs. However, even if one were able to somehow find explicit formulae for all the facet inequalities of the HEC, it would still remain totally obscure where they come from and what they mean.

A first attempt at circumventing this limitation was presented in \cite{Hubeny:2018trv,Hubeny:2018ijt}, which introduced the notion of \textit{proto-entropies} and \textit{holographic entropy polyhedron}. Motivated by
the physical requirement of cut-off independence, this work suggested focusing more on topological features of the relevant RT surfaces, specifically their connectivity, rather than on their actual areas. This ultimately translated the search for the inequalities into the search of certain generating configurations, called building blocks, making the problem more combinatorial in nature. One of the features of this approach is that any inequality found by this procedure is guaranteed by construction to be a facet of the polyhedron. Intuitively, one can think of this procedure as being akin to deriving the facets of a polyhedral cone by first finding its extreme rays. In the present paper we continue in this direction. By combining some of the techniques of \cite{Bao:2015bfa} based on graph models of holographic entanglement with the ideas of \cite{Hubeny:2018trv,Hubeny:2018ijt} based on connectivity of entanglement wedges, we take a significant step towards the derivation of the HEC for an arbitrary number of parties. However in doing so, we will also slightly change our perspective on the problem.

As evident from these earlier works, one of the main problems one has to face while trying to derive the facets of the HEC is the complexity of the combinatorics, which is typically characterised by a doubly exponential scaling in the number $\N$ of parties involved. In addition, the explicit form of the inequalities might be highly dependent on $\N$, which can make it extremely hard to find a convenient parametrisation of all the inequalities. And even if one could circumvent these complications, it is far from clear how to interpret the inequalities, for example by searching for explicit structures of density matrices which are ruled out by them. Ultimately, if there exists a general lesson to be learned about the entanglement structure of geometric states, it might be very hard to understand what it is merely by looking at a very large number of complicated and seemingly unrelated expressions. For all these reasons we will not look for an explicit derivation of the inequalities. Instead, we will argue that the holographic entropy cone can be derived, at least in principle, from the solution to a much simpler problem. 

For any density matrix on a given number of parties $\N$, and an arbitrary purification of it, one can consider a pair of subsystems (possibly composite and possibly including the purifier) and compute their mutual information to determine whether these subsystems are correlated or not. One can then repeat this analysis for any pair of subsystems to determine what was called in \cite{Hernandez-Cuenca:2019jpv} the \textit{pattern of marginal independence} (PMI) of the density matrix. Any PMI, specifying which pairs of subsystems are marginally independent (i.e.\ have vanishing mutual information) and which are not, can be viewed as a linear subspace in a certain vector space called the \emph{entropy space}. In particular, the PMI is the supporting subspace of a face of a cone, the \textit{subadditivity cone} (SAC), built only from the instances of SA for that $\N$.\footnote{\, Said more explicitly, the supporting subspaces of the \emph{facets} of the SAC are the hyperplanes of vanishing mutual information (i.e.\ the saturation of some instance of SA), and their intersections then form the supporting subspaces of the \emph{faces} of the SAC.}
With this structure in hand, one can conversely ask for which PMI does there exist a density matrix corresponding to it. This was dubbed in \cite{Hernandez-Cuenca:2019jpv} the marginal independence problem, and for the restricted class of states corresponding to geometric states in holography, the \textit{holographic marginal independence problem} (HMIP).

The main goal of this work is to argue that the HEC can be fully reconstructed from the solution to the HMIP, and that the solution to this problem amounts to establishing which extreme rays of the SAC can be realized by geometric states. Importantly, as we will explain, if one wants to construct the HEC for a given number of parties $\N$, it will not in general be sufficient to know the solution to the HMIP for the same number of parties. We will argue however that there always exists a finite $\N_{\text{max}}(\N)$ such that the $\N$-party HEC can be constructed from the solution of the $\N_{\text{max}}$-party HMIP.

We will not be able to provide a definite proof that this reconstruction is possible, but we will formulate and discuss certain conjectures on graph models realizing the extreme rays of the HEC which imply that this is the case, and provide evidence in their support. We stress that the focus of this work however is not on any specific algorithm for an explicit reconstruction of the HEC, but rather on the possibility of the reconstruction from \textit{only} the seemingly limited information contained in the solution to the HMIP. A conclusive proof of such a possibility would in fact amount to the proof of a deep equivalence between the information contained in the set of \textit{all} holographic entropy cones (for all values of $\N$) and in the set of \textit{all} PMIs that can be realized in holography. Any deeper question about constraints on the entanglement structure of geometric states should then be formulated in terms of these more fundamental objects. Furthermore, since for the reason mentioned above this equivalence would in general not be attained for any specific value of $\N$, the significance of certain specific objects which are manifestly $\N$-dependent (like the holographic entropy inequalities) should be questioned, as these objects might be significantly affected by structural artifacts of the formulation.

A priori it might seem surprising that all the holographic entropy inequalities can be derived from a simpler structure which only involves SA -- after all, SA is a universal relation which does not `know about' holography. Intuitively, one can imagine that this is ultimately related to the fact that in quantum field theory the values of the entropies are typically immaterial because of the cut-off dependence, and that the saturation of SA (at leading order in $N$) 
is sensitive solely to the connectivity of the entanglement wedge. And even though certain `balanced' combinations of entropies such as the mutual information can be often ascribed a finite value independent of the cutoff \cite{Sorce:2019zce}, such numerical data can typically be `dialed' by for example deforming the state or the subsystem specification.  Since the HEC is specified by the limit to which such dialing can be pushed, it should not involve rescalable numerical values.  In other words, for delimiting the HEC, it should only matter whether such finite quantities are zero or non-zero, which is indeed borne out in our results. Even if the starting point in the development of our framework will be the graph models from \cite{Bao:2015bfa}, where the edge weights can be dialed at will, any dependence on a specific choice will be effectively modded out by the fact that we will formulate all our results purely in terms of equivalence classes of graph models inspired by the proto-entropies of \cite{Hubeny:2018trv,Hubeny:2018ijt}.

The structure of the paper is as follows. In \cref{sec:basics} we review some of the basic definitions which were already used in previous works. In \cref{sec:tools} we introduce the main tools which will be used in later sections, in particular the notion of an equivalence class of graph models, specified by a \textit{min-cut structure} on a \textit{topological graph model} of holographic entanglement, and of its corresponding \textit{min-cut subspace}. Section \ref{sec:cone_marginal} reviews the concept of marginal independence from \cite{Hernandez-Cuenca:2019jpv}, and establishes a connection between patterns of marginal independence and min-cut subspaces for a certain class of tree graphs. Section \ref{sec:recolorings} analyses how min-cut structures and subspaces transform when one varies the number of parties, and generalizes the results for tree graphs from \cref{sec:cone_marginal}. All these tools will then be used in \cref{sec:logic}, where we introduce our conjectures and discuss what evidence we have to support them, how they are related to each other, and their implications for the derivation of the holographic entropy cone. We conclude in \cref{sec:discussion} with a discussion of the main questions that still need to be answered in order to obtain a full characterization of the holographic entropy cone for an arbitrary number of parties, and comments on other future directions. For convenience, in \cref{tab:summary} we specify our font conventions and collect the notation for the main constructs we use, organized by what spaces they live in, with reference to the place in the main text where they are first defined. Throughout the text we will occasionally decorate various symbols to stress particular choices of the corresponding objects which satisfy additional requirements or have important additional features.

\vspace{0.8cm}
{\footnotesize
\begin{longtable}{|| c | l | l ||}
\hline \hline
symbol 
	& definition
	& reference 
	\\
\hline \hline  
\multicolumn{3}{|| c ||}{Colors, collections of colors and colorings:} 
	\\
\hline
$\N$ 
	& number of parties
	& \cref{sec:intro} 
	\\
$\D$ 
	& number of polychromatic subsystems
	& \cref{ssec:cones} 
	\\
\hline
$[\N+1]$ 
	& set of colors, including the purifier ($\N+1$)
	& \cref{ssec:cones} 
	\\
$\ell$ 
	& single color (possibly the purifier)
	& \cref{ssec:cones} 
	\\
$\I,\J,\ldots$ 
	& polychromatic subsystem that does not include the purifier
	& \cref{ssec:cones} 
	\\
$\underline{\I},\underline{\J},\ldots$ 
	& polychromatic subsystem that might include the purifier
	& \cref{ssec:cones} 
	\\
$\beta$
	& coloring of boundary regions / boundary vertices
	& \cref{sec:basics} 
	\\
\hline \hline  
\multicolumn{3}{|| c ||}{Graph constructs:} \\
\hline
$\E$
	& number of edges in a graph $G=(E,V)$
	& \cref{ssec:graphs}
	\\	
$w_e$ or $w(e)$
	& weight of an edge $e$
	& \cref{ssec:graphs}
	\\
$\partial V\subseteq V$
	& set of ``boundary'' vertices
	& \cref{ssec:graphs}
	\\
\hline 
$\tgm{\N}$
	& topological graph model of holographic entanglement
	& \cref{ssec:graphs}
	\\
$\gm{\N}$
	& graph model of holographic entanglement
	& \cref{ssec:graphs}
	\\
\hline 
$U\subseteq V$ 
    & an arbitrary cut
    & \cref{ssec:graphs}
    \\
$U_{\I}$ 
	& $\I$-cut
	& 	\cref{ssec:graphs} 
	\\
$\cut(U_{\I})$ or $\cut_{\I}$
	& set of cut edges of an $\I$-cut $U_{\I}$
	& \cref{eq:cutedges}
	\\
$\mU_{\I}$ or ${\mU_{\I}}^{\alpha}$
	& min-cut for $\I$
	& 	\cref{ssec:graphs} 
	\\
$\mathscr{U}_{\I}$  
	& set of min-cuts for $\I$
	& \cref{sec:tools}
	\\
$\mathfrak{m}$
	& min-cut structure on a topological graph model
	& \cref{def:mincutstructure}
	\\
\hline \hline  
\multicolumn{3}{|| c ||}{Objects in entropy space $\R^{\D}$:} \\
\hline
$S_{\I}$ 
	& von Neumann entropy of subsystem $\I$
	&  \cref{ssec:cones} 
	\\
$\mathbf{S}$ 
	& entropy vector
	&  \cref{ssec:cones} 
	\\
\hline
HEC$_{\N}$
	& $\N$-party holographic entropy cone
	&  \cref{ssec:configurations} 
	\\
SAC$_{\N}$
	& $\N$-party subadditivity cone
	& \cref{def:SAcone}
	\\
\hline
$\mathcal{S}$  
	& S-cell 
	& \cref{def:scell}
	\\
$\mathbb{S}$ 
	& min-cut subspace 
	& \cref{def:vspace} 
	\\
\hline 
$\mathbb{P}$
	& pattern of marginal independence (PMI)
	& \cref{def:patterns}
	\\
$\pi$
	& map that gives the PMI of a graph model
	& \cref{subsec:pmireview}
	\\
$\Pi(\mathbb{P})$
	& matrix of MI instances that determine a PMI
	& \cref{eq:pmimatrix}
	\\
\hline \hline  
\multicolumn{3}{|| c ||}{Objects in the space of edge weights $\R^{\E}$:} \\
\hline
$\mathbf{w}$
	& weight vector
	& \cref{ssec:graphs}
	\\
$\mathcal{W}$ 
	& W-cell   
	& \cref{subsec:gmnodeg}
	\\
$\mathbb{W}$ 
	& subspace specified by degeneracy equations 
	& \cref{sssec:gmdeg}
	\\
$\mathbf{w}_{\mathcal{W}}$
	& weight vector in $\mathcal{W}$
	& \cref{subsec:properties}
	\\
$\mathbf{x}_{\mathcal{W}}$
	& extreme ray of $\overline{\mathcal{W}}$
	& \cref{subsec:properties}
	\\
\hline 
$\mathbf{\Gamma}_{\I}$
	& incidence vector
	& \cref{eq:incidence}
	\\
$\Gamma$
	& map from W-cell to S-cell
	& \cref{subsec:gmnodeg}
	\\	
\hline \hline
\multicolumn{3}{|| c ||}{Geometric constructs:} \\
\hline
$\overline{\mathcal{F}}_d$
	& $d-$dimensional face of a cone
	& \cref{subsec:properties}
	\\
$\mathcal{F}_d$
	& interior of a face $\overline{\mathcal{F}}_d$
	& \cref{subsec:properties}
	\\
$\mathbb{F}_d$
	& minimal linear supporting subspace of a face $\overline{\mathcal{F}}_d$
	& \cref{subsec:properties}
	\\
\hline \hline
\multicolumn{3}{|| c ||}{Recolorings:} \\
\hline
$\beta^{\downarrow}$
	& recoloring that reduces $\N$
	& \cref{subsec:cg}
	\\
$\phi$
	& coarse-graining induced by the recoloring $\beta^{\downarrow}$
	& \cref{subsec:cg}
	\\
$\Phi_{\N \rightarrow\N'}$
	& projection associated to the coarse-graining $\phi$
	& \cref{eq:colorproj}
	\\
\hline
$\beta^{\uparrow}$
	& recoloring that increases $\N$
	& \cref{subsec:fg}
	\\
$\mathfrak{m}\big\uparrow\!_{_{\N'}}$
	& set of fine-grainings of a min-cut structure $\mathfrak{m}$ to $\N'$ parties
	& \cref{eq:mstructurelifts} 
	\\
\hline \hline
\caption{A reference table summarizing the main notation and terminology used. To facilitate orientation, we typically reserve sans-serif font for natural numbers, mathcal for geometrical regions, mathbb for linear subspaces, bold face letters for vectors, mathscript for sets, and mathfrak for more complicated but important constructs. 
}
\label{tab:summary}
\end{longtable}
}

\section{Basic definitions and notation}
\label{sec:basics}

In this section we briefly review some of the main definitions, most of which were already used in previous works \cite{pippenger2003inequalities,Bao:2015bfa,Hubeny:2018ijt,Hubeny:2018trv,Bao:2020mqq,Avis:2021xnz}. In \cref{ssec:cones} we introduce the concept of entropy cones in quantum mechanics. In \cref{ssec:configurations} we review the holographic set-up, the definition of the holographic entropy cone and the concept of proto-entropies. In \cref{ssec:graphs} we review the basic definitions of the graph models of holographic entropies. Finally, in \cref{ssec:graph_lemmas} we list some useful immediate consequences of the main definitions. For more details the reader is referred to the original works.

\subsection{Entropy cones}
\label{ssec:cones}

Consider a Hilbert space which is a tensor product of $\N$ factors,
\begin{equation}
\label{eq:hilbert-space}
\mathcal{H}\coloneqq\mathcal{H}_1\otimes\mathcal{H}_2\otimes\dots\otimes\mathcal{H}_{\N},
\end{equation}
and a density matrix $\rho$ acting on it. For any non-empty subset $\I\subseteq[\N]\coloneqq \{1,2,\dots,\N\}$ of these factors, the von Neumann entropy is defined as
\begin{equation}
S_{\I}\coloneqq S(\rho_{\I})=-\Tr\left(\rho_{\I}\log\rho_{\I}\right),
\end{equation}
where
\begin{equation}
\rho_{\I}\coloneqq\Tr_{\mathcal{H}_{[\N]\setminus\I}}\rho,\qquad 
\text{with} \qquad
\mathcal{H}_{\I}\coloneqq\bigotimes_{\ell\in\I}\mathcal{H}_\ell,
\end{equation}
is the reduced density matrix, or \textit{marginal}, for the subsystem $\I$. The \textit{entropy vector} corresponding to the density matrix $\rho$ is the ordered\footnote{\, We will follow our previous convention of ordering the entropies first by cardinality of $\I$ and then lexicographically, though the actual order will not play a significant role in what follows.
In all specific examples, we will use alternate (and more conventional) notation of letters $A,B,C, \ldots$ instead of numbers $1,2,3,\ldots$ to denote colors, with the letter $O$ reserved for the purifier. 
} collection of entropies of all its marginals, namely,
\begin{equation}
\mathbf{S}(\rho)=\{S_{\I}\; \text{for all}\; \I\}.
\end{equation}
The vector space where entropy vectors live is $\R^{\D}$, where $\D=2^{\N}-1$, and will be referred to as \textit{entropy space}.
The labeling of the Hilbert space factors in \cref{eq:hilbert-space} will be called a \textit{coloring}, the subscripts $\ell\in[\N]$ are referred to as \textit{colors} and any non-empty set of colors $\I$ is a \textit{polychromatic} index.

For a fixed number of parties $\N$, the collection of all entropy vectors for all possible Hilbert spaces and density matrices was shown by \cite{pippenger2003inequalities} to be a convex cone\footnote{\, A closed convex cone is a set of vectors such that for any two vectors $\mathbf{v}_1,\mathbf{v}_2$ in the set, the \textit{conical combination} $\alpha\mathbf{v}_1+\beta\mathbf{v}_2$ (where $\alpha,\beta\geq 0$) also belongs to the set.} known as the $\N$-party \textit{quantum entropy cone} (QEC$_{\N}$).\footnote{\, More precisely, it is the topological closure of this set which is a convex cone, while the set itself has a more complicated structure.} 
By construction, the QEC$_{\N}$ is clearly symmetric under an arbitrary permutation of the $\N$ parties, as can be seen by just permuting Hilbert space factors. In fact, it exhibits a larger symmetry group of permutations of $[\N+1]$, as we now explain.

For a density matrix $\rho$, a \textit{purification} is any pure state $\ket{\psi}$ in an enlarged Hilbert space
\begin{equation}
\label{eq:purification}
\underline{\mathcal{H}}\coloneqq\mathcal{H}_1\otimes\mathcal{H}_2\otimes\dots\otimes\mathcal{H}_{\N}\otimes\mathcal{H}_{\N+1}
\end{equation}
such that
\begin{equation}
\rho=\Tr_{\mathcal{H}_{\N+1}}\ket{\psi}\bra{\psi} \, .
\end{equation}
We will refer to the additional auxiliary subsystem $\mathcal{H}_{\N+1}$ as the \textit{purifier}\footnote{\, The purifier will often be referred to as  color $\N+1$.} and denote a non-empty subset of $[\N+1]$ by an underlined index $\underline{\I}$. Occasionally we will take complements of (not necessarily underlined) polychromatic indices, and we will always define these complements with respect to the set $[\N+1]$.\footnote{\, Note that according to this definition the set of non-underlined indices is not closed under this operation, however to simplify the notation we will write the complement of $\I$ as  $\I^{\complement}$ instead of  $\underline{\I^{\complement}}$.} Since for any pure state the entropy of a subsystem is equal to the entropy of its complement, the entropies of all the subsystems of $\ket{\psi}$ are already encoded in the entries of $\bf{S}(\rho)$.\footnote{\, By convention, given a pair of complementary subsystems $(\underline{\I},\underline{\I}^\complement)$ we have $S_{\underline{\I}}=S_{\underline{\I}^\complement}$ and we denote the entropy of each of them by the index which does \textit{not} include the purifier.} Any permutation of the $\N+1$ factors in \cref{eq:purification} will map an $\N$-party entropy vector to another $\N$-party entropy vector, resulting in the extended symmetry mentioned above. In what follows, when we consider permutations of entropy vectors or inequalities, we will always mean permutations of $[\N+1]$.

For $\N=2,3$, the QEC$_{\N}$ is known to be a polyhedral cone and can therefore be specified by a finite set of inequalities.\footnote{\, A polyhedral cone is said to be \textit{pointed} when it does not contain any non-trivial linear subspace, and in what follows all polyhedral cones that we will consider will be pointed. Any pointed polyhedral cone can equivalently be described as the \textit{conical hull} (the set of all possible conical combinations) of the set of its extreme rays.} In the $\N=2$ case, the facets are given by the permutations of \textit{subadditivity} (SA),\footnote{\, These include the \textit{Araki-Lieb inequality} $S_1+S_{12}\geq S_{2}$.} 
\begin{equation}
\label{eq:subad}
S_1+S_2\geq S_{12},
\end{equation}
while for $\N=3$ there exists a new inequality known as \textit{strong subadditivity} (SSA),\footnote{\, Its permutations include \textit{weak monotonicity} $S_{12}+S_{23}\geq S_1+S_3$.}
\begin{equation}
\label{eq:ssa}
S_{12}+S_{23}\geq S_{2}+S_{123}. 
\end{equation}
It is important to notice that for $\N=3$, the permutations of \cref{eq:ssa} 
constitute only a proper subset of the full set of facet-defining inequalities of the cone. The additional facets correspond to certain particular instances of SA. For example, one can easily verify that \cref{eq:subad} specifies a facet while the ``lift''
\begin{equation}
\label{eq:salift}
S_{1}+S_{23}\geq S_{123}
\end{equation}
does not, being just the sum of \cref{eq:subad} and \cref{eq:ssa}.\footnote{\, 
Geometrically, the saturation of this weaker inequality \cref{eq:salift} corresponds to a hyperplane which intersects the boundary of the HEC only along a lower dimensional subspace.} For $\N\geq 4$ the QEC$_{\N}$ is essentially unknown; however, for any $\N$, one can easily construct an ``outer bound'' (a larger cone that contains it) by considering all instances of SA for all possible pairs of disjoint subsets of $[\N+1]$. For any given $\N$, these define the $\N$-party subadditivity cone (SAC$_\N$), an object which will play a central role in our derivation of the holographic entropy cone (see also \cite{Hernandez-Cuenca:2019jpv}).\footnote{\, The reader might already wonder if one could not derive a more stringent bound by also including SSA. This is of course correct, but as we will see it is SA, rather than SSA, which plays a more fundamental role.}

\subsection{Holographic constructions}
\label{ssec:configurations}

Having introduced an $\N$-party entropy space and the QEC$_{\N}$ therein, we now consider the construct of an entropy cone in the context of holography.  A natural splitting of the Hilbert space is achieved\footnote{\, Even if strictly speaking the Hilbert space of a QFT does not factorize, we assume that this splitting is a valid approximation, see \cite{Witten:2018zxz} and references therein for more details on this issue.} by partitioning the space on which the CFT lives, which we now specify in more detail to indicate the generality of the setup.
Consider a (not necessarily connected) asymptotically AdS manifold $\mathcal{M}$ with $\M$ boundaries, $\partial \mathcal{M}=\bigcup_{m=1}^{\M}\partial\mathcal{M}_m$.\footnote{\, The dimension of $\mathcal{M}$ will not play any role.} For each boundary component $\partial\mathcal{M}_m$, consider a Cauchy slice $\Sigma_m$ and a partition of it into an arbitrary number of connected regions $\mathcal{A}_m^i$. Given a number of parties $\N$, we introduce a surjective \textit{coloring} of these regions
\begin{equation}
    \beta:\{\mathcal{A}_m^i\st\forall m,i\}\rightarrow [\N+1].
\end{equation}
In other words, each Hilbert space factor $\mathcal{H}_\ell$, which we label by a color, is associated to some collection of these regions. Given a polychromatic index $\I$, the set of regions which under $\beta$ receive a color $\ell\in\I$ is the preimage $\beta^{-1}(\I)$, and will be referred to as the \textit{subsystem} associated to $\I$ (or even more simply the subsystem $\I$).\footnote{\, We use similar notation and terminology for indices $\underline{\I}$ which include the purifier.}
A choice of such a manifold, Cauchy slice, partition and a coloring defines what we will call an $\N$-party \textit{holographic configuration}, denoted by $\conf_{\N}$.\footnote{\, Throughout our construction, we will be working in the regime where the holographic dual describes a classical bulk spacetime, i.e., we will not consider stringy or quantum corrections; see the discussion in \cref{sec:discussion} for additional comments.} 

Using the HRT prescription \cite{Hubeny_2007}, we can associate an entropy vector $\mathbf{S}(\conf_{\N},g)$ to any given holographic configuration.\footnote{\, The second argument in $\mathbf{S}(\conf_{\N},g)$ is a shorthand meant to indicate the dependence on the bulk spacetime metric $g_{ab}$.  Although it may seem more natural to use the CFT quantity $\rho$, we choose to leave this implicit (as describing any density matrix which gives the bulk $g_{ab}$) to emphasize that in the present context
the entanglement entropy is given by a geometrical construct (namely the relevant set of extremal surfaces).}
However, notice that typically the entropies will not be finite, which occurs whenever an extremal surface is anchored on the AdS boundary, in other words when the corresponding region admits an entangling surface.%
\footnote{\, The only case where this does not happen is when each Cauchy surface $\Sigma_m$ is not partitioned at all, and each region $\mathcal{A}_m$ is then the entire $\Sigma_m$.} To circumvent this problem, \cite{Bao:2015bfa} introduced a cut-off surface, making all the entropies finite at the expense of associating to a configuration an entropy vector $\mathbf{S}_{\epsilon}(\conf_{\N},g)$ which is now cut-off dependent.

Given a number of parties $\N$, one can consider the set of all possible (finite) entropy vectors, for all possible $\N$-party configurations and choices of cut-off. It is then easy to show \cite{Bao:2015bfa} that this set has the structure of a convex cone, and is known as the $\N$-party \textit{holographic entropy cone} (HEC$_{\N}$).\footnote{\, In \cite{Bao:2015bfa} the definition was given for the static case, but the same formulation pertains also in the dynamical case. The reader who is already familiar with these constructions is reminded that this is however not the case for the proof of polyhedrality.} Note that despite being called ``holographic'', the construction of this cone is purely geometric, and in using the HRT formula we implicitly assumed that for any $\mathcal{M}$ the bulk dynamics is dual to the evolution of a tensor product of $\M$ copies of a holographic CFT living on $\partial\mathcal{M}$. In what follows we will always make this assumption, leaving the analysis of this subtlety to future work.\footnote{\, See \cite{Marolf:2017shp} for more details regarding this issue.}

While the use of a cut-off is a convenient computational tool in defining the cone, one might worry that it could render the resulting construct intrinsically ill-defined.  In particular, would the cone be meaningful if its building blocks required a specification of a cut-off?  Fortunately, we can circumvent this subtlety by realizing the cone's extreme rays by configurations where none of the minimal surfaces is anchored to the boundary, specifically by the bulk geometry corresponding to multi-boundary wormholes and subsystems covering the entire connected pieces of the boundary \cite{Bao:2015bfa}. In such configurations the corresponding HRT surface areas are therefore finite, with no need for cut-offs.  

While this proof makes it clear that the ``static'' HEC is a physically well motivated object to study, the construction has some limitations. First of all, one would like to extend the result to the ``covariant'' HEC (see \cite{Wall:2012uf,Rota:2017ubr,Bao:2018wwd,Czech:2019lps} for discussions in this direction), which is not even known to be polyhedral.\footnote{ \,
While there is no well-established definition of such a covariant HEC, the intent is to characterize the entropy vectors for all configurations in all physical time-dependent holographic spacetimes.  Since the same subtlety mentioned above and explored in \cite{Marolf:2017shp} is likely more severe in the time-dependent context, we again retreat to geometrical definition, as a collection of entropy vectors given by the HRT prescription, in spacetimes where the latter applies.  Most simplistically, we would then restrict to classical bulk spacetimes obeying the NEC.}
Second, it would be interesting to understand to what extent the properties of the HEC, even in the static case, depend on multiboundary wormhole solutions and on the choice of configurations where all entropies are finite. Finally, it is interesting to investigate how specific substructures of the cone (for example certain internal regions or portions of its boundary) are related to different holographic configurations, especially the more typical ones which have divergent entropies.

Motivated by this, the work of \cite{Hubeny:2018trv} introduced the concept of \textit{proto-entropy} and \textit{proto-entropy vector}, as manifestly cut-off independent objects associated to a configuration and a metric. Given an arbitrary boundary region $\mathcal{A}$ (not necessarily connected), consider a minimal extremal surface $\xi_{\mathcal{A}}$ whose area computes the entropy of $\mathcal{A}$. Such extremal surface can be composed of multiple disjoint pieces, which may or may not be anchored on the boundary. The proto-entropy of $\mathcal{A}$ is then simply defined as the formal sum of the connected components of $\xi_{\mathcal{A}}$. In practice, one should just imagine the HRT prescription without the area functional, with the sum of areas of connected surfaces replaced by a formal sum of the surfaces themselves. The convenience of this construction is that all surfaces are now treated on the same footing, irrespective of whether they are anchored to the boundary (and therefore have infinite area) or not, and no cut-off is ever introduced. Given a pair $(\conf_{\N},g)$, its proto-entropy vector is then defined as the ordered collection of the proto-entropies of the boundary subsystems.

Using this formulation, \cite{Hubeny:2018trv,Hubeny:2018ijt} suggested an approach to the derivation of new holographic entropy inequalities where it is the connectivity of the entanglement wedges, rather than the area of the surfaces, that plays a central role. While in the present work we will focus on the graph models of \cite{Bao:2015bfa} rather than on proto-entropies, the approach to the reconstruction of the HEC presented here will follow the same intuition. We leave a more detailed analysis of the relationship with proto-entropies, and in particular between the HEC and the \textit{holographic entropy polyhedron} to future work \cite{Hernandez-Cuenca:2022a}.

\subsection{Graph models}
\label{ssec:graphs}

In the previous subsection we gave a definition of a holographic configuration in the most general scenario, where the bulk spacetime is dynamical. We will now restrict to static spacetimes, such that the entropies are computed by minimal (rather than extremal) surfaces. In this case, given a pair $(\conf_{\N},g)$, the minimal surfaces define a partition of the bulk time slice, and we can conveniently describe the configuration by a graph constructed as follows \cite{Bao:2015bfa}. 

To each region in the partition of the bulk time slice we associate a vertex. If a region is adjoining to the boundary, we label the vertex by the color of the adjoining boundary region. If two vertices correspond to two adjoining regions of the bulk time slice, they are connected by an edge. Any such edge corresponds to a piece of a minimal surface and will carry a weight equal to its area.\footnote{\, Pieces of RT surfaces which reach the boundary will correspond to edges with infinite weights. If necessary, one can again imagine introducing a cut-off, and work with a graph having all edge weights finite.}
Note that such graphs can get quite complicated;
for example, even in a simple configuration with regions specified by symmetric distribution of disks on a spatial slice of $\R^{2,1}$, two correlated regions would correspond to a planar graph with 4 vertices and 3 edges, while a graph encoding four pairwise-correlated regions would be a non-planar one with 43 vertices and 90 edges. 
However, we will not need to consider the specific details of such graphs.

The convenience of this representation stems from the fact that it suggests how to formulate an alternative, but \textit{equivalent}, definition of the (static) HEC which is entirely based on graphs, without the need to consider more explicit holographic configurations as introduced above. In essence, one first defines a ``graph model'' of holographic entanglement, with a prescription for how to compute an entropy vector starting from it, and then proves that the set of entropy vectors realized by such graph models is equal to the HEC.\footnote{\, Given an arbitrary graph model one can convert it into a ``canonical form'' using certain entropy-preserving transformations (explicitly reviewed in \cref{sec:gops}). One then shows that the entropy vector of any graph model in canonical form can be explicitly realized by a multiboundary wormhole configuration. See \cite{Bao:2015bfa} for more details.} Since these graph models are the main tool that we will use in the rest of the paper, we now review their original definition in detail.    
 
Consider an undirected graph $G=(V,E)$ with vertex set $V$ and edge set $E$ of unordered pairs of vertices.\footnote{\, For simplicity the graph is also assumed to be simple (it is not a multigraph) and loopless ($i\neq j$ for all $(i,j)\in E$), although without these restrictions the entropy cone would remain unchanged.}
With a slight (but by now well-established) abuse of notation, we
denote a set of \emph{boundary vertices} by $\partial V \subseteq V$. We will refer to the other vertices (if any\footnote{\, As pointed out by \cite{Avis:2021xnz}, the restriction to only $\partial V = V$ corresponds to the study of the cut function, which when written as a vector is well known to span a polyhedral cone whose facets are the subadditive inequalities \cite{Cu85,Fu91}.}) in $V$ as the \emph{bulk vertices}. 
Analogously to what we discussed for configurations, to each boundary vertex we assign a color via the coloring
$\beta:\partial V \to [\N+1]$ and $\beta^{-1}(\I)$ is the preimage of the subsystem $\I$. A graph $G$ together with a specification of boundary vertices and a coloring $\beta$ defines a \emph{topological graph model} of holographic entanglement, which will be denoted by $\tgm{\N}$. 

If this structure is further endowed with a \emph{weight map} $w:E\to\R_{>0}$,  we will call it a \emph{graph model} of holographic entanglement, and denote it by $\gm{\N}$. We will often think of the map $w$ as a \emph{weight vector} $\mathbf{w}=w(E)$
in the space of weights $\mathbb{R}^{\E}$ (where $\E\coloneqq\abs{E}$), consisting of all weights $w_e\coloneqq w(e)$ for all $e\in E$ with some ordering.  
 
Any subset $U\subseteq V$ characterizes a bipartition or \emph{cut} of $G$, which defines a set of \emph{cut edges} $\cut(U)\subseteq E$ as
\begin{equation}
\label{eq:cutedges}
  \cut(U) \coloneqq \{(v, v') \in E \st v \in U ,\; v'\in U^\complement\}
\end{equation}
(where $U^\complement \coloneqq V \setminus U$ is the complementary set of vertices).
For a graph model $\gm{\N}$, the \emph{cut weight} of a cut $U$ is defined as the total weight of its edges 
\begin{equation}
\label{eq:cost}
    \norm{\cut(U)} \coloneqq \sum_{e \in \cut(U)} w(e).
\end{equation}
For any non-empty $\I\subseteq[\N]$, a set $U\subseteq V$ is a \emph{cut homologous to $\I$}, or an \emph{$\I$-cut}, if it contains precisely the boundary vertices colored by $\I$, i.e. if $U\cap\partial V = \beta^{-1}(\I)$. We will denote an arbitrary $\I$-cut by $U_{\I}$ and to simplify the notation, we will occasionally denote the corresponding set of cut edges by $\cut_{\I}$ (as a shorthand for $\cut (U_{\I})$). The minimum cut weight among all $\I$-cuts gives the \emph{entropy} $S_{\I}$ of the associated subsystem, i.e.,  
\begin{equation}
\label{eq:formalentropy}
    S_{\I} \coloneqq \min_{U_{\I}} \norm{\cut(U_{\I})} \ .
\end{equation}
Any $\I$-cut $U_{\I}$ with minimum cut weight $\norm{\cut(U_{\I})}=S_{\I}$ is a \emph{min-cut} for $\I$ and will be denoted by $\mU_{\I}$ (similarly, we will occasionally denote the set of cut edges for a min-cut by $\cut^*_{\I}$). Notice that from the definition it immediately follows that the complement of a min-cut $U^*_{\I}$ is a min-cut for the complementary subsystem $\I^{\complement}$. Min-cuts however are not necessarily unique, as for holographic configurations. We will come back to this ``degeneracy'' in later sections, since it will play an important role.

\subsection{Basic properties of min-cuts}
\label{ssec:graph_lemmas}

Here we collect some basic results about min-cuts that will be useful later. 

For a given topological graph model $\tgm{\N}$, consider an arbitrary subsystem $\I$,\footnote{\, Throughout this subsection we use non-underlined indices $\I$ rather than $\underline{\I}$ to simplify the notation, but all lemmas obviously also hold for subsystems that include the purifier.} an arbitrary $\I$-cut $U_{\I}$, and the induced subgraph\footnote{\, The induced subgraph $G[U]$ of a graph $G=(V,E)$ is the graph $G[U]=(U,F)$ with vertex set $U\subset V$ and edges $F\subseteq E$ connecting only vertices in $U$.} $\tgm{\N}\left[U_{\I}\right]$, with vertex set $U_{\I}$ and boundary vertices $\partial V\cap U_{\I}$ with the coloring inherited from $\tgm{\N}$.
In the special case of a min-cut $\mU_{\I}$, this subgraph is the natural analog of the entanglement wedge (or more precisely the homology region) for the subsystem $\I$.
However, as we now show, each connected component of a min-cut $\mU_{\I}$ has to be connected to the boundary, i.e. it has to contain at least one boundary vertex:%
\footnote{\, We are ignoring the trivial case where $\tgm{\N}$ is disconnected and some of its connected components do not include any boundary vertices. However it should be clear that even in these cases, such disconnected components are never cut by a min-cut.}
\begin{nlemma}[Topological minimality]
\label{lem:topmin}
If the induced subgraph $\tgm{\N}\left[U_{\I}\right]$ of a topological graph model $\tgm{\N}$ on the vertices of a $\I$-cut $U_{\I}$ is disconnected, and there is a connected component which does not include any of the boundary vertices of $\tgm{\N}$, then the cut $U_{\I}$ cannot be a min-cut for any graph model $\gm{\N}$ on $\tgm{\N}$.
\end{nlemma}

\begin{proof}
    Denote by $U^{\varnothing}_{\I}$ a connected component of $\tgm{\N}\left[U_{\I}\right]$ such that $U^{\varnothing}_{\I}\cap\partial V=\varnothing$. Then $U'_{\I}=U_{\I}\setminus U^{\varnothing}_{\I}$ is a new $\I$-cut on $\tgm{\N}$ such that 
    \begin{equation}
    \norm{\cut(U'_{\I})}<\norm{\cut(U_{\I})}.
     \end{equation}
    Hence $U_{\I}$ cannot be a min-cut.
\end{proof}

The second basic property of min-cuts is the graph version of the property of HRRT known as ``entanglement wedge nesting'', which states that entanglement wedges of nested regions must themselves be nested. We will state this property without proof (see \cite{Avis:2021xnz} for more details):

\begin{nlemma}[Nesting]
\label{lem:nesting}
For any graph model $\gm{\N}$, subsystems $\I,\K$ with $\K\subset\I$, and min-cut $\mU_{\I}$ for $\I$, there exists a min-cut $\mU_{\K}$ for $\K$ such that $\mU_{\K}\subset \mU_{\I}$.
\end{nlemma}

Notice that in case of degeneracy, the inclusion $\mU_{\K}\subset \mU_{\I}$ need not necessarily hold for all choices of min-cut pairs.

Using this lemma, one can immediately derive the following result:

\begin{nlemma}[No-crossing]
\label{lem:nocrossing}
For any graph model $\gm{\N}$, and subsystems $\I,\K$ with $\I\cap\K=\emptyset$, there exist min-cuts $\mU_{\I}$ and $\mU_{\K}$ such that $\mU_{\I}\cap\mU_{\K}=\emptyset$.
\end{nlemma}
\begin{proof}
    Consider the subsystem $\I^{\complement}$. Since $\I\cap\K=\emptyset$, it follows that $\K\subset\I^{\complement}$ and by \cref{lem:nesting} there exists min-cuts $\mU_{\I^{\complement}}$ and $\mU_{\K}$ such that $\mU_{\K}\subset\mU_{\I^{\complement}}$. The complement of $\mU_{\I^{\complement}}$ is then a min-cut $\mU_{\I}$ such that $\mU_{\I}\cap\mU_{\K}=\emptyset$.
\end{proof}

Finally, we mention another property which is closely related to nesting and will be particularly useful in later proofs:

\begin{nlemma}[Min-cut decomposition]
\label{lem:mincutdecaltgen}
For any graph model $\gm{\N}$, if the induced subgraph $\gm{\N}\left[\mU_{\I}\right]$ for a min-cut $\mU_{\I}$ is composed of two disjoint components, $$\gm{\N}\left[\mU_{\I}\right] = \gm{\N}\left[U_{\J}\right] \oplus \gm{\N}\left[U_{\K}\right]$$ with $\I = \J \cup \K$, then $U_{\J}$ and $U_{\K}$ are min-cuts for the corresponding subsystems.
\end{nlemma}
\begin{proof}
    Proceed by contradiction. Suppose $U_{\J}$ is not a min-cut, so an actual min-cut $U_{\J}^*$ for $\J$ has $\norm{\cut (U_{\J}^*)} < \norm{\cut(U_{\J})}$.
    By hypothesis the union $U_{\J}^* \cup U_{\K}$ defines a cut for $\I$.
    By subadditivity,\footnote{\, Phrased more generally, this follows from the fact that the min-cut function on a graph is submodular.} $\norm{\cut(U_{\J}^* \cup U_{\K})} \leq \norm{\cut(U_{\J}^*)} + \norm{\cut(U_{\K})} < \norm{\cut(U_{\J})} + \norm{\cut(U_{\K})}$, where we used minimality in the last step.
    But by the disjointness assumption, the right-hand side equals $\norm{\cut(U_{\I}^*)}$. 
    Hence $U_{\J}^* \cup U_{\K}$ is a cut for $\I$ of smaller weight than the claimed min-cut $U_{\I}^*$, a contradiction (and the same argument of course also applies to $\K$).
\end{proof}

An alternate way of stating \cref{lem:mincutdecaltgen} is that if the set of cut edges for a min-cut $\mU_\I$ can be split into two sets of cut edges for cuts of subsystems bipartitioning $\I$, then each set corresponds to a min-cut for that subsystem. Note that this statement can be trivially iterated when $\gm{\N}\left[\mU_{\I}\right]$ is composed of multiple disjoint components.

The min-cut decomposition also makes an immediate connection to SA saturation: if $\gm{\N}\left[U_{\I}^*\right] = \gm{\N}\left[U_{\J}^*\right] \oplus \gm{\N}\left[U_{\K}^*\right]$ with $\I = \J \cup \K$, then $S_{\I}=S_{\J}+S_{\K}$, so $\J$ and $\K$ have no mutual information and are fully decorrelated.\footnote{\, The existence of such decomposition of the induced subgraph $\gm{\N}\left[U_{\I}^*\right]$ is also a necessary condition for the mutual information to vanish, as we will review later in \cref{subsec:hologpmi} (see \cref{lem:vanishingI2}).}

\section{The min-cut subspace of a graph model}
\label{sec:tools}

In the previous section we reviewed the definition of a graph model of holographic entanglement and how an entropy vector is associated to it. We now begin to consider coarser objects associated to graph models. This will be a recurring theme throughout this work, and we will consider an even coarser object in \cref{sec:cone_marginal}.

In \cref{subsec:equivclasses} we introduce the notions of a ``min-cut structure'' and its corresponding ``W-cell'' in the space of edge weights. These allow us to organize graph models into equivalence classes, and we will introduce the main objects associated to an equivalence class, namely the ``S-cell'' and its ``min-cut subspace''. These subspaces will then be used in \cref{sec:logic} to resolve the structure of the holographic entropy cone. The definition is not entirely new, as variations of it were already used previously. Here we sharpen this notion and explore its properties in much greater detail, highlighting in particular the role played by ``degeneracy'', namely the possible coexistence of alternative min-cuts which compute the entropy of a subsystem. 

In \cref{subsec:gmnodeg} and \cref{sssec:gmdeg} we then explain how the S-cell and min-cut subspace of an equivalence class of graph models can be determined from the W-cell. While in \cref{sec:basics} we have defined weights as strictly positive, in some cases it will be useful to consider extremal situations where some of them vanish, and we clarify how to deal with this type of situations in \cref{ssec:vanishing}. In \cref{subsec:properties}, we present certain important properties of W-cells, S-cells and min-cut subspaces which will be used in later proofs. Finally, in \cref{subsec:disconnected}, we explain how for disconnected graphs all these constructs can be obtained from those of the connected ``building blocks''.

Additional important properties of the objects introduced in this section, related to how these transform under recoloring of boundary vertices in a graph, will be analysed in \cref{sec:recolorings}.\footnote{\, In the discussion, \cref{sec:discussion}, we will also briefly comment on the relation between the min-cut subspaces defined here, and similar construct introduced in \cite{Hubeny:2018ijt,Hubeny:2018trv} for proto-entropies and the definition of the holographic entropy polyhedron.}

\subsection{Equivalence classes of graph models}
\label{subsec:equivclasses}

Consider a graph model of holographic entanglement $\gm{\N}$, as defined in \cref{ssec:graphs}. As we anticipated above, the min-cut $\mU_{\I}$ for an arbitrary subsystem $\I$ is not necessarily unique. We denote the \emph{set of min-cuts for $\I$} by $\mathscr{U}_{\I}$.
%
%
A subsystem indexed by $\I$ will be said to be \emph{generic} (\emph{degenerate}) if $\mathscr{U}_{\I}$ has cardinality equal to (greater than) one. If $\mathscr{U}_{\I}$ has more than one element, every $\mU_{\I}\in \mathscr{U}_{\I}$ is referred to as a \emph{degenerate min-cut}, in the sense that any such cut achieves the minimum weight among all possible cuts for $\I$.\footnote{\, Notice that the cardinality of $\mathscr{U}_{\I}$ is bounded by $1\leq \abs{\mathscr{U}_{\I}} \leq 2^{\abs{V}-(\N+1)}$, since only the $\abs{V}-(\N+1)$ bulk vertices can be optionally included in the $\I$-cut.
} 
Similarly, we introduce the following terminology for graph models

\begin{ndefi}[Generic and degenerate graph models]\label{def:generic}
A graph model is \emph{generic} if every subsystem is \emph{generic}, while it is \emph{degenerate} if at least one subsystem is \emph{degenerate}. 
\end{ndefi}

Notice that a ``general'' graph model may or may not be ``generic'' according to this definition. The motivation for this choice of terminology is that random graph models typically do not have any degeneracy. 

For an arbitrary graph model $\gm{\N}$ we define its \textit{min-cut structure} as follows

\begin{ndefi}[Min-cut structure of a graph model]
\label{def:mincutstructure}
The \emph{min-cut structure} $\mathfrak{m}(\gm{\N})$ of a graph model $\gm{\N}$ is the collection of its min-cut sets for all polychromatic indices, i.e.
\begin{equation}
    \mathfrak{m}(\gm{\N}) \coloneqq \{\mathscr{U}_{\I}\; \text{\emph{for all}}\; \I \}
\end{equation}
\end{ndefi}

Any two graph models on the same \textit{topological} graph model $\tgm{\N}$ will be considered equivalent if their min-cut structures are equal. This equivalence relation allows us to organize all such graph models into a finite set of equivalence classes, each one corresponding to a distinct min-cut structure $\mathfrak{m}$. We write such an equivalence class as a pair $(\tgm{\N},\mathfrak{m})$, stressing the dependence on both the min-cut structure and the underlying topological graph model. Occasionally, we will also write $(\tgm{\N},\mathfrak{m})[\gm{\N}]$ to stress that $(\tgm{\N},\mathfrak{m})$ 
is not an arbitrary equivalence class, but the class specified by the representative graph model $\gm{\N}$. Any element of a class is specified by a choice of weights consistent with the min-cut structure $\mathfrak{m}$, and it will be convenient to associate to each equivalence class a region in the space of edge weights $\mathbb{R}^{\E}_{>0}$. 

\begin{ndefi}[W-cell]
\label{def:wcell}
    The \emph{W-cell} $\mathcal{W}(\tgm{\N},\mathfrak{m})\subset\mathbb{R}^{\E}$ of the min-cut structure $\mathfrak{m}$ on the topological graph model $\tgm{\N}$ is the set of weight vectors of all graph models $\gm{\N}$ in the class $(\tgm{\N},\mathfrak{m})$.
\end{ndefi}

\noindent In what follows, when the specification a topological graph model $\tgm{\N}$ is clear from context, and we only need to keep track of the dependence of a W-cell on a min-cut structure $\mathfrak{m}$, we will simply write $\mathcal{W}_{\mathfrak{m}}$ instead of $\mathcal{W}(\tgm{\N},\mathfrak{m})$.

Given a topological graph model $\tgm{\N}$, one may want to specify a min-cut structure on it more abstractly, by directly listing the min-cut sets $\mathscr{U}_{\I}$ without recurring to graph models and explicit choices of edge weights. Notice however that an arbitrary choice of cuts for each index $\I$ is not necessarily a meaningful min-cut structure, since there may not exist a choice of weights that make all these cuts (and no others) minimal
(for example by violating some of the lemmas in \cref{ssec:graph_lemmas}). We will comment more explicitly on this point in the next section.

Each equivalence class will also be associated to a specific region of entropy space, called the \textit{S-cell} of $(\tgm{\N},\mathfrak{m})$, defined as follows
\begin{ndefi}[S-cell]
\label{def:scell}
    The \emph{S-cell} $\mathcal{S}(\tgm{\N},\mathfrak{m})\subset\mathbb{R}^{\D}$ of the min-cut structure $\mathfrak{m}$ on the topological graph model $\tgm{\N}$ is the set of entropy vectors of all graph models $\gm{\N}$ in the class $(\tgm{\N},\mathfrak{m})$.
\end{ndefi}

By definition, each S-cell is contained in the HEC, since it is a set of entropy vectors that can be realized by graph models.  As we will see, different S-cells can have different dimensions, and clearly the union of all S-cells that can be obtained from all possible topological graph models and min-cut structures is the whole HEC$_{\N}$.  However, the set of all S-cells does not form a partition of the HEC$_{\N}$ since a given entropy vector can in general belong to multiple S-cells.\footnote{\, This can easily happen for example with disconnected graph models. We will comment on this type of situation in more detail in \cref{subsec:properties}.} 
Furthermore, the relation between equivalence classes and S-cells is not a bijection, since the same S-cell can be associated to distinct equivalence classes.\footnote{\, For example one can start from a class $(\tgm{\N},\mathfrak{m})$ and add new vertices to $\tgm{\N}$ to obtain a new topological graph model $\tgm{\N}'$ such that the min-cut structure $\mathfrak{m}$ and the sets of cut edges $\cut^*_{\J}$ for all subsystems remain unchanged (a trivial way to do so is by adding a disconnected component of only bulk vertices).}

Having introduced the notion of an S-cell, we now consider a linear subspace of entropy space which is naturally associated to it. Specifically, we define the \textit{min-cut subspace} of a class $(\tgm{\N},\mathfrak{m})$ as follows
\begin{ndefi}[Min-cut subspace]
\label{def:vspace}
    The \emph{min-cut subspace} $\mathbb{S}(\tgm{\N},\mathfrak{m})$ of the min-cut structure $\mathfrak{m}$ on the topological graph model $\tgm{\N}$ is the minimal linear subspace which contains the \emph{S-cell} $\mathcal{S}(\tgm{\N},\mathfrak{m})$, i.e.
    \begin{equation}
        \mathbb{S}(\tgm{\N},\mathfrak{m})\coloneqq\text{\emph{Span}}(\mathcal{S}(\tgm{\N},\mathfrak{m}))
    \end{equation}
\end{ndefi}
\noindent Notice that even the relation between S-cells and min-cut subspaces is not a bijection, since different S-cells can give the same min-cut subspace.\footnote{\, This situation can also easily be realized with disconnected graphs, see \cref{subsec:properties}.} 

Having introduced the main definitions, we will now proceed to explain 
how, given a topological graph model and a min-cut structure, one can compute the S-cell and the min-cut subspace explicitly. We will start from the slightly simpler case of generic graphs and then extend the analysis to situations where some min-cuts are degenerate.

\subsection{Min-cut structures without degeneracy}
\label{subsec:gmnodeg}

Given a topological graph model $\tgm{\N}$, and a \textit{generic} min-cut structure $\mathfrak{m}$, we can determine the W-cell of $(\tgm{\N},\mathfrak{m})$ as follows. To each edge of $\tgm{\N}$ we associate a weight \textit{variable} $w_e$ (rather than a weight value as in a graph model $\gm{\N}$). For any polychromatic index $\I$, we can consider all cuts $U_{\I}$, and the min-cut structure $\mathfrak{m}$ specifies which one is the min-cut $\mU_{\I}$. The minimality of $\mU_{\I}$ translates into a set of linear inequalities in the weight variables, which take the form 
\begin{equation}
\label{eq:open_constaints}
\norm{\cut(\mU_{\I})}<\norm{\cut(U_{\I})}\qquad \forall\, U_{\I}\neq\mU_{\I}.
\end{equation}
The collection of all these inequalities in \cref{eq:open_constaints}, for all polychromatic indices $\I$, combined with the condition that each edge weight must be non negative
\begin{equation}
\label{eq:weight-nonnegative-generic}
    w_e>0\qquad \forall\, e\in E
\end{equation}
specifies a polyhedral cone in $\mathbb{R}^{\E}_{>0}$. The cone structure follows from the inequalities being homogeneous, and since there are at least $\E$ linearly independent ones (those in \cref{eq:weight-nonnegative-generic}), the cone is pointed. Furthermore, since all the inequalities are strict, the (non-empty) solution corresponds to the interior of this cone, and the cone is full-dimensional.\footnote{\, The linear span of the vectors inside the cone is the full space.} This is the W-cell $\mathcal{W}(\tgm{\N},\mathfrak{m})$.

As we mentioned earlier, an arbitrary choice of min-cut for each subsystem $\I$ does not necessarily correspond to a meaningful min-cut structure. This is because the system of inequalities described above could have no solutions, in which case the W-cell is just the empty set. In the following we will always ignore these situations, and whenever we consider a min-cut structure on a certain topological graph model, we will implicitly assume that its W-cell is non-trivial.

Having showed how to determine the W-cell $\mathcal{W}(\tgm{\N},\mathfrak{m})$ explicitly, we will now explain how to derive the corresponding S-cell $\mathcal{S}(\tgm{\N},\mathfrak{m})$ in entropy space. Notice that for any graph model $\gm{\N}$, each entropy $S_{\I}$ is computed by a sum of weights, each one with unit coefficient, cf. \cref{eq:cost}. For each subsystem $\I$ we introduce an incidence vector $\mathbf{\Gamma}_{\I}\in\{0,1\}^{\E}$,
\begin{equation}
\label{eq:incidence}
    \Gamma_{\I}^e \coloneqq \begin{cases}
1\qquad \text{if}\; e\in\cut(\mU_{\I})\\
0\qquad \text{otherwise}
\end{cases}
\end{equation}
which specifies which edges participate in the cut. We can then write the entropy as 
\begin{equation}
\label{eq:gamma_def}
    S_{\I}=\Gamma_{\I}^{e}\,w_e
\end{equation}
where the $\D\!\times\!\E$ matrix $\Gamma_{\I}^{e}$ represents a linear map $\Gamma: \mathbb{R}^{\E}\to \mathbb{R}^{\D}$. The map which associates entropy vectors to weight vectors of fixed topological graph and min-cut structure is the \textit{restriction} $\Gamma|_{\mathcal{W}}$ of $\Gamma$ to the W-cell $\mathcal{W}(\tgm{\N},\mathfrak{m})$. The S-cell is the image of this restricted map\footnote{\, Equivalently, one could also define the S-cell as the image of $\mathcal{W}$ under the \textit{unrestricted} map $\Gamma$, i.e., $\mathcal{S}=\Gamma(\mathcal{W})$, but we preferred the option in the main text because the matrix $\Gamma$ is specified by the min-cut structure, which is only defined within $\mathcal{W}$.}
\begin{equation}
    \mathcal{S}(\tgm{\N},\mathfrak{m})=\text{Im}\;\Gamma|_{\mathcal{W}}
\end{equation}
and the min-cut subspace is
\begin{equation}
\label{eq:genericV}
    \mathbb{S}(\tgm{\N},\mathfrak{m})=\text{Span}(\text{Im}\;\Gamma|_{\mathcal{W}})=\text{Im}\;\Gamma
\end{equation}
where the last equality follows from the fact that, as explained above,
$\text{Span}(\mathcal{W})=\mathbb{R}^{\E}$.

In conclusion, in the case of a generic min-cut structure $\mathfrak{m}$ on a topological graph model $\tgm{\N}$, the min-cut subspace is simply the column space of the matrix $\Gamma$.

\subsection{Min-cut structures with possible degeneracy}
\label{sssec:gmdeg}

We will now generalize the previous construction to the case where some min-cuts could be degenerate. We start again by determining the W-cell of a given min-cut structure $\mathfrak{m}$ and topological graph model $\tgm{\N}$. 

For each degenerate subsystem $\I$, we label the min-cuts $\mU_{\I}$ in the set $\mathscr{U}_{\I}$ with an upper index $\alpha\in\{1,2,\ldots,\abs{\mathscr{U}_{\I}}\}$. The W-cell is now specified by the following set of inequalities
\begin{equation}
\label{eq:open_contraints2}
    \norm{\cut({\mU_{\I}}^{\alpha})} <\norm{\cut(U_{\I})}\qquad \forall\,U_{\I}\neq{\mU_{\I}}^{\alpha},\;\; \forall\,\alpha
\end{equation}
for each polychromatic index,\footnote{\, For each generic subsystem, there is just a single index $\alpha\in\{1\}$.} together with the strict positivity of the weights (\ref{eq:weight-nonnegative-generic}), and crucially a set of  \textit{degeneracy equations}
\begin{equation}
\label{eq:degeneracy_eq}
    \norm{\cut({\mU_{\I}}^{\alpha})}=\norm{\smash{\cut({\mU_{\I}}^{\beta})}} \qquad \forall\,\I,\;\; \forall\,\alpha, \beta
\end{equation}

Like in the generic case, the W-cell is the interior of a polyhedral cone, but now with a crucial difference. Since all inequalities are strict, denoting by $\mathbb{W}$ the proper linear subspace of $\mathbb{R}^{\E}$ which corresponds to the solution of the degeneracy equations (\ref{eq:degeneracy_eq}), we have 
\begin{equation}
     \text{Span}(\mathcal{W}(\tgm{\N},\mathfrak{m}))=\mathbb{W}
\end{equation}
rather than the full space.

To find the S-cell we can proceed similarly to generic case. For each degenerate subsystem $\I$, we first choose
a ``representative'' min-cut ${\mU_{\I}}^{\alpha}$, by a specific choice of $\alpha$. We denote a choice of min-cuts for all $\I$ schematically by $\{\alpha\}$.
With this choice we can then construct an incidence vector $\mathbf{\Gamma}_{\I}^{\alpha}$ for each polychromatic index $\I$, and as before a linear map $\Gamma^{\{\alpha\}}: \mathbb{R}^{\E}\to \mathbb{R}^{\D}$. Different choices of representatives will in general give different linear maps, however once restricted to the W-cell, all these maps have the same image, the S-cell of $\mathfrak{m}$, i.e., 
\begin{equation}
\label{eq:repindependence}
    \mathcal{S}(\tgm{\N},\mathfrak{m})= \text{Im}\;\Gamma^{\{\alpha\}}|_{\mathcal{W}}\qquad \forall\,\{\alpha\}
\end{equation}

To see this, consider an arbitrary weight vector $\mathbf{w}\in\mathcal{W}$ and two maps $\Gamma^{\{\alpha\}}$ and $\Gamma^{\{\alpha'\}}$ for two different choices of representative min-cuts. We want to show that for any such choices
\begin{equation}
    \Gamma^{\{\alpha\}}(\mathbf{w})=\Gamma^{\{\alpha'\}}(\mathbf{w})
\end{equation}
which we can rewrite as
\begin{equation}
    (\Gamma^{\{\alpha\}}-\Gamma^{\{\alpha'\}})(\mathbf{w})=0
\end{equation}
For a row $\I$ of the matrix that appears in the equation above we should then have
\begin{equation}
\label{eq:rowcondition}
    (\mathbf{\Gamma}_{\I}^{\alpha}-\mathbf{\Gamma}_{\I}^{\alpha'})(\mathbf{w})=0
\end{equation}
If the choice of representatives for the subsystem $\I$ in $\{\alpha\}$ and $\{\alpha'\}$ was the same, \cref{eq:rowcondition} is trivial. Otherwise it is precisely  one of the degeneracy equations that specify $\mathcal{W}$, and it is therefore also satisfied by $\mathbf{w}$.

As for the generic case, the min-cut subspace of a min-cut structure with degeneracies is again defined as the span of the S-cell. Unlike the generic case however, the min-cut subspace is now generally not equal to the image of any of the unrestricted maps $\Gamma^{\{\alpha\}}$.\footnote{\, It is tempting to guess that the min-cut subspace of a degenerate min-cut structure is the intersection of all the images of the unrestricted maps for all choices of representatives, i.e., that $\mathbb{S}=\bigcap_{\{\alpha\}}\text{Im}\;\Gamma^{\{\alpha\}}$. This however is not entirely clear and we leave this question for future work.} Hence
\begin{equation}
\label{eq:restV}
    \mathbb{S}(\tgm{\N},\mathfrak{m})=\text{Span}(\text{Im}\;\Gamma^{\{\alpha\}}|_{\mathcal{W}})\subseteq\text{Span}(\text{Im}\;\Gamma^{\{\alpha\}})
\end{equation}
and typically $\mathbb{S}(\tgm{\N},\mathfrak{m}) \subset \text{Span}(\text{Im}\;\Gamma^{\{\alpha\}})$.

Using this construction it is straightforward to prove the following result, which provides a more direct way to determine the min-cut subspace of an arbitrary class $(\tgm{\N},\mathfrak{m})$ when one is not interested in the W-cell or the S-cell.

\begin{nlemma}
\label{lem:VfromW}
For any class $(\tgm{\N},\mathfrak{m})$ and choice of $\Gamma^{\{\alpha\}}$, the min-cut subspace of $(\tgm{\N},\mathfrak{m})$ is the image of $\mathbb{W}$ under $\Gamma^{\{\alpha\}}$
\begin{equation}
\label{eq:VfromW}
    \mathbb{S}(\tgm{\N},\mathfrak{m})=\Gamma^{\{\alpha\}}(\mathbb{W})
\end{equation}
\end{nlemma}
\begin{proof}
    Since the min-cut subspace is defined as the linear span of the S-cell, $\mathcal{S}$ contains 
    some basis
    $\mathcal{B}_{\mathbb{S}}$ of $\mathbb{S}$. And by \cref{eq:restV}, there exists a collection of vectors in $\mathcal{W}$ which are mapped to $\mathcal{B}_{\mathbb{S}}$, therefore $\Gamma^{\{\alpha\}}(\mathbb{W})\supseteq\mathbb{S}$ for any $\{\alpha\}$. On the other hand, the W-cell is full dimensional in $\mathbb{W}$, i.e., its linear span is $\mathbb{W}$. Therefore, any $\mathbf{w}\in\mathbb{W}$ can be written as a linear combination of a basis $\mathcal{B}_{\mathbb{W}}\subseteq\mathcal{W}$, which is mapped inside $\mathcal{S}$ by $\Gamma^{\{\alpha\}}$, implying $\Gamma^{\{\alpha\}}(\mathbb{W})\subseteq\mathbb{S}$. Combining the two inclusions, $\Gamma^{\{\alpha\}}(\mathbb{W})=\mathbb{S}$.
\end{proof}

\noindent Notice that in the generic case, where $\mathbb{W}=\mathbb{R}^{\E}$, \cref{eq:VfromW} reduces to \cref{eq:genericV}. Because of \cref{eq:repindependence}, the specific choice of min-cut representatives for a min-cut structure is often immaterial, and in what follows we will often drop the explicit dependence of $\Gamma^{\{\alpha\}}$ on such a choice, writing simply $\Gamma$.

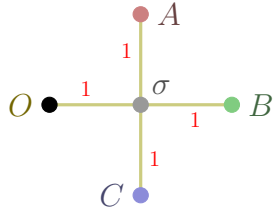
\begin{figure}[tb]
\centering
\begin{tikzpicture}
	\tikzmath{
		\edgelabelsize=0.7;
		\v =0.1;  
		\e=1.2;  
	}
	\coordinate (bv) at (0,0);  
	\coordinate (Av) at (90:\e);  
	\coordinate (Bv) at (0:\e);  
	\coordinate (Cv) at (270:\e);  
	\coordinate (Dv) at (180:\e);  
	\draw[edgestyle] (Av) -- node[scale=\edgelabelsize,edgeweightcolor,left,pos=0.4]{$1$} (bv);
	\draw[edgestyle] (Bv) -- node[scale=\edgelabelsize,edgeweightcolor,below,pos=0.4]{$1$} (bv);
	\draw[edgestyle] (Cv) -- node[scale=\edgelabelsize,edgeweightcolor,right,pos=0.4]{$1$}  (bv);
	\draw[edgestyle] (Dv) -- node[scale=\edgelabelsize,edgeweightcolor,above,pos=0.4]{$1$} (bv);
	\filldraw [color=bvcolor] (bv) circle (\v) node[bvcolor!50!black,above right=0pt]{$\sigma$} ;
	\filldraw [color=Acolor] (Av) circle (\v) node[Acolor!50!black,right=2pt]{$A$} ;
	\filldraw [color=Bcolor] (Bv) circle (\v) node[Bcolor!50!black,right=2pt]{$B$} ;
	\filldraw [color=Ccolor] (Cv) circle (\v) node[Ccolor!50!black,left=2pt]{$C$} ;
	\filldraw [color=Ocolor] (Dv) circle (\v) node[Dcolor!50!black,left=2pt]{$O$} ;
\end{tikzpicture} 
    \caption{The graph model $\gm{3}$ which generates the extreme ray of the HEC$_3$ corresponding to the $4$-party perfect state. The boundary vertices are $\partial V=\{A,B,C,O\}$, and there is a single bulk vertex, $\sigma$. Our drawing conventions for this and all subsequent figures will be to use colors for the boundary vertices (black for the purifier $O$), and to display all bulk vertices in gray.}
    \label{fig:exampleone}
\end{figure}

We conclude this subsection with an explicit example of these constructs. Consider the graph model $\gm{3}$ in \cref{fig:exampleone}. The indicated choice of weights on the underlying topological graph model $\tgm{3}$ specifies the following min-cut structure 
\begin{equation}
\label{eq:wcellexample}
\begin{aligned}
    \mathscr{U}_{A} & = \{ \{A\} \}\quad & \mathscr{U}_{AB} &= \{ \{A,B\}, \{A,B,\sigma\} \}\quad & \mathscr{U}_{ABC} &= \{ \{A,B,C\} \}\\
   \mathscr{U}_{B} & = \{ \{B\} \}  & \mathscr{U}_{AC} &= \{ \{A,C\}, \{A,C,\sigma\} \} &  \\
   \mathscr{U}_{C} & = \{ \{C\} \} & \mathscr{U}_{BC} &= \{ \{B,C\}, \{B,C,\sigma\} \} & 
\end{aligned}
\end{equation}
The degenerate subsystems are $AB$, $AC$ and $BC$, and they give rise to the following degeneracy equations
\begin{equation}
\label{eq:degeqexample}
\begin{aligned}
    w_{_{A\sigma}}+w_{_{B\sigma}} &= w_{_{C\sigma}}+w_{_{O\sigma}} \\
    w_{_{A\sigma}}+w_{_{C\sigma}} &= w_{_{B\sigma}}+w_{_{O\sigma}} \\
    w_{_{B\sigma}}+w_{_{C\sigma}} &= w_{_{A\sigma}}+w_{_{O\sigma}}
\end{aligned}
\end{equation}
The solution of the degeneracy equations is the $1$-dimensional subspace $\mathbb{W}\subseteq \mathbb{R}^{4}$ generated by the vector $(1,1,1,1)$. The W-cell is the ray $w(1,1,1,1)$, with $w>0$, and its image under the map (the rows are labeled by the entropies in the order $(A,B,C,AB,AC,BC,ABC)$, and the columns by the weights in the order $(A\sigma,B\sigma,C\sigma,O\sigma)$)
\begin{equation}
\label{eq:gammaperfect}
    \Gamma=\left(
    \begin{array}{cccc}
    1 & 0 & 0 & 0 \\
    0 & 1 & 0 & 0 \\
    0 & 0 & 1 & 0 \\
    1 & 1 & 0 & 0 \\
    1 & 0 & 1 & 0 \\
    0 & 1 & 1 & 0 \\
    0 & 0 & 0 & 1 \\
    \end{array}
    \right)
\end{equation}
fixed by the choice of representatives
\begin{align}
    & \mU_{\!AB} = \{A,B\}\nonumber\\
    & \mU_{\!AC} = \{A,C\}\nonumber\\
    & \mU_{\!BC} = \{B,C\}
\end{align}
is the S-cell
\begin{equation}
\label{eq:perfectray}
    \mathcal{S} = \lambda \, (1,1,1,2,2,2,1),\qquad \lambda>0
\end{equation}
which is the 4-party ``perfect state''\footnote{\, This is the pure state on four parties which is absolutely maximally entangled, i.e., the state such that each subsystem has maximal entropy.} extreme ray of the HEC$_3$. In this example the min-cut subspace is therefore $1$-dimensional. We will see momentarily that this is always the case for graph models realizing the extreme rays of the HEC.

\subsection{Vanishing weights}
\label{ssec:vanishing}

Up to this point we have only considered situations where all the weights in a graph model are strictly positive. However, one may wonder if there could be any issue when allowing some of the weights to vanish, and in this subsection we explore this situation carefully. The upshot is that if some of the weights vanish in a graph model, all our previous definitions and constructions should be applied to a new graph where all edges with vanishing weight have been deleted. 

Consider a topological graph model $\tgm{\N}$, and suppose that instead of specifying a min-cut structure by a choice of edge weights, we instead try to specify it by a list of min-cut sets for all polychromatic subsystems, which we denote by $\mcs{\mathfrak{m}}$. Any such choice will specify a region $\mcs{\mathcal{W}}$ of the space of edge weights $\mathbb{R}^{\E}$ via a set of min-cut inequalities 
(\ref{eq:open_contraints2}) and degeneracy equations (\ref{eq:degeneracy_eq}), together with the condition of strict positivity of the edge weights (\ref{eq:weight-nonnegative-generic}). However, as already mentioned before, this region can be empty, in which case $\mcs{\mathfrak{m}}$ would not correspond to a valid min-cut structure on $\tgm{\N}$. For a random choice of min-cut sets, this would typically be the case, since the min-cut sets would easily violate some of the basic properties of min-cuts reviewed in \cref{sec:basics}. However, it can also happen that the only reason why $\mcs{\mathcal{W}}$ is empty is that some of the degeneracy equations are forcing some of the edge weights to vanish.

Consider now a bipartition $(E_0,E_{>})$ of the edge set $E$ and the following constraints on the weights
\begin{align}
\label{eq:0weights}
   & w_e=0\qquad \forall\, e\in E_0 \nonumber\\
   & w_e>0\qquad \forall\, e\in E_{>}
\end{align}
Suppose that the new region $\mathcal{W}$ specified by (\ref{eq:0weights}), and precisely the same min-cut inequalities and degeneracy equations that participated in the specification of $\mcs{\mathcal{W}}$, is non-empty. According to our definition, even if $\mathcal{W}$ is non-empty, it is not a W-cell for any min-cut structure on $\tgm{\N}$, since some weights vanish. However, $\mathcal{W}$ can also be specified by an equivalent set of inequalities and equations that can be obtained by simply canceling the terms $w_e$ for all $e\in E_0$ from all the original min-cut inequalities and degeneracy equations that defined $\mcs{\mathcal{W}}$, and by also removing the corresponding equations from \cref{eq:0weights}. 

Consider now a new topological graph model $\mathring{\tgm{\N}}$ obtained by simply deleting all edges in $E_0$ from $\tgm{\N}$. The space of weights for $\mathring{\tgm{\N}}$ is now $\mathbb{R}^{\E-|E_0|}_{>0}$ and \cref{eq:0weights} imposes that all remaining weights are strictly positive. Furthermore for each min-cut ${\mU_{\I}}^{\alpha}$ in $\mcs{\mathfrak{m}}$ we have
\begin{equation}
    \norm{\cut ({\mU_{\I}}^{\alpha})}\,-\!\!\!\sum_{e\,\in\, E_0(\I,\alpha)} \!\!\!w_e\, =\, \norm{\smash{\mathring{\cut} ({\mU_{\I}}^{\alpha})}}
\end{equation}
where 
\begin{equation}
   E_0(\I,\alpha)=\cut ({\mU_{\I}}^{\alpha})\cap E_0 
\end{equation}
and $\smash{\mathring{\cut} ({\mU_{\I}}^{\alpha})}$ is the set of cut-edges for the min-cut ${\mU_{\I}}^{\alpha}$ on the new graph $\mathring{\tgm{\N}}$. The equations and inequalities obtained by cancelling the vanishing weights can then be reinterpreted as a set of min-cut inequalities and degeneracy equations for a min-cut structure  $\mathring{\mathfrak{m}}=\mcs{\mathfrak{m}}$ on the new graph, and the region $\mathcal{W}$ is therefore the W-cell $\mathcal{W}(\mathring{\tgm{\N}},\mathring{\mathfrak{m}})$. We will see an example of this reduction procedure in the next section.

Finally, let us briefly comment on the extreme situation where $\mcs{\mathcal{W}}$ is the $0$-dimensional region that only contains the origin of $\mathbb{R}^{\E}$. In this case the reduction we just described simply produces a new topological graph model $\mathring{\tgm{\N}}$ with no edges, whose space of weights is now $\mathbb{R}^0_{>0}=\{\mathbf{0}\}$. In this space there is only a single W-cell $\mathcal{W}$, which is simply the entire space. The corresponding min-cut structure is generic only in the case where all the vertices in $\mathring{\tgm{\N}}$ are boundary vertices. Otherwise, the min-cut structure is ``maximally degenerate'' in the sense that any subset of the bulk vertices can be included in the min-cut of each polychromatic subsystem.

\subsection{Main properties of W-cells, S-cells and min-cut subspaces}
\label{subsec:properties}

In this subsection we comment on a few important properties of W-cells, S-cells and min-cut subspaces which will be used in later proofs. We begin with a simple observation about the set of W-cells associated to a given topological graph model:

\begin{nlemma}
\label{lem:Wpartition}
For any topological graph model $\tgm{\N}$, the set of \emph{W}-cells associated to all the min-cut structures that can be specified on $\tgm{\N}$ forms a partition of the space of weights $\mathbb{R}^{\E}_{>0}$.
\end{nlemma}
\begin{proof}
    Given a topological graph model $\tgm{\N}$, any choice of weight vector $\mathbf{w}\in\mathbb{R}^{\E}_{>0}$ specifies a min-cut structure uniquely.
\end{proof}

Recall that any W-cell is the interior of a pointed polyhedral cone whose linear span is the subspace $\mathbb{W}$ determined by the degeneracy equations. We now want to consider the boundary of such a cone, and in particular its extreme rays. For a given W-cell $\mathcal{W}$ we denote its closure
by $\overline{\mathcal{W}}$. A $d$-dimensional \textit{face} $\overline{\mathcal{F}}_d$  of $\overline{\mathcal{W}}$ is defined as the intersection of $\overline{\mathcal{W}}$ with a hyperplane $\mathbb{H}$ such that $\overline{\mathcal{W}}$ 
is contained entirely in one of the half-spaces specified by $\mathbb{H}$ (including $\mathbb{H}$ itself). 
According to this definition, each face is again a closed polyhedral cone, and we denote its interior\footnote{\, We define the interior with respect to the subspace topology.} by $\mathcal{F}_d$. Notice in particular that the $1$-dimensional faces are also closed; they are \textit{closed} extreme rays, since they contain the origin. 

In general it is not clear if the interior of a face is by itself a W-cell for some min-cut structure,\footnote{\, Indeed, the structure of the set of W-cells for a given topological graph model is an interesting object to study. We leave this problem for future work. \label{ftn:Wstructures}} but we will prove that this is the case for certain extreme rays. We denote a closed extreme ray of $\overline{\mathcal{W}}$ by $\overline{\mathbf{x}}_{\mathcal{W}}$ and its interior (the corresponding \textit{open} ray) by $\mathbf{x}_{\mathcal{W}}$. Moreover, we will say that a ray (either closed or open) is \textit{nowhere-zero} if all the components of
any vector in its interior are strictly positive. We then have the following lemma:
\begin{nlemma}
\label{lem:Wray}
For any class $(\tgm{\N},\mathfrak{m})$ and a nowhere-zero extreme ray $\overline{\mathbf{x}}_{\mathcal{W}}$ of  $\overline{\mathcal{W}}(\tgm{\N}$,$\mathfrak{m})$, there exists a min-cut structure $\mathfrak{m}'$ on $\tgm{\N}$ whose \emph{W}-cell is $\mathbf{x}_{\mathcal{W}}$.
\end{nlemma}
\begin{proof}
    Consider a class $(\tgm{\N},\mathfrak{m})$ and a nowhere-zero extreme ray $\overline{\mathbf{x}}_{\mathcal{W}}\in\overline{\mathcal{W}}_{\mathfrak{m}}$, where we have introduced the shorthand notation $\overline{\mathcal{W}}_{\mathfrak{m}}=\overline{\mathcal{W}}(\tgm{\N},\mathfrak{m})$. We denote by $\mathbb{F}_1$ the $1$-dimensional
    linear subspace generated by $\overline{\mathbf{x}}_{\mathcal{W}}$. Since $\overline{\mathbf{x}}_{\mathcal{W}}$ is nowhere-zero, by \cref{lem:Wpartition} there exists a unique min-cut structure $\mathfrak{m}'$ whose W-cell $\mathcal{W}_{\mathfrak{m}'}$ contains $\mathbf{x}_{\mathcal{W}}$. Suppose now that $\mathbf{x}_{\mathcal{W}}$ is not a W-cell by itself, i.e., that $\mathcal{W}_{\mathfrak{m}'}\neq\mathbf{x}_{\mathcal{W}}$. Then there exists at least another open ray $\mathbf{w}\in\mathcal{W}_{\mathfrak{m}'}$ which corresponds to the same min-cut structure $\mathfrak{m}'$ of $\mathbf{x}_{\mathcal{W}}$. In general the facets of $\overline{\mathcal{W}}$ are supported by two different types of hyperplanes. Some correspond to degeneracy equations for min-cuts in $\mathfrak{m}$, and others are facets of $\mathbb{R}^{\E}_{\geq 0}$. However by the assumption that $\overline{\mathbf{x}}_{\mathcal{W}}$ is nowhere-zero, it follows that $\mathbb{F}_1$ is completely determined by the degeneracy equations only. Since any ray that is not contained in $\mathbb{F}_1$ would violate at least one of these equations, it must be that $\mathbf{w}=\lambda\,\mathbf{x}_{\mathcal{W}}$ with $\lambda>0$, and $\mathbf{x}_{\mathcal{W}}$ is the whole W-cell of $\mathfrak{m}'$.
\end{proof}

Using this lemma, we can then show that if a class $(\tgm{\N},\mathfrak{m})$ has a $1$-dimensional min-cut subspace, we can always find another class such that the min-cut subspace is preserved and the new W-cell is just a single ray. 
\begin{nlemma}
\label{lem:key}
    For any class $(\tgm{\N},\mathfrak{m})$ such that $\mathbb{S}(\tgm{\N},\mathfrak{m})$ is $1$-dimensional, there exists a class $(\htgm{\N},\hat{\mathfrak{m}})$ such that
    \begin{equation}
        \mathbb{S}(\htgm{\N},\hat{\mathfrak{m}})=\mathbb{S}(\tgm{\N},\mathfrak{m})
    \end{equation}
    and $\mathcal{W}(\htgm{\N},\hat{\mathfrak{m}})$ is a single ray.
\end{nlemma}
\begin{proof}
Consider an arbitrary class $(\tgm{\N},\mathfrak{m})$ such that the min-cut subspace is $1$-dimensional, and suppose its W-cell $\mathcal{W}_{\mathfrak{m}}$ is \textit{not} a single ray. Since $\mathbb{S}(\tgm{\N},\mathfrak{m})$ is $1$-dimensional, the S-cell on the other hand is just a single ray, which we denote by $\mathbf{S}$. For any choice of $\Gamma$ (fixed by a choice of representative min-cuts in $\mathfrak{m}$) we then have
\begin{equation}
    \Gamma\mathbf{w}=\lambda_\mathbf{w}\mathbf{S}\qquad \forall\,\mathbf{w}\in\mathcal{W}_{\mathfrak{m}}
\end{equation}
where $\lambda_\mathbf{w}>0$ is a scaling factor that depends on $\mathbf{w}$. By linearity this implies that
\begin{equation}
\label{eq:Wboundaryimage}
    \Gamma\mathbf{w}'=\lambda_{\mathbf{w}'}\mathbf{S}\qquad \forall\,\mathbf{w}'\in\partial\mathcal{W}_{\mathfrak{m}}
\end{equation}
where $\partial\mathcal{W}_{\mathfrak{m}}=\overline{\mathcal{W}}_{\mathfrak{m}}\setminus\mathcal{W}_{\mathfrak{m}}$ denotes the boundary of $\mathcal{W}_{\mathfrak{m}}$, and $\lambda_{\mathbf{w}'}\geq 0$ is a new rescaling factor which depends on $\mathbf{w}'$. Notice that in general $\lambda_{\mathbf{w}'}$ can now vanish. However, any $\mathbf{w}\in\overline{\mathcal{W}}_{\mathfrak{m}}$ can be written as a conical combination of the extreme rays of $\overline{\mathcal{W}}_{\mathfrak{m}}$. Therefore, since for any $\mathbf{w}$ we have $\lambda_\mathbf{w}>0$, there must exist at least one extreme ray $\mathbf{x}_{\mathcal{W}_{\mathfrak{m}}}\in\partial\mathcal{W}_{\mathfrak{m}}$ such that if we apply \cref{eq:Wboundaryimage}, with $\mathbf{w}'=\mathbf{x}_{\mathcal{W}_{\mathfrak{m}}}$, we have $\lambda_{\mathbf{w}'}>0$.

We now consider one the extreme rays which satisfy this condition, and we have to distinguish two cases, depending on whether $\mathbf{x}_{\mathcal{W}_{\mathfrak{m}}}$ has some vanishing components or not. 

If $\mathbf{x}_{\mathcal{W}_{\mathfrak{m}}}$ has no vanishing components, it follows from \cref{lem:Wray} that it is by itself a W-cell $\mathcal{W}_{\hat{\mathfrak{m}}}$ for some min-cut structure $\hat{\mathfrak{m}}$ on $\tgm{\N}$. Furthermore, by \cref{eq:Wboundaryimage}, $\mathbf{x}_{\mathcal{W}_{\mathfrak{m}}}$ is mapped by $\Gamma$ to the same S-cell as all other weight vectors in $\mathcal{W}_{\mathfrak{m}}$. Therefore $(\htgm{\N},\hat{\mathfrak{m}})$, with $\htgm{\N}=\tgm{\N}$, has the same min-cut subspace as $(\tgm{\N},\mathfrak{m})$.

On the other hand, if $\mathbf{x}_{\mathcal{W}_{\mathfrak{m}}}$ has one or more vanishing components, we can first use the reduction described in \cref{ssec:vanishing} (we simply delete the edges with vanishing weights) to obtain a new graph $\htgm{\N}$ and a new ray $\hat{\mathbf{x}}_{\mathcal{W}_\mathfrak{m}}$ which has no vanishing weights. Since $\hat{\mathbf{x}}_{\mathcal{W}_\mathfrak{m}}$ is determined by all the original degeneracy equations and inequalities that determined $\mathbf{x}_{\mathcal{W}_{\mathfrak{m}}}$, now adapted to the new graph, it is still an extreme ray of the closure of a W-cell $\mathcal{W}_{\hat{\mathfrak{m}}}$ for a min-cut structure $\hat{\mathfrak{m}}$ on $\htgm{\N}$, and therefore by \cref{lem:Wray} a W-cell by itself. Finally, we just need to verify that any new map $\hat{\Gamma}$ defined for $\mathcal{W}_{\hat{\mathfrak{m}}}$ on $\htgm{\N}$ will map $\hat{\mathbf{x}}_{\mathcal{W}_{\hat{\mathfrak{m}}}}$ to the same $\mathbf{S}$. To see this, notice that starting from our initial choice of $\Gamma$ for the class $(\tgm{\N},\mathfrak{m})$, we can obtain a valid choice of $\hat{\Gamma}$ for $(\htgm{\N},\hat{\mathfrak{m}})$ by deleting the columns corresponding to the edges that we have removed from $\tgm{\N}$. The equality
\begin{equation}
    \hat{\Gamma}\; \hat{\mathbf{x}}_{\mathcal{W}_{\hat{\mathfrak{m}}}}=\Gamma\; \mathbf{x}_{\mathcal{W}_\mathfrak{m}}
\end{equation}
then simply follows from the fact that the columns which have been removed from $\Gamma$ are precisely the columns which were multiplied by the vanishing components of $\mathbf{x}_{\mathcal{W}_\mathfrak{m}}$.
\end{proof}

\begin{figure}[tb]
    \centering
    \begin{subfigure}{0.3\textwidth}
    \centering
    \begin{tikzpicture}

	\tikzmath{
	\edgelabelsize=0.7;
		\v =0.1;  
	}
	\coordinate (Av) at (2,0.5);  
	\coordinate (Bv) at (3,2);  
	\coordinate (Ov) at (2,3.5);  
	\coordinate (bv) at (1,2);  

	\draw[edgestyle] (Av) -- node[scale=\edgelabelsize,edgeweightcolor,left,midway]{$1$} (bv);
	\draw[edgestyle] (Ov) -- node[scale=\edgelabelsize,edgeweightcolor,left,midway]{$1$} (bv);
	\draw[edgestyle] (Av) -- node[scale=\edgelabelsize,edgeweightcolor,right,midway]{$1$} (Ov);
	
	\filldraw [color=bvcolor] (bv) circle (\v) node[bvcolor!50!black,left=2pt]{$\sigma$} ;
	\filldraw [color=Acolor] (Av) circle (\v) node[Acolor!50!black,right=2pt]{$A$} ;
	\filldraw [color=Bcolor] (Bv) circle (\v) node[Bcolor!50!black,right=2pt]{$B$} ;
	\filldraw [color=Ocolor] (Ov) circle (\v) node[Ocolor!50!black,right=2pt]{$O$} ;

\end{tikzpicture} 
\subcaption[]{}
\label{fig:Akeylemma}
    \end{subfigure}
    \begin{subfigure}{0.3\textwidth}
    \centering
    \begin{tikzpicture}

	\tikzmath{
	\edgelabelsize=0.7;
		\v =0.1;  
	}
	\coordinate (Av) at (2,0.5);  
	\coordinate (Bv) at (3,2);  
	\coordinate (Ov) at (2,3.5);  
	\coordinate (bv) at (1,2);  

	\draw[edgestyle] (Av) -- node[scale=\edgelabelsize,edgeweightcolor,left,midway]{$1$} (bv);
	\draw[edgestyle] (Ov) -- node[scale=\edgelabelsize,edgeweightcolor,left,midway]{$1$} (bv);

	\filldraw [color=bvcolor] (bv) circle (\v) node[bvcolor!50!black,left=2pt]{$\sigma$} ;
	\filldraw [color=Acolor] (Av) circle (\v) node[Acolor!50!black,right=2pt]{$A$} ;
	\filldraw [color=Bcolor] (Bv) circle (\v) node[Bcolor!50!black,right=2pt]{$B$} ;
	\filldraw [color=Ocolor] (Ov) circle (\v) node[Ocolor!50!black,right=2pt]{$O$} ;

\end{tikzpicture} 
\subcaption[]{}
\label{fig:Bkeylemma}
    \end{subfigure}
        \begin{subfigure}{0.3\textwidth}
    \centering
    \begin{tikzpicture}

	\tikzmath{
	\edgelabelsize=0.7;
		\v =0.1;  
	}
	\coordinate (Av) at (2,0.5);  
	\coordinate (Bv) at (3,2);  
	\coordinate (Ov) at (2,3.5);  
	\coordinate (bv) at (1,2);  

	\draw[edgestyle] (Av) -- node[scale=\edgelabelsize,edgeweightcolor,right,midway]{$1$} (Ov);
	
	\filldraw [color=bvcolor] (bv) circle (\v) node[bvcolor!50!black,left=2pt]{$\sigma$} ;
	\filldraw [color=Acolor] (Av) circle (\v) node[Acolor!50!black,right=2pt]{$A$} ;
	\filldraw [color=Bcolor] (Bv) circle (\v) node[Bcolor!50!black,right=2pt]{$B$} ;
	\filldraw [color=Ocolor] (Ov) circle (\v) node[Ocolor!50!black,right=2pt]{$O$} ;
	
\end{tikzpicture} 
\subcaption[]{}
\label{fig:Ckeylemma}
    \end{subfigure}
    \caption{An example illustrating \cref{lem:key} and the reduction described in \cref{ssec:vanishing} for vanishing weights. The graph (a) is our starting point. The graphs in (b) and (c) are obtained by deleting from (a) the edges that correspond to the vanishing entries of the extreme rays of the closure of the W-cell specified by (a), cf. \cref{eq:2dW1dV}.}
    \label{fig:keylemma}
\end{figure}
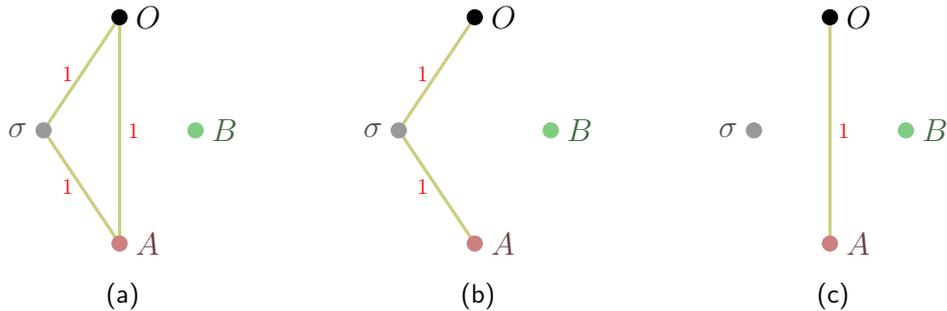

As an example of this lemma and the reduction procedure for vanishing weights presented in \cref{ssec:vanishing}, consider the graph model in \cref{fig:Akeylemma}. The min-cut structure specified by the chosen weights is
\begin{equation}
    \mathscr{U}_{A}=\{\{A\},\{A,\sigma\}\},\qquad \mathscr{U}_{B}=\{\{B\}\},\qquad \mathscr{U}_{AB}=\{\{A,B\},\{A,B,\sigma\}\}
\end{equation}
and there is only one degeneracy equation (it is the same for the indices $A$ and $AB$)
\begin{equation}
    w_{_{AO}}+w_{_{A\sigma}} = w_{_{AO}}+w_{_{O\sigma}}
\end{equation}
The solution to this degeneracy equation is the $2$-dimensional subspace
\begin{equation}
    \mathbb{W}=(0,1,-1)^{\perp}\subset\mathbb{R}^3
\end{equation}
where we ordered the weights according to $(AO,A\sigma,O\sigma)$. The W-cell is the interior of the following polyhedral cone in $\mathbb{W}$ (written as embedded in $\mathbb{R}^3_{>0}$)
\begin{equation}
\label{eq:2dW1dV}
    \text{cone}\,\{(0,1,1),(1,0,0)\}
\end{equation}
Its image under the map
\begin{equation}
\label{eq:2dW1dVmap}
    \Gamma=\left(
    \begin{tabular}{ccc}
        1 & 1 & 0 \\
        0 & 0 & 0 \\
        1 & 0 & 1 \\
    \end{tabular}
    \right)
\end{equation}
fixed by the choice of representatives $\mU_{A}=\{A\}$ and $\mU_{AB}=\{A,B,\sigma\}$, is the S-cell, which is just the single ray (it is the extreme ray of the HEC$_2$ corresponding to a Bell pair for the subsystem $AO$)
\begin{equation}
\label{eq:Scell2dW1dVmap}
    \mathcal{S}=\lambda\, (1,0,1),\qquad \lambda>0
\end{equation}
since both elements of \cref{eq:2dW1dV} are mapped to the same vector by \cref{eq:2dW1dVmap}. 

We then have a situation where an equivalence class has a $1$-dimensional min-cut subspace and a $2$-dimensional W-cell, and we can therefore apply \cref{lem:key}. The extreme rays of the closure of this W-cell are given in \cref{eq:2dW1dV}, and since neither of them is nowhere-zero, we first need to apply the reduction described in \cref{ssec:vanishing}. Deleting from the graph in \cref{fig:Akeylemma} the edges corresponding to the vanishing entries of the vectors in \cref{eq:2dW1dV}, we obtain the graphs shown in \cref{fig:Bkeylemma} and \cref{fig:Ckeylemma}. The non-zero entries of the same vectors specify the weights of the remaining edges. The W-cell of the min-cut structure for the graph in \cref{fig:Ckeylemma} is obviously $1$-dimensional (since $\E=1$), in agreement with \cref{lem:key}. On the other hand, for the graph in \cref{fig:Bkeylemma}, the space of edge weights is $2$-dimensional. However, there is still a degeneracy, since there are two options for the min-cut for $A$, namely $\{A\}$ and $\{A,\sigma\}$. Therefore, also for this graph, the W-cell is $1$-dimensional, again in agreement with \cref{lem:key}. Notice that the generators of these $1$-dimensional W-cells are obtained by deleting the vanishing components of the vectors in \cref{eq:2dW1dV}, and thus are simply $(1,1)$ and $(1)$. Finally, the new maps $\hat{\Gamma}$ for the new min-cut structures on the new graphs in \cref{fig:Bkeylemma} and \cref{fig:Ckeylemma} are respectively
\begin{equation}
    \hat{\Gamma}=\left(
    \begin{tabular}{cc}
        1 & 0 \\
        0 & 0 \\
        0 & 1 \\
    \end{tabular}
    \right) \qquad \text{and}\qquad
    \hat{\Gamma}=\left(
    \begin{tabular}{c}
        1 \\
        0 \\
        1 \\
    \end{tabular}
    \right)
\end{equation}
which are obtained by deleting the columns of $\Gamma$ from \cref{eq:2dW1dVmap} which correspond to the vanishing components of the vectors in \cref{eq:2dW1dV}. It is then immediate to verify that once applied to the aforementioned generators of the new $1$-dimensional W-cells, these maps give the same S-cell (\ref{eq:Scell2dW1dVmap}) as the original graph from \cref{fig:keylemma}, and therefore also the same-min-cut subspace. 

Similarly to W-cells, S-cells also correspond to the interior of a polyhedral cone, and it is again interesting to explore situations where an (open) extreme ray of the closure of an S-cell is an S-cell by itself.\footnote{\, Notice that since S-cells do not form a partition of the HEC, an (open) extreme ray of the closure of an S-cell can in principle be an S-cell by itself, even if it is also in the interior of a different S-cell. This type of questions, like the one mentioned in footnote \ref{ftn:Wstructures}, is related to investigations of substructures of the HEC which are beyond the scope of this work.} We will not answer this question in general, but we will construct a class of examples where this is the case,\footnote{\, One can construct these examples by building disconnected graphs starting from graph realizations of the extreme rays of the HEC (see below).} starting from the prototypical situations where S-cells are $1$-dimensional, i.e., graph models realizing the extreme rays of the HEC. 

Let us begin by first proving a basic fact about a given graph model $\gm{\N}$ and the face\footnote{ \, 
Recall that a face can have any dimension $0\leq d\leq \D$.}
of the HEC$_{\N}$ that contains the entropy vector of $\gm{\N}$:      
\begin{nthm}\label{thm:deg}
    Given a graph model $\gm{\N}$ and a face $\overline{\mathcal{F}}_d$ of the \emph{HEC}$_{\N}$ such that the entropy vector $\mathbf{S}(\gm{\N})$ belongs to $\mathcal{F}_d$, the minimal supporting linear subspace $\mathbb{F}_d$ of $\overline{\mathcal{F}}_d$ contains the min-cut subspace $\mathbb{S}(\tgm{\N},\mathfrak{m})[\gm{\N}]$.
\end{nthm}
\begin{proof}
    Given a graph model $\gm{\N}$, all entropy vectors in the S-cell $\mathcal{S}$ of $(\tgm{\N},\mathfrak{m})[\gm{\N}]$ can be realized by simply varying the weights in $\gm{\N}$, so clearly $\mathcal{S}\subseteq\text{HEC}_{\N}$. Furthermore, $\mathcal{S}$ is an open set, since it is the image under a linear map of an open set in the space of weights (the W-cell). By definition, a face $\overline{\mathcal{F}}_d$ of the HEC$_{\N}$ is $\overline{\mathcal{F}}_d=\mathbb{H}\cap\text{HEC}_{\N}$, where $\mathbb{H}$ is a certain hyperplane in $\mathbb{R}^{\D}$ such that the HEC$_{\N}$ is entirely contained in one of the two half-spaces specified by $\mathbb{H}$ and in $\mathbb{H}$ itself. Since $\mathcal{S}$ is open, and by assumption $\mathcal{F}_d\cap\mathcal{S}$ is non-empty, it must be that $\mathcal{S}\subset \overline{\mathcal{F}}_d$, otherwise there would exist elements of $\mathcal{S}$ on both sides of $\mathbb{H}$, contradicting the fact that $\mathcal{S}\subseteq\text{HEC}_{\N}$ and $\overline{\mathcal{F}}_d$ is a face of the HEC$_{\N}$. Therefore $\mathbb{S}=\text{Span}(\mathcal{S})\subset\text{Span}(\overline{\mathcal{F}}_d)=\mathbb{H}$.
\end{proof}
This immediately leads to the following corollary:
\begin{ncor}
\label{cor:ext}
    A graph model $\gm{\N}$ realizes an extreme ray of the \emph{HEC}$_{\N}$ only if the min-cut subspace $\mathbb{S}(\tgm{\N},\mathfrak{m})[\gm{\N}]$ is $1$-dimensional.
\end{ncor}
\begin{proof}
    This is just Theorem \ref{thm:deg} applied to the specific case where the face $\overline{\mathcal{F}}_1$ is an extreme ray of the HEC$_{\N}$.
\end{proof}

This result will play a central role in the reconstruction of the HEC discussed in \cref{sec:logic}. It can be seen to be equivalent to Theorem $2(b)$ of \cite{Avis:2021xnz}, where it was proven using a different setting wherein $1$-dimensional min-cut subspaces constitute the extreme rays of a polyhedral cone whose facet description is known. Because of this, were the converse of \cref{cor:ext} to hold, it would provide an explicit derivation of the HEC via direct computation of its extreme rays. Unfortunately though, the converse of \cref{cor:ext} is in fact false. Namely, a graph model with a $1$-dimensional min-cut subspace is not guaranteed to yield an extreme ray of the HEC, as also shown by \cite{Avis:2021xnz} through a counterexample in appendix B therein.\footnote{\, See \cite{Avis:2021xnz} for more details on how \cref{cor:ext} can nonetheless be used to construct the HEC.} An example of \cref{cor:ext} is instead the graph model $\gm{3}$ from \cref{sssec:gmdeg} which realizes the $4$-party perfect state extreme ray of the HEC$_3$ (cf. \cref{fig:exampleone}).

\subsection{Disconnected graphs}
\label{subsec:disconnected}

In this subsection we comment on the construction of a topological graph model and min-cut structure $(\tgm{\N},\mathfrak{m})$ via the disjoint union of simpler \textit{building blocks}, and show how the W-cell, S-cell and min-cut subspace of the new graph can be obtained from those of its components. Conversely, the same analysis also clarifies how given an arbitrary class $(\tgm{\N},\mathfrak{m})$ where $\tgm{\N}$ is disconnected, one can decompose these structures according to the connected components of $\tgm{\N}$.

Let us begin with the simple situation where we consider an arbitrary collection of $k$ $\N$-party topological graph models and min-cut structures 
\begin{equation}
    \{(\tgm{\N}^i,\mathfrak{m}^i)\,, \forall i\in[k]\}
\end{equation}
where the boundary vertices of all topological graph models in the collection are colored with the same set of colors $[\N+1]$. We then construct a new topological graph model 
\begin{equation}
    \tgm{\N}=\bigoplus_{i\in[k]}\tgm{\N}^i
\end{equation}
and min-cut structure
\begin{equation}
    \mathfrak{m}=\{\bigcup_{i\in[k]}\mathscr{U}_{\I}^i\;\text{for all}\; \I\}
\end{equation}
The space of weights of the new graph is the direct sum of the weight spaces of the individual graphs 
\begin{equation}
    \mathbb{R}^{\E}=\bigoplus_{i\in[k]}\mathbb{R}^{\E^i}
\end{equation}
and since the weights on any given graph can be varied independently from those of the others, the W-cell is also a \textit{direct sum}\footnote{\, The direct sum of cones is the special case of the Minkowski sum (see below) where the individual cones are contained in orthogonal subspaces.} of W-cells
\begin{equation}
    \mathcal{W}(\tgm{\N},\mathfrak{m})=\bigoplus_{i\in[k]}\mathcal{W}(\tgm{\N}^i,\mathfrak{m}^i)
\end{equation}
The new S-cell on the other hand is not the direct sum of the individual S-cells, because the W-cells are not mapped to orthogonal subspaces of entropy space. However, since the weights can still be varied independently, the S-cell is the \textit{Minkowski sum}\footnote{\, The Minkowski sum of two sets of vectors $X,Y$ is the set of sums $\mathbf{x}+\mathbf{y}$ for all $\mathbf{x}\in X$, $\mathbf{y}\in Y$.} of the individual S-cells 
\begin{equation}
    \mathcal{S}(\tgm{\N},\mathfrak{m})=\bigplus_{i\in[k]}\mathcal{S}(\tgm{\N}^i,\mathfrak{m}^i)
\end{equation}
and the min-cut subspace is simply the sum of the individual subspaces 
\begin{equation}
    \mathbb{S}(\tgm{\N},\mathfrak{m})=\bigplus_{i\in[k]}\mathbb{S}(\tgm{\N}^i,\mathfrak{m}^i)
\end{equation}

In the particular case where the individual graphs realize $1$-dimensional min-cut subspaces, or equivalently where their S-cells are single rays, the S-cell of the composite graph is simply the conical hull of such rays.\footnote{\, As exemplified for \cref{cor:ext}, for each component graph there is only a single free variable, i.e., a global rescaling of the weights. The fact that the rays are open and the rescaling factors have to be strictly positive implies that the resulting S-cell is an open set (as it should be).} A special instance of this construction is when the individual S-cells correspond to extreme rays of the HEC$_{\N}$. For example, by considering the full list of extreme rays one can construct a new graph whose S-cell is the interior of the HEC$_{\N}$. But by taking different collections of extreme rays, one can also use this procedure to construct graphs whose S-cells overlap, clarifying as we anticipated in the previous section that S-cells do not form a partition of the HEC$_{\N}$.\footnote{\, For this last construction the cone has to be non-simplicial, i.e., the number of extreme rays should be strictly larger than its dimension, which happens to be the case for any $\N\geq 4$.} Moreover, the same type of construction can also be used to generate graphs with different S-cells but same min-cut subspace, for example by considering two collections of extreme rays which span the supporting subspace of a given face of the HEC$_{\N}$.

We now want to extend this construction to more general situations, where the boundary vertices of each building block may be colored by a different set of colors. Notice that in the most general scenario the sets of colors of any two building blocks may or may not intersect, even if the boundary vertices of the two topological graph models are labeled by the same number of colors.\footnote{\, For example, we could have a topological graph model $\tgm{2}=\tgm{2}^1\oplus\tgm{2}^2$ where the boundary vertices of $\tgm{2}^1$ are labeled by $\{A,B,O\}$ and those of $\tgm{2}^2$ by $\{A,C,O\}$. Furthermore, if we imagine that the collection of building blocks is the set of connected components of a larger topological graph model that we want to decompose, there is only a single purifier, which does not even need to appear in each component.} To deal with this type of situations, it will be convenient to transform all the building blocks in such a way that their boundary vertices are labeled by precisely the same set of colors. This operation however changes the number of parties, and therefore the dimension of entropy space of a given building block. Because of this, before we can apply the procedure described above, we first need to make sure that we are working in the same space, and we will achieve this by embedding the entropy space of a building block into the entropy space of the entire collection in a precise way.\footnote{\, Of course this step is only necessary in the composition of building blocks, and not when we want to analyse the components of a given, disconnected, topological graph model.} 

At same time, the transformation of a given building block should be designed in such a way that the relevant data is not altered in any way, i.e., the same embedding should map the S-cell and min-cut subspace of the original building block to the new the S-cell and min-cut subspace obtained after the transformation. Since this simple construction will be particularly useful in later sections, we explain it in detail for a single graph. 

Consider a topological graph model $\tgm{\N}$ (which by itself does not have to be connected), and suppose we want to use it as a building block in a collection that comprises $\N'$ colors, with $\N'>\N$. We define:
\begin{ndefi}[Standard lift of a topological graph model]
\label{def:tgmlift}
The standard lift of a topological graph model $\tgm{\N}$ to $\N'$ parties, with $\N'>\N$, is the topological graph model $\tgm{\N'}'$ obtained from $\tgm{\N}$ by first labeling by $\N'+1$ all the boundary vertices in $\tgm{\N}$ which were initially labeled by $\N+1$, and then adjoining $\N'-\N$ disconnected boundary vertices labeled by the new colors in $\{\N+1,\ldots,\N'\}$.
\end{ndefi}

A standard lift of a topological graph model $\tgm{\N}$ to $\N'$ parties is naturally associated to the following embedding of $\mathbb{R}^{\D}$ into $\mathbb{R}^{\D'}$

\begin{ndefi}[Standard embedding]
\label{def:uncorembed}
Given $\N$ and $\N'>\N$, the standard embedding of  $\mathbb{R}^{\D}$ into $\mathbb{R}^{\D'}$ is the embedding specified by the following equations
\begin{equation}
    S_{\I'}=
    \begin{cases}
    S_{\I'\cap[\N]}     & \forall\,\I'\;\;\text{\emph{s.t.}}\;\; \I'\cap[\N]\neq\emptyset\\
    0                    & \text{\emph{otherwise}}
    \end{cases}
\end{equation}
\end{ndefi}

We leave it as an exercise for the reader to verify that given a topological graph model $\tgm{\N}$ and its standard lift to $\N'$ parties, the S-cell of any min-cut structure, is the standard embedding of the S-cell for the equivalent min-cut structure on $\tgm{\N}$. Notice that this also implies the analogous result for min-cut subspaces.

Given an arbitrary collection of building blocks, we can then apply the standard lift procedure to each one of them, to obtain a new (equivalent) collection where all building blocks have precisely the same set of colors. The W-cells, S-cells and min-cut subspaces of the various min-cut structures can then be obtained as described above. Finally, given the results of this subsection, in what follows we will typically focus our attention on connected graphs (most notably in \cref{subsec:simpletreepmi}). It should however always be clear that we can use the construction presented here to generalize any result obtained for connected graph to this more general scenario.

\section{Min-cut subspaces from marginal independence}
\label{sec:cone_marginal}

In the preceding section we have defined a particular subspace of the entropy space: the min-cut subspace of a graph model. This subspace provides a key construct to characterize the structure of the HEC since it naturally encapsulates the facets, as well as the extreme rays, of the cone. 
In the present section, we consider another class of subspaces of the entropy space: the ``patterns of marginal independence'' (PMI). As we will see, these new subspaces typically contain the min-cut subspaces properly, and one may expect them to be too coarse. Instead, it  will turn out that these subspaces are the ones that distill the essential information. 

We begin by briefly reviewing in \cref{subsec:pmireview} the formal definition of ``patterns of marginal independence'' from \cite{Hernandez-Cuenca:2019jpv}. We then define ``holographic PMIs'' for min-cut structures on topological graph models in \cref{subsec:hologpmi}. Finally in \cref{subsec:simpletreepmi} we show that when a topological graph model has the topology of a tree and all boundary vertices have a different color, the min-cut subspace for any min-cut structure coincides with the PMI. A more general analysis of the relation between min-cut subspaces and PMIs, which involves recolorings of boundary vertices, will be carried out in \cref{sec:recolorings}.

\subsection{Review of marginal independence and PMIs}
\label{subsec:pmireview}
 
 Consider an arbitrary $\N$-partite quantum system, described for convenience by a pure state $\ket{\psi}$ on a Hilbert space with $\N+1$ tensor factors, cf. \cref{eq:purification}. Given two non-overlapping subsystems $\underline{\I},\underline{\K}$ we can measure the total amount of correlation between the corresponding marginals $\rho_{\underline{\I}},\rho_{\underline{\K}}$ by evaluating their \textit{mutual information}\footnote{\, By non-overlapping we mean $\underline{\I}\cap\underline{\K}=\varnothing$. Notice that if $\underline{\I}\cup\underline{\K}=[\N+1]$ the mutual information reduces to twice the entropy of either subsystem.}
 \begin{equation}
 \label{eq:i2instance}
     {\sf{I}}(\underline{\I}:\underline{\K})
     \coloneqq
     S_{\underline{\I}}+S_{\underline{\K}}-S_{\underline{\I} \cup \underline{\K}} \ .
 \end{equation}
 If and only if the mutual information vanishes, the two subsystems are \textit{independent} and we have the factorization
 \begin{equation}
     \rho_{\underline{\I}\cup \underline{\K}}=\rho_{\underline{\I}}\otimes\rho_{\underline{\K}} \ .
 \end{equation}
 For any given $\N$-partite state, it is straightforward to determine which pairs of subsystems $(\underline{\I},\underline{\K})$ are independent, since one simply has to compute all possible instances of the mutual information. 
 
 Conversely, one can imagine specifying a certain pattern of independences by listing which pairs of subsystems are independent and which pairs are not (notice that a pattern is ``complete'', in the sense that each instance of the mutual information is demanded to either vanish or not). For a given pattern, one can then ask if there exists a density matrix which realizes it. This problem was first introduced in \cite{Hernandez-Cuenca:2019jpv} and dubbed the \textit{marginal independence problem}.

 One obvious restriction to the set of realizable patterns simply comes from the linear dependences among various instances of the mutual information. Since the total number of instances is greater than the dimension of entropy space, the instances are linearly dependent. Writing a particular instance as a linear combination of other instances 
 \begin{equation}
 \label{eq:linear_comb}
     \sf{I}(\underline{\I}:\underline{\K})=\sum_{\alpha}c_{\alpha}\, \sf{I}(\underline{\I}^{\alpha}:\underline{\K}^{\alpha})
 \end{equation}
 one immediately sees that if 
 \begin{equation}
 \label{eq:partialpattern}
     \sf{I}(\underline{\I}^{\alpha}:\underline{\K}^{\alpha})=0\qquad \forall \alpha
 \end{equation}
 then it must be that $ \sf{I}(\underline{\I}:\underline{\K})=0$. This type of implications restricts the set of meaningful patterns, since a pattern that requires all instances in \cref{eq:partialpattern} to vanish and $\sf{I}(\underline{\I}:\underline{\K})$ to be non-vanishing clearly can never be realized.
 
 A separate kind of restriction comes from physical constraints such as subadditivity. To see how instances of SA constrain the set of possibly realizable patterns, we can rewrite \cref{eq:linear_comb} as 
 \begin{equation}
 \label{eq:sign_linear_comb}
     \sf{I}(\underline{\I}:\underline{\K})=\sum_{\alpha_{+}}c_{\alpha_{+}}\sf{I}(\underline{\I}^{\alpha_{+}}:\underline{\K}^{\alpha_{+}})+\sum_{\alpha_{-}}c_{\alpha_{-}}\sf{I}(\underline{\I}^{\alpha_{-}}:\underline{\K}^{\alpha_{-}}),
 \end{equation}
 where $c_{\alpha_{+}}$ ($c_{\alpha_{-}}$) refer to the positive (negative) coefficients $c_{\alpha}$. 
 Suppose now that we partially specify a pattern by requiring that all $c_{\alpha_{+}}$ terms in \cref{eq:sign_linear_comb} vanish. Since SA implies that for any state each instance of the mutual information is non-negative, the only way to satisfy \cref{eq:sign_linear_comb} without violating SA is by also requiring that all other instances appearing in \cref{eq:sign_linear_comb} vanish. 
 
 Given the above fundamental restrictions,\footnote{ \, There are of course further physical restrictions such as SSA, but it will turn out that considering just SA is particularly useful.}
 \cite{Hernandez-Cuenca:2019jpv} defined the marginal independence problem by only considering patterns of independences which are consistent with linear dependences among the instances of the mutual information, and all the instances of SA. Geometrically, one can think of these consistent patterns as corresponding to the faces of the polyhedral cone in entropy space defined by subadditivity. More precisely, one first defines the \textit{subadditivity cone} (SAC) as follows (we already mentioned this object in \cref{ssec:cones} but we repeat the precise definition here for convenience)
 
 \begin{ndefi}[SAC]\label{def:SAcone}
    The $\N$-party \emph{subadditivity cone (SAC$_\N$)} is the polyhedral cone in $\mathbb{R}^{\D}$ obtained from the intersection of all the half-spaces specified by the inequalities $\sf{I}(\underline{\I}:\underline{\K})\geq 0$ for all pairs $(\underline{\I},\underline{\K})$ of non-intersecting subsystems.
\end{ndefi}

One then defines the \textit{patterns of marginal independence} (PMI) as specific linear subspaces of the entropy space which naturally characterize the instances of vanishing mutual information:

 \begin{ndefi}[Pattern of marginal independence]\label{def:patterns}
    A \emph{pattern of marginal independence (PMI)} is the linear supporting subspace $\mathbb{P}$ of a face of the subadditivity cone.
\end{ndefi}

The reason behind this definition is the following: any PMI is now a geometric object, corresponding to the intersection of a certain set of hyperplanes of the form $\sf{I}(\underline{\I}:\underline{\K}) = 0$, and by construction this set respects the linear dependences among the hyperplanes in the sense described above. Furthermore, since we are only considering subspaces that correspond to faces of the SA cone, any PMI is guaranteed to contain a region of entropy space (the face) such that all entropy vectors in this region respect all instances of SA. All naive patterns which do not respect the linear dependences among the instances of the mutual information, or that do not include any (non-trivial) region of entropy space whose elements respect all instances of SA, are automatically excluded. Notice in particular that all $1$-dimensional PMIs are generated by the extreme rays of the SAC$_{\N}$.

Even if \cref{def:patterns} excludes a large set of meaningless patterns, it does not guarantee that each PMI contains at least on entropy vector that can be realized by a density matrix. We will say that a PMI $\mathbb{P}$ is \textit{realizable} if there exists a density matrix $\rho$ such that the entropy vector $\bf{S}(\rho)$ belongs to $\mathbb{P}$ but not to any lower dimensional subspace $\mathbb{P}'\subset\mathbb{P}$. Notice that the fact that a PMI is realizable does not imply that each vector $\bf{S}\in\mathbb{P}$ (even within the intersection of $\mathbb{P}$ and the SAC) is the entropy vector of some density matrix. For any PMI $\mathbb{P}$ \textit{realized} by a density matrix $\rho$, we will denote by $\pi$ the map which associates $\mathbb{P}$ to $\rho$, $\mathbb{P}=\pi(\rho)$. Furthermore, for any PMI $\mathbb{P}$, we denote by $\Pi(\mathbb{P})$ the matrix such that
\begin{equation}
\label{eq:pmimatrix}
    \text{Ker}\;\Pi^{\intercal}=\mathbb{P}
\end{equation}
In other words, the columns of $\Pi$ are the coefficients of the instances of the vanishing mutual information in $\mathbb{P}$. For example, for $\N=3$, the $1$-dimensional PMI generated by the perfect state extreme ray of \cref{eq:perfectray} corresponds to the matrix
\begin{equation}
    \Pi=\left(
  \begin{array}{cccccc}
  1 & 1 & 0 & 1 & 0 & 0\\
  1 & 0 & 1 & 0 & 1 & 0\\
  0 & 1 & 1 & 0 & 0 & 1\\
  -1 & 0 & 0 & 0 & 0 & -1\\
  0 & -1 & 0 & 0 & -1 & 0\\
  0 & 0 & -1 & -1 & 0 & 0\\
  0 & 0 & 0 & 1 & 1 & 1\\
  \end{array}
  \right)
\end{equation}
Occasionally we will informally say that an instance of the mutual information $\sf{I}(\underline{\I}:\underline{\K})$ is ``in a PMI $\mathbb{P}$'', meaning that $\sf{I}(\underline{\I}:\underline{\K})=0$ is one of the hyperplanes that determine $\mathbb{P}$, or equivalently, that the vector normal to this hyperplane is one of the columns of $\Pi(\mathbb{P})$.

In general, the marginal independence problem asks which PMIs are realizable by a given class of states. In the context of the present work, the states of interest are the geometric states in holographic theories. We will argue in \cref{sec:logic} that the solution to this \textit{holographic} marginal independence problem (HMIP) provides sufficient information to reconstruct the holographic entropy cone. But in order to do this, we first need to clarify what we mean by ``holographic PMI'' and then to establish a connection between PMIs and min-cut subspaces of graph models.

\subsection{Patterns of marginal independence for graph models}
\label{subsec:hologpmi}

Having reviewed the definition of a pattern of marginal independence for arbitrary quantum states, we will now introduce a similar definition for graph models, and explain some of its basic properties.

As for quantum states, any graph model $\gm{\N}$ gives an entropy vector, and it is therefore straightforward to determine the corresponding PMI, defined as follows
\begin{ndefi}[PMI of a graph model]\label{def:gpmi}
    Given a graph model $\gm{\N}$, its \emph{PMI}   $\mathbb{P}=\pi(\gm{\N})$ is the $\N$-party \emph{PMI} of smallest dimension that contains the entropy vector $\mathbf{S}(\gm{\N})$.
\end{ndefi}
However, we are as usual interested in min-cut structures on topological graph models, rather than in specific graph models, but in order to be able to work with PMIs of min-cut structures, we first need to check that this is a well defined concept.
Let us first recall the necessary and sufficient conditions for an instance of the mutual information to vanish in a graph model:\footnote{\, For holographic configurations this is a well know consequence of the HRRT formula.}

\begin{nlemma}
\label{lem:vanishingI2}
Given a graph model $\gm{\N}$ and two subsystems $\underline{\I}$ and $\underline{\K}$ with $\underline{\I}\cap\underline{\K}=\varnothing$, the mutual information $ \sf{I}(\underline{\I}:\underline{\K})$ vanishes if and only if there exist min-cuts $\mU_{\underline{\I}}$, $\mU_{\underline{\K}}$ and $\mU_{\underline{\I}\cup\underline{\K}}$ such that
\begin{equation}
\label{eq:EWdiscon}
    G[\mU_{\underline{\I}\cup\underline{\K}}]=G[\mU_{\underline{\I}}]\oplus G[\mU_{\underline{\K}}]
\end{equation}
\end{nlemma}

This implies that the vanishing of any instance of the mutual information is ``detected'' by the min-cut structure and that one can therefore think of a PMI as being determined directly by the min-cut structure, specifically:

\begin{ncor}
\label{cor:mpmi}
For any topological graph model $\tgm{\N}$ and min-cut structure $\mathfrak{m}$ we have
\begin{equation}
    \pi(\gm{\N})=\pi(\gm{\N}\,')\qquad \forall\; \gm{\N},\gm{\N}\,'\in(\tgm{\N},\mathfrak{m})
\end{equation}
\end{ncor}
\begin{proof}
    Given $(\tgm{\N},\mathfrak{m})$ consider any two graph models $\gm{\N}$ and $\gm{\N}\,'$ in $(\tgm{\N},\mathfrak{m})$ and any instance $\sf{I}(\underline{\I}:\underline{\K})$ which vanishes for
    $\pi(\gm{\N})$. 
    By \cref{lem:vanishingI2} there exist min-cuts $W_{\underline{\I}}$, $W_{\underline{\K}}$ and $W_{\underline{\I}\cup\underline{\K}}$ in  $\mathfrak{m}$ such that \cref{eq:EWdiscon} holds, and by the same lemma $\sf{I}(\underline{\I}:\underline{\K})$ is also one of the instances that vanish for $\pi(\gm{\N}\,')$.
\end{proof}

Having showed that the PMI of a min-cut structure is a well defined concept, from now on we will always work with these objects and denote them by $\pi(\tgm{\N},\mathfrak{m})$
\begin{equation}
    \pi(\tgm{\N},\mathfrak{m})\coloneqq \pi(\gm{\N})\qquad\text{for any}\quad\gm{\N}\in(\tgm{\N},\mathfrak{m})
\end{equation}
Nevertheless, recall that not all PMIs are realizable holographically (or even by arbitrary quantum states). This prompts us to introduce the definition of a holographic PMI as one which can be realized by a graph model $\gm{\N}$. In terms of equivalence classes we define it as follows:
\begin{ndefi}[Holographic PMI]
\label{def:hologpmi}
    A \emph{PMI} $\mathbb{P}$ is \emph{holographic} if there exists a topological graph model $\tgm{\N}$ and a min-cut structure $\mathfrak{m}$ such that
    \begin{equation}
        \mathbb{P}=\pi(\tgm{\N},\mathfrak{m})
    \end{equation}
\end{ndefi}

\begin{figure}[tb]
    \centering
    \begin{subfigure}{0.49\textwidth}
    \centering
    \begin{tikzpicture}

	\tikzmath{
	\edgelabelsize=0.7;
		\v =0.1;  
	}
	\coordinate (Av) at (1,3);  
	\coordinate (Bv) at (4,2);  
	\coordinate (Cv) at (0.5,1);  
	\coordinate (Ov) at (2,4);  
	\coordinate (bv) at (2,2);  

	\draw[edgestyle] (Av) -- node[scale=\edgelabelsize,edgeweightcolor,below,midway]{$1$} (bv);
	\draw[edgestyle] (Bv) -- node[scale=\edgelabelsize,edgeweightcolor,above,pos=0.6]{$1$} (bv);
	\draw[edgestyle] (Cv) -- node[scale=\edgelabelsize,edgeweightcolor,above,midway]{$1$}  (bv);
	\draw[edgestyle] (Av) -- node[scale=\edgelabelsize,edgeweightcolor,above,midway]{$1$} (Bv);
	\draw[edgestyle] (Bv) -- node[scale=\edgelabelsize,edgeweightcolor,below,midway]{$1$} (Cv);
	\draw[edgestyle] (Cv) -- node[scale=\edgelabelsize,edgeweightcolor,left,midway]{$1$}  (Av);
	\draw[edgestyle] (Ov) -- node[scale=\edgelabelsize,edgeweightcolor,right,pos=0.3]{$4$} (bv); 
	
	\filldraw [color=bvcolor] (bv) circle (\v) node[bvcolor!50!black,below]{$\sigma$} ;
	\filldraw [color=Acolor] (Av) circle (\v) node[Acolor!50!black,above=2pt]{$A$} ;
	\filldraw [color=Bcolor] (Bv) circle (\v) node[Bcolor!50!black,right]{$B$} ;
	\filldraw [color=Ccolor] (Cv) circle (\v) node[Ccolor!50!black,left]{$C$} ;
	\filldraw [color=Ocolor] (Ov) circle (\v) node[Ocolor!50!black,above=2pt]{$O$} ;

\end{tikzpicture} 
    \end{subfigure}
    \begin{subfigure}{0.49\textwidth}
    \centering
    \begin{tikzpicture}

	\tikzmath{
	\edgelabelsize=0.7;
		\v =0.1;  
	}
	\coordinate (Av) at (1,3);  
	\coordinate (Bv) at (4,2);  
	\coordinate (Cv) at (0.5,1);  
	\coordinate (Ov) at (2,4);  
	\coordinate (bv) at (2,2);  

	\draw[edgestyle] (Av) -- node[scale=\edgelabelsize,edgeweightcolor,below,midway]{$1$} (bv);
	\draw[edgestyle] (Bv) -- node[scale=\edgelabelsize,edgeweightcolor,above,pos=0.6]{$1$} (bv);
	\draw[edgestyle] (Cv) -- node[scale=\edgelabelsize,edgeweightcolor,above,midway]{$1$}  (bv);
	\draw[edgestyle] (Av) -- node[scale=\edgelabelsize,edgeweightcolor,above,midway]{$1$} (Bv);
	\draw[edgestyle] (Bv) -- node[scale=\edgelabelsize,edgeweightcolor,below,midway]{$1$} (Cv);
	\draw[edgestyle] (Cv) -- node[scale=\edgelabelsize,edgeweightcolor,left,midway]{$1$}  (Av);
	\draw[edgestyle] (Ov) -- node[scale=\edgelabelsize,edgeweightcolor,right,pos=0.3]{$2$} (bv); 
	
	\filldraw [color=bvcolor] (bv) circle (\v) node[bvcolor!50!black,below]{$\sigma$} ;
	\filldraw [color=Acolor] (Av) circle (\v) node[Acolor!50!black,above=2pt]{$A$} ;
	\filldraw [color=Bcolor] (Bv) circle (\v) node[Bcolor!50!black,right]{$B$} ;
	\filldraw [color=Ccolor] (Cv) circle (\v) node[Ccolor!50!black,left]{$C$} ;
	\filldraw [color=Ocolor] (Ov) circle (\v) node[Ocolor!50!black,above=2pt]{$O$} ;

\end{tikzpicture} 
    \end{subfigure}
    \hfill
    \begin{subfigure}{0.49\textwidth}
    \centering
    \small{
    \begin{equation*}
    \Gamma=\left(
        \begin{array}{ccccccc}
        1 & 1 & 0 & 1 & 0 & 0 & 0\\
        1 & 0 & 1 & 0 & 1 & 0 & 0\\
        0 & 1 & 1 & 0 & 0 & 1 & 0\\
        0 & 1 & 1 & 1 & 1 & 0 & 0\\
        1 & 0 & 1 & 1 & 0 & 1 & 0\\
        1 & 1 & 0 & 0 & 1 & 1 & 0\\
        0 & 0 & 0 & 1 & 1 & 1 & 0\\
        \end{array}
    \right)
    \end{equation*}
    }
    \end{subfigure}
    \begin{subfigure}{0.49\textwidth}
    \centering
    \small{
    \begin{equation*}
        \Gamma=\left(
        \begin{array}{ccccccc}
        1 & 1 & 0 & 1 & 0 & 0 & 0\\
        1 & 0 & 1 & 0 & 1 & 0 & 0\\
        0 & 1 & 1 & 0 & 0 & 1 & 0\\
        0 & 1 & 1 & 1 & 1 & 0 & 0\\
        1 & 0 & 1 & 1 & 0 & 1 & 0\\
        1 & 1 & 0 & 0 & 1 & 1 & 0\\
        0 & 0 & 0 & 0 & 0 & 0 & 1\\
        \end{array}
    \right)
    \end{equation*}
    }
    \end{subfigure}
    \caption{An example of two graph models, with the same underlying topological model but different (generic) min-cut structures, corresponding to the same PMI but different min-cut subspaces. The min-cut structures are specified by the $\Gamma$ matrices, where the rows are labeled by polychromatic indices $(A,B,C,AB,AC,BC,ABC)$ and the columns by the edges $(AB,AC,BC,A\sigma,B\sigma,C\sigma,O\sigma)$. In both cases the PMI is the full entropy space $\mathbb{R}^7$ (since no mutual information vanishes), but the graph on the left has the $6$-dimensional min-cut subspace $\mathbb{S}=(1,1,1,-1,-1,-1,1)^{\perp}$, while for the graph on the right $\mathbb{S}=\mathbb{R}^7$. Note that the $6$-dimensional min-cut subspace of the graph on the left is the hyperplane defined by the vanishing of the tripartite information, corresponding to the saturation of MMI.}
    \label{fig:i3generator}
\end{figure}
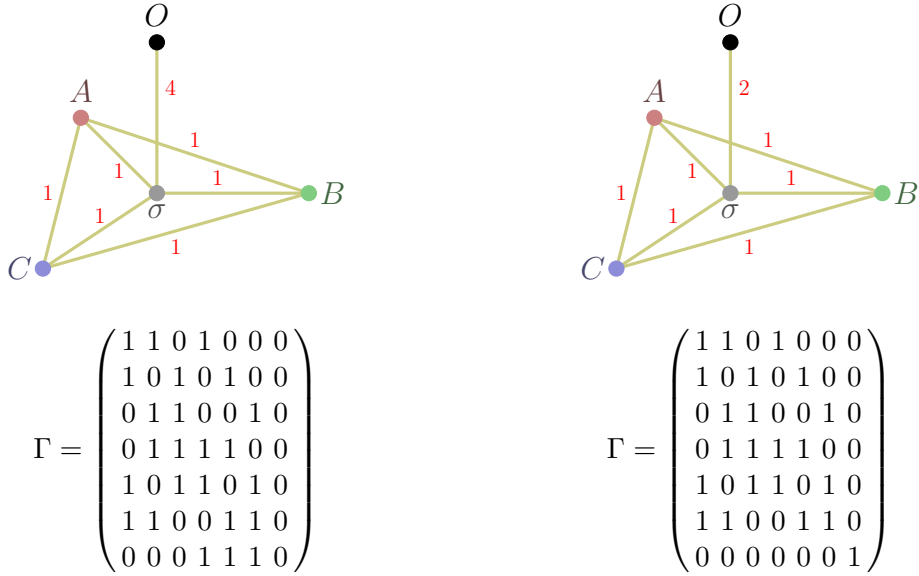

Having defined the PMI of a class $(\tgm{\N},\mathfrak{m})$, a natural question to ask is how is such a PMI related to the min-cut subspace. In general we have the following inclusion relation:

\begin{nlemma}
\label{lem:vinpmi}
For any class $(\tgm{\N},\mathfrak{m})$, the \emph{PMI} $\pi(\tgm{\N},\mathfrak{m})$ is the lowest dimensional $\N$-party \emph{PMI} that contains the min-cut subspace
\begin{equation}
    \mathbb{S}(\tgm{\N},\mathfrak{m})\subseteq\pi(\tgm{\N},\mathfrak{m})
\end{equation}
\end{nlemma}
\begin{proof}
    Given a class $(\tgm{\N},\mathfrak{m})$, \cref{cor:mpmi} implies that any graph model $\gm{\N}$ in this equivalence class has the same PMI, $\pi(\gm{\N})=\pi(\tgm{\N},\mathfrak{m})$. One then has the inclusion $\mathcal{S}(\tgm{\N},\mathfrak{m})\subset\pi(\tgm{\N},\mathfrak{m})$, and since $\mathbb{S}(\tgm{\N},\mathfrak{m})$ is the linear span of $\mathcal{S}(\tgm{\N},\mathfrak{m})$, that  $\mathbb{S}(\tgm{\N},\mathfrak{m})\subseteq\pi(\tgm{\N},\mathfrak{m})$. Furthermore, by \cref{def:gpmi}, it also follows that $\pi(\tgm{\N},\mathfrak{m})$ is the lowest dimensional $\N$-party PMI such that this inclusion holds.
\end{proof}

Let us now denote by $\mathbb{V}$ and $\mathbb{P}$ the min-cut subspace and PMI of $(\tgm{\N},\mathfrak{m})$ respectively. One may wonder if there could exist another topological graph model $\tgm{\N}'$ and min-cut structure $\mathfrak{m}'$ such that $\pi(\tgm{\N}',\mathfrak{m}')=\mathbb{P}$, while $\mathbb{S}(\tgm{\N}',\mathfrak{m}')\neq\mathbb{V}$. This can easily happen, as exemplified in \cref{fig:i3generator}. On the other hand, the opposite is not possible, as clarified by the following corollary:
\begin{ncor}
\label{cor:vfixespmi}
Given any two classes $(\tgm{\N},\mathfrak{m})$ and $(\tgm{\N}',\mathfrak{m}')$
\begin{equation}
    \mathbb{S}(\tgm{\N},\mathfrak{m})=\mathbb{S}(\tgm{\N}',\mathfrak{m}')\quad \Longrightarrow \quad \pi(\tgm{\N},\mathfrak{m})=\pi(\tgm{\N}',\mathfrak{m}')
\end{equation}
\end{ncor}
\begin{proof}
    By \cref{lem:vinpmi}, the PMI of a class $(\tgm{\N},\mathfrak{m})$ is the lowest dimensional $\N$-party PMI that contains the min-cut subspace $\mathbb{S}(\tgm{\N},\mathfrak{m})$, and is therefore uniquely fixed by such subspace.
\end{proof}

As a consequence of \cref{cor:vfixespmi}, a PMI $\mathbb{P}$ is completely determined by a min-cut subspace, and it is interesting to ask in what cases the min-cut subspace and PMI coincide. We will see in the next subsection that this is the case at least for a particular class of topological graph models called ``simple trees''. 

But before we proceed to the next section, let us briefly pause to summarize the list of coarser and coarser objects that along the way we have introduced and associated to graph models. The landscape of these constructs and the maps between them is shown in \cref{fig:summary}. We stress that, as we have clarified with various examples, none of these maps is injective.

\begin{figure}[tb]
\begin{tikzpicture}
    \hspace*{1.7cm}
    \node{{\Large $\gm{\N}\; \longmapsto\; (\tgm{\N},\mathfrak{m})\; \longmapsto\; \mathcal{S}\; \longmapsto\; \mathbb{S}\; \longmapsto\; \mathbb{P}$}};
    \draw[thick,dotted] (-3.65,0.1) -- (-3.65,1.5);
    \draw[thick,dotted] (0,-0.25) -- (0,-1.65);
    \draw[thick,dotted] (2,0.1) -- (2,1.5);
    \draw[thick,dotted] (4,-0.25) -- (4,-1.65);
    \node[align=left,font=\footnotesize\sf] at (-3.65,2.2) {organization of\\ graph models into\\ an equivalence class};
    \node[align=left,font=\footnotesize\sf] at (0,-2.35) {construction of\\ the S-cell from an\\ equivalence class};
    \node[align=left,font=\footnotesize\sf] at (2,2) {linear span\\ of the S-cell};
    \node[align=left,font=\footnotesize\sf] at (4,-2.35) {PMI of smallest\\ dimension that\\ contains $\mathbb{S}$};
\end{tikzpicture}
\vspace{0.5cm}
\caption{A summary of the various constructs that we have associated to graph models. Starting from a graph model $\gm{\N}$, each map along this chain associates to an object a coarser one. As we have exemplified throughout the text, none of these maps is injective.}
\label{fig:summary}
\end{figure}
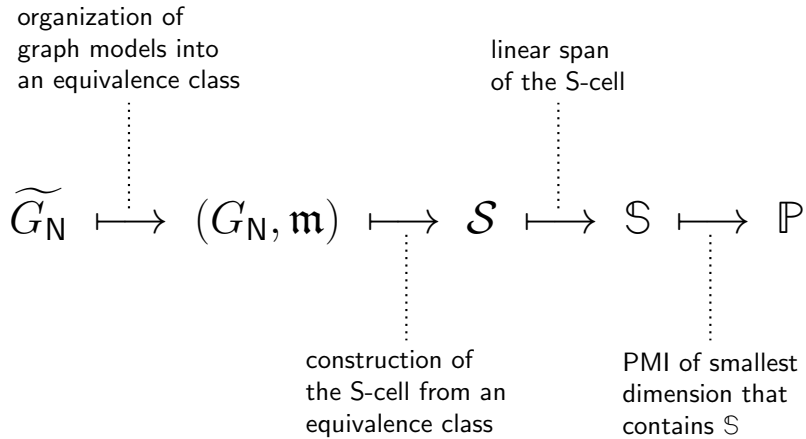

\subsection{The min-cut subspace of a simple tree graph}
\label{subsec:simpletreepmi}

Having introduced the notion of the pattern of marginal independence for a min-cut structure on a topological graph model, we will now show that for a particular class of such models that we call \textit{simple trees}, the pattern of marginal independence is equal to the min-cut subspace. 

The key attribute of such graphs, as we argue below, is that each edge defines a cut for some subsystem, and can thus be naturally associated to the corresponding polychromatic index.\footnote{\, There is a simple exception that, as we will shortly explain, is irrelevant.}
This will allow us to view any relation between the edge weights, which determines the min-cut subspace, purely in terms of subsystem entropies; this in turn can be recast in terms of mutual informations, and hence related to the PMI. Let us start from the basic definition 

\begin{ndefi}[Simple tree graph]\label{def:simple}
A topological graph  model $\tgm{\N}$ with the topology of a tree is \emph{simple} if each boundary vertex is labeled by a different color.
\end{ndefi}
On a given simple tree, consider an arbitrary edge $e\in E$, and the partition of the vertex set into the two complementary subsets $U$ and $U^{\complement}$ separated by $e$. If both $U$ and $U^{\complement}$ contain at least one boundary vertex, then $e$ corresponds to  a bipartition $(\underline{\I},\underline{\I}^{\complement})$ of $[\N+1]$ given by
\begin{equation}
    \beta(\partial V\cap U)=\underline{\I}
\end{equation}
By convention we define $U$ to be the subset that does not include the purifier, and simply write $\beta(\partial V\cap U)=\I$ (not underlined). The subsystem $\I$ associated to an edge via this prescription will be denoted by $\I(e)$. For an arbitrary tree graph and choice of $e$, one may also have $\beta(\partial V\cap U)=\varnothing$. This can happen if one or more leaves are not boundary vertices. In this case we write $\I(e)=\varnothing$, with a little abuse of notation since in this case $\I$ is not a proper polychromatic index according to our definition. While we take into account this possibility for the sake of completeness, notice that by topological minimality (cf., \cref{lem:topmin}) no edge $e$ with $\I(e)=\varnothing$ can belong to the set of cut edges for any min-cut.

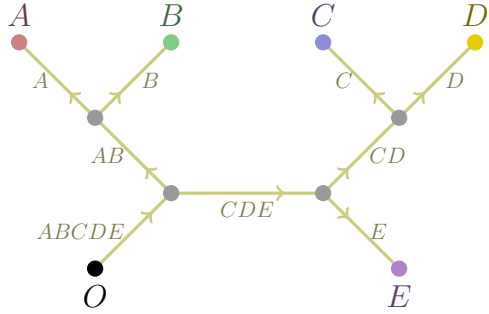
\begin{figure}[tb]
\centering
\begin{tikzpicture}
	\tikzset{->-/.style={decoration={
	  markings,
	  mark=at position .75 with {\arrow{>}}},postaction={decorate}}}
	\tikzset{-<-/.style={decoration={
	  markings,
	  mark=at position .75 with {\arrow{<}}},postaction={decorate}}}

\usetikzlibrary{decorations.markings}

	\tikzmath{
	\edgelabelsize=0.7;
		\v =0.1;  
	}
	\coordinate (Av) at (0,4);  
	\coordinate (Bv) at (2,4);  
	\coordinate (Cv) at (4,4);  
	\coordinate (Dv) at (6,4);  
	\coordinate (Ev) at (5,1);  
	\coordinate (Ov) at (1,1);  
	\coordinate (bvAB) at (1,3);  
	\coordinate (bvCD) at (5,3);  
	\coordinate (bvABCDE) at (2,2);  
	\coordinate (bvCDE) at (4,2);  

	\draw[edgestyle,-<-] (Av) -- node[scale=\edgelabelsize,edgecolor!60!black,left,midway]{$A$} (bvAB);
	\draw[edgestyle,-<-] (Bv) -- node[scale=\edgelabelsize,edgecolor!60!black,right,midway]{$B$} (bvAB);
	\draw[edgestyle,-<-] (Cv) -- node[scale=\edgelabelsize,edgecolor!60!black,left,midway]{$C$}  (bvCD);
	\draw[edgestyle,-<-] (Dv) -- node[scale=\edgelabelsize,edgecolor!60!black,right,midway]{$D$} (bvCD);
	\draw[edgestyle,-<-] (bvAB) -- node[scale=\edgelabelsize,edgecolor!60!black,left,midway]{$AB$} (bvABCDE);
	\draw[edgestyle,-<-] (bvCD) -- node[scale=\edgelabelsize,edgecolor!60!black,right,midway]{$CD$} (bvCDE);
	\draw[edgestyle,->-] (bvABCDE) -- node[scale=\edgelabelsize,edgecolor!60!black,below,midway]{$CDE$} (bvCDE);
	\draw[edgestyle,-<-] (Ev) -- node[scale=\edgelabelsize,edgecolor!60!black,right,midway]{$E$} (bvCDE);
	\draw[edgestyle,->-] (Ov) -- node[scale=\edgelabelsize,edgecolor!60!black,left,midway]{$ABCDE$} (bvABCDE);
	
	\filldraw [color=bvcolor] (bvAB) circle (\v) node[bvcolor!50!black,above right=0pt]{} ;
	\filldraw [color=bvcolor] (bvCD) circle (\v) node[bvcolor!50!black,above right=0pt]{} ;
	\filldraw [color=bvcolor] (bvABCDE) circle (\v) node[bvcolor!50!black,above right=0pt]{} ;
	\filldraw [color=bvcolor] (bvCDE) circle (\v) node[bvcolor!50!black,above right=0pt]{} ;

	\filldraw [color=Acolor] (Av) circle (\v) node[Acolor!50!black,above=2pt]{$A$} ;
	\filldraw [color=Bcolor] (Bv) circle (\v) node[Bcolor!50!black,above=2pt]{$B$} ;
	\filldraw [color=Ccolor] (Cv) circle (\v) node[Ccolor!50!black,above=2pt]{$C$} ;
	\filldraw [color=Dcolor] (Dv) circle (\v) node[Dcolor!50!black,above=2pt]{$D$} ;
	\filldraw [color=Ecolor] (Ev) circle (\v) node[Ecolor!50!black,below=2pt]{$E$} ;
	\filldraw [color=Ocolor] (Ov) circle (\v) node[Ocolor!50!black,below=2pt]{$O$} ;
	
\end{tikzpicture} 
    \caption{An example of a simple tree graph $\tgm{5}$, with explicit edge labeling and orientation indicated.}

    \label{fig:example_simple_tree_label}
\end{figure}

With this convention at hand, it is also convenient to introduce a canonical orientation of the edges of the tree which will induce an inclusion relation among the indices $\I(e)$. Specifically, denoting by $v_{_{\!\N+1}}$ the boundary vertex of the graph labeled by the purifier, and by $v_{_{\!L}}$ an arbitrary leaf, we consider the path from $v_{_{\!\N+1}}$ to $v_{_{\!L}}$ and orient the edges to turn it into a directed path which starts at $v_{_{\!\N+1}}$ and ends at $v_{_{\!L}}$. We will denote such a directed path by $\mathscr{P}(v_{_{\!\N+1}},v_{_{\!L}})$. By repeating this procedure for all leaves, we fix an orientation for the whole graph.

Using this orientation, we can then introduce a partial order on the set of edges. Given $e,f\in E$ we say that $e<f$ if there exists a leaf $L$ such that $e,f\in \mathscr{P}(v_{_{\!\N+1}},v_{_{\!L}})$ and $e$ precedes $f$. We then have
\begin{align}
\label{eq:ordering}
    e<f &\implies \I(e)\supseteq\I(f)\nonumber\\
    e,f \;\;\text{incomparable} &\implies \I(e)\cap\I(f)=\varnothing
\end{align}
where the equality $\I(e)=\I(f)$ can be attained if $e$ and $f$ are the edges adjoining to a degree-$2$ bulk vertex.\footnote{ \, Note that the implication in \cref{eq:ordering} applies in reverse as well, however most polychromatic subsystems do not have an associated edge (since a simple $\N$-party tree graph has at most $2\N-1$ edges, as compared to $\D=2^\N -1$ polychromatic subsystems).  In particular, two subsystems $\I(e)$ and $\I(f)$ can never be crossing (i.e. have a non-empty intersection which is a proper subset of both).}
An example of a simple tree graph with the explicit edge labeling and orientation is shown in \cref{fig:example_simple_tree_label}. 

This relation between edges and ``homologous'' subsystems allows us to conveniently translate the description of any $\I$-cut from the usual one based on vertices to a new one based on edges. Specifically, from the assumption of simplicity it follows that:

\begin{nlemma}
\label{lem:symdif}
    On a simple tree $\tgm{\N}$, a non-empty collection of edges $X\subseteq E$ is the set of cut-edges $\cut (U_{\I})$ for an $\I$-cut $U_{\I}$, where 
    \begin{equation}
    \label{eq:symdif}
      \I=\I(X)=\underset{e\in X}{\triangle}\K(e)
    \end{equation}
and $\triangle$ denotes the symmetric difference.\footnote{\, For any pair of sets $(X,Y)$, the symmetric difference is the disjunctive union
$$
	X \triangle Y 
	= (X \setminus Y) \cup (Y \setminus X)
	= (X \cup Y) \setminus (X \cap Y)
	= Y \triangle X
$$
Since the symmetric difference is associative, we can iterate this straightforwardly, so that for $n$-ary symmetric difference of a collection of sets we have just the elements which are in an odd number of the sets in that collection.  For example, for $n=3$
$$
	\underset{V\in\{X,Y,Z\}}{\triangle} V
	= X \triangle Y \triangle Z
	= (X \setminus (Y\cup Z)) \cup (Y \setminus (X\cup Z)) \cup (Z \setminus (X\cup Y)) \cup (X \cap Y \cap Z)\,.
$$
}

\end{nlemma}

\begin{proof}
    Given a simple tree $\tgm{\N}$ and a non-empty collection of edges $X$, we start by constructing a cut $U$ such that $\cut (U)=X$. Consider an edge $e\in X$ and a leaf $v_{_{\!L}}$ such that $e\in \mathscr{P}(v_{_{\!\N+1}},v_{_{\!L}})$. Starting from $v_{_{\!\N+1}}$, we follow the path $\mathscr{P}(v_{_{\!\N+1}},v_{_{\!L}})$ and label each vertex by $U$ or $U^{\complement}$ as follows. We label $v_{_{\!\N+1}}$ by $U^{\complement}$ then we follow the path and continue labeling the vertices by $U^{\complement}$ until we reach an edge in $X$. After we cross the edge we label all the vertices by $U$ until we reach another edge in $X$. We proceed in this fashion, alternating between $U^{\complement}$ and $U$ each time we cross an edge in $X$ until we reach $v_{_{\!L}}$. Then we repeat the same procedure following other paths, until we have labeled all the vertices in the graph. 
    
    Having constructed the desired cut, we now need to determine $\I(X)=\beta(U\cap\partial V)$. For a color $\ell\in [\N]$, denote by $v_{_{\!\ell}}$ the (unique by simplicity) vertex in $\tgm{\N}$ labeled by $\ell$. From the construction of $U$ described above, it follows that $\ell\in\I(X)$ if and only if the path $\mathscr{P}(v_{_{\!\N+1}},v_{_{\!\ell}})$ includes an odd number of edges in $X$. Furthermore, the index $\K(e)$ associated to an edge $e$ can be seen as the set of colors labeling the boundary vertices that follow $e$ in any path $\mathscr{P}(v_{_{\!\N+1}},v_{_{\!\ell}})$ that includes $e$. Therefore, a color $\ell$ appears in the expression at the right hand side of \cref{eq:symdif} if and only if the path $\mathscr{P}(v_{_{\!\N+1}},v_{_{\!\ell}})$ includes an odd number of edges in $X$, concluding the proof.
    
\end{proof}

Using this translation from a description in terms of vertices to one in terms of edges, we can then prove a useful property of min-cuts on simple tree graphs which is reminiscent of \cref{lem:mincutdecaltgen}.

\begin{nlemma}[Min-cut decomposition for simple trees]
\label{lem:mincutdecalt}
    Let $\tgm{\N}$ be a simple tree with a min-cut structure $\mathfrak{m}$, $\I$ a subsystem, and $\cut_{\I}^*$ the set of edges for some min-cut. Then any $X\subseteq\cut_{\I}^*$ is the set of edges $\cut_{\K}^*=X$ for a min-cut for the subsystem $\K=\K(X)$ given by \cref{eq:symdif}.
\end{nlemma}
\begin{proof}
    By \cref{lem:symdif}, $X\subseteq\cut_{\I}^*$ specifies a $\K$-cut for $\K=\K(X)$, so we just need to show it is minimal. Notice that, similarly,
    $Y=\cut_{\I}^* \setminus X$ specifies a $\J$-cut for $\J=\J(Y)$, and that by assumption $\cut_{\I}^* = X\cup Y$ are the min-cut edges for $\I=\I(X\cup Y)=\K(X) \Delta \J(Y)$. To show that $X$ are min-cut edges for $\K$, assume for contradiction that its actual min-cut edges are $X'$, with weight $\norm{X'}<\norm{X}$. Then since $\I=\K(X') \Delta \J(Y)$, $X'\cup Y$ specifies an $\I$-cut, and its weight gives the desired contradiction $\norm{X'\cup Y}<\norm{\cut_{\I}^*}$.
\end{proof}

In the particular case where a cut involves two edges, this in turn implies a particularly useful connection to the vanishing of an instance of the mutual information:

\begin{nlemma}
\label{lem:2edgesMI}
   Given a simple tree $\tgm{\N}$, a min-cut structure $\mathfrak{m}$ and any pair of edges $e,f\in E$, if $\{e,f\}=\cut^*_{\I(e)\triangle\I(f)}$ (for some choice of representative min-cut in case of degeneracy), then 
   \begin{equation}
   \label{eq:I2twoedges}
       S_{\I(e)}+S_{\I(f)}-S_{\I(e)\triangle\I(f)}
   \end{equation}
   is an instance of the mutual information in the \emph{PMI} of $(\tgm{\N},\mathfrak{m})$.
\end{nlemma}
\begin{proof}
   By \cref{lem:mincutdecalt} we have $\{e\}=\cut^*_{\I(e)}$ and $\{f\}=\cut^*_{\I(f)}$ (again for some choice of representatives in case of degeneracy) and the combination in \cref{eq:I2twoedges} vanishes. Therefore all we need to show is that there exists a choice of underlined indices $\underline{\J},\underline{\K}$ such that the expression in \cref{eq:i2instance} is equal to \cref{eq:I2twoedges}. If $e$ and $f$ are incomparable, cf. \cref{eq:ordering}, then $\I(e)\cap\I(f)=\varnothing$ which implies  $\I(e)\triangle\I(f)=\I(e)\cup\I(f)$ and we can choose $\underline{\J}=\I(e)$ and $\underline{\K}=\I(f)$. If $e<f$ (if $f<e$ simply swap $e$ and $f$ in what follows), then $\I(e)\supset\I(f)$ which implies $\I(e)\triangle\I(f)=\I(e)\setminus\I(f)$ and we can choose $\underline{\J}=\I(e)^{\complement}$ and $\underline{\K}=\I(f)$.
\end{proof}
Notice that as exemplified in \cref{fig:mincutdecexample} the implications of \cref{lem:mincutdecalt} and \cref{lem:2edgesMI} are stronger than what would follow from a straightforward iteration of \cref{lem:mincutdecaltgen}.

\begin{figure}[tb]
\begin{subfigure}{0.49\textwidth}
\centering
\begin{tikzpicture}

	\tikzmath{
	\edgelabelsize=0.7;
		\v =0.1;  
	}
	\coordinate (Av) at (0,4);  
	\coordinate (Bv) at (2,4);  
	\coordinate (Cv) at (3,3);  
	\coordinate (Dv) at (3,1);  
	\coordinate (Ev) at (2,0);  
	\coordinate (Ov) at (0,0);  
	\coordinate (bvAB) at (1,3);  
	\coordinate (bvCD) at (2,2);  
	\coordinate (bvABCD) at (1,2);  
	\coordinate (bvABCDE) at (1,1);  

	\draw[cutedgestyle] (Av) -- node[scale=\edgelabelsize,cutedgecolor!60!black,left,midway]{$A$} (bvAB);
	\draw[edgestyle] (Bv) -- node[scale=\edgelabelsize,edgecolor!60!black,right,midway]{$B$} (bvAB);
	\draw[edgestyle] (Cv) -- node[scale=\edgelabelsize,edgecolor!60!black,left,midway]{$C$}  (bvCD);
	\draw[edgestyle] (Dv) -- node[scale=\edgelabelsize,edgecolor!60!black,left,midway]{$D$} (bvCD);
	\draw[cutedgestyle] (bvAB) -- node[scale=\edgelabelsize,cutedgecolor!60!black,left,midway]{$AB$} (bvABCD);
	\draw[edgestyle] (bvCD) -- node[scale=\edgelabelsize,edgecolor!60!black,above,midway]{$CD$} (bvABCD);
	\draw[cutedgestyle] (bvABCD) -- node[scale=\edgelabelsize,cutedgecolor!60!black,left,midway]{$ABCD$} (bvABCDE);
	\draw[cutedgestyle] (Ev) -- node[scale=\edgelabelsize,cutedgecolor!60!black,right,midway]{$E$} (bvABCDE);
	\draw[edgestyle] (Ov) -- node[scale=\edgelabelsize,edgecolor!60!black,left,midway]{$ABCDE$} (bvABCDE);
	
	\filldraw [color=bvcolor] (bvAB) circle (\v) node[bvcolor!50!black,above right=0pt]{} ;
	\filldraw [color=bvcolor] (bvCD) circle (\v) node[bvcolor!50!black,above right=0pt]{} ;
	\filldraw [color=bvcolor] (bvABCD) circle (\v) node[bvcolor!50!black,above right=0pt]{} ;
	\filldraw [color=bvcolor] (bvABCDE) circle (\v) node[bvcolor!50!black,above right=0pt]{} ;

	\filldraw [color=Acolor] (Av) circle (\v) node[Acolor!50!black,above=2pt]{$A$} ;
	\filldraw [color=Bcolor] (Bv) circle (\v) node[Bcolor!50!black,above=2pt]{$B$} ;
	\filldraw [color=Ccolor] (Cv) circle (\v) node[Ccolor!50!black,above=2pt]{$C$} ;
	\filldraw [color=Dcolor] (Dv) circle (\v) node[Dcolor!50!black,above=2pt]{$D$} ;
	\filldraw [color=Ecolor] (Ev) circle (\v) node[Ecolor!50!black,below=2pt]{$E$} ;
	\filldraw [color=Ocolor] (Ov) circle (\v) node[Ocolor!50!black,below=2pt]{$O$} ;



\end{tikzpicture} 
\end{subfigure}
\begin{subfigure}{0.49\textwidth}
    \centering
    \footnotesize{
    \begin{align*}
        & \mathscr{L}_1\!:\;\;  AC\!D\!E \\[10 pt]
        & \mathscr{L}_2\!:\;\;  AC\!D,\; C\!D\!E,\; A\!E,\; A,\; C\!D,\; E\; \\[10 pt]
        & \mathscr{L}_3\!:\;\;  B,\; ABC\!D,\; AB,\; ABC\!D\!E\; \\[10 pt]
        & \mathscr{L}_4\!:\;\;  BC\!D,\; AB\!E,\; BC\!D\!E,\; B\!E\;
    \end{align*}
    }
\end{subfigure}
    \caption{An example of a simple tree with a choice of min-cut for $AC\!D\!E$, showing the different implications of \cref{lem:mincutdecaltgen} and \cref{lem:mincutdecalt}. The min-cut for $AC\!D\!E$ is specified by its set of cut edges $\cut^*_{AC\!D\!E}$, shown in red in the figure. We have labeled each edge $e$ in the graph by the corresponding polychromatic index $\I(e)$. The right panel shows the full list of subsystems whose min-cuts are fixed by $\cut^*_{AC\!D\!E}$ according to \cref{lem:mincutdecalt}, each one corresponding to a subset of $\cut^*_{AC\!D\!E}$ (including our starting choice $\mathscr{L}_1$). The min-cut for each subsystem in $\mathscr{L}_2$ can equivalently be obtained by a straightforward application of \cref{lem:mincutdecaltgen} (after some iteration). For the subsystems in $\mathscr{L}_3$ it is still sufficient to use \cref{lem:mincutdecaltgen}, but one also needs to consider min-cuts for complementary subsystems, which include the purifier (for example, the min-cut for $B$ is fixed by the fact that the min-cut for $BEO$, which is the complement of the min-cut for $ACD$, is disconnected). Even complementarity however is not sufficient to determine the min-cuts for the subsystems in $\mathscr{L}_4$ using \cref{lem:mincutdecaltgen}. Notice that we also have for example ${\sf I}(AB:E)=0$, which follows from \cref{lem:2edgesMI} again via \cref{lem:mincutdecalt}, but which is not implied by \cref{lem:mincutdecaltgen} (since $ABE$ is in the set $\mathscr{L}_4$).}
\label{fig:mincutdecexample}
\end{figure}
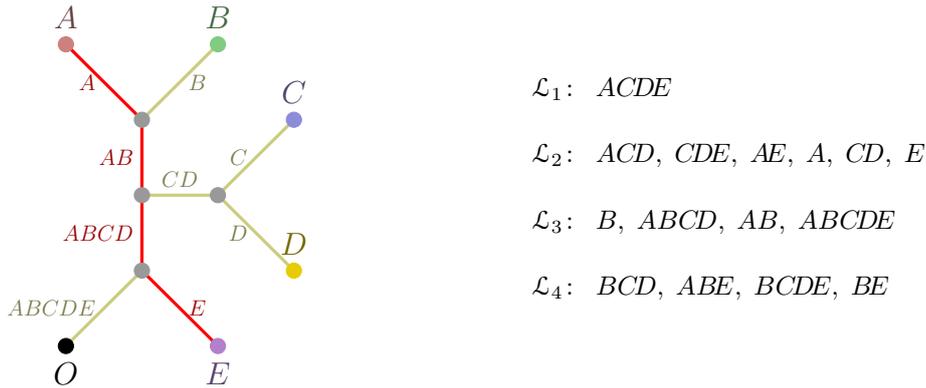

We are now ready to establish the anticipated connection between min-cut subspaces and PMIs for simple tree graphs. As usual we will discuss the case of generic min-cut structures first, and then extend the proof to min-cut structures that might include degeneracies. 

In the generic case, the proof will proceed by showing that due to the tree topology and simplicity, the set of vanishing instances of the mutual information completely determine the min-cut structure.
 
\begin{nlemma}
\label{lem:vpmisimple}
For any simple tree $\tgm{\N}$, and any generic min-cut structure $\mathfrak{m}$ on $\tgm{\N}$, the min-cut subspace $\mathbb{S}(\tgm{\N},\mathfrak{m})$ and the pattern of marginal independence $\pi(\tgm{\N},\mathfrak{m})$ coincide.
\end{nlemma}

\begin{proof}
Consider an arbitrary simple tree $\tgm{\N}$ and an arbitrary, but generic, min-cut structure $\mathfrak{m}$. We denote by $\mathbb{S}$ and $\mathbb{P}$ the min-cut subspace and PMI of $(\tgm{\N},\mathfrak{m})$ respectively, and by $\Gamma$ the linear map defined in \cref{eq:gamma_def}. To show that $\mathbb{S}$ and $\mathbb{P}$ coincide we only need to show that $\mathbb{S}\supseteq\mathbb{P}$, since by \cref{lem:vinpmi} we have $\mathbb{S}\subseteq\mathbb{P}$.

To show that $\mathbb{S}\supseteq\mathbb{P}$, we will prove that $\mathbb{S}^{\perp}\subseteq\mathbb{P}^{\perp}$, where $\perp$ denotes the orthogonal complement of a linear subspace. Since $\tgm{\N}$ is a simple tree, and $\mathfrak{m}$ is generic, the matrix which represents the map $\Gamma$ can be put schematically in the following form 
\begin{equation}
\label{eq:simpletreegamma}
\Gamma =
\left(
  \begin{array}{ccc}
  \widetilde{\Gamma} & \mathbb{0} & \mathbb{0}\\
  \mathbb{1} & \mathbb{0} & \mathbb{0}\\
  \mathbb{0} & \mathbb{1} & \mathbb{0}\\
  \end{array}
  \right)
\end{equation}
simply by permuting its rows and columns. Here the first, second and third columns of this block matrix correspond respectively to the subsets of edges which participate in at least two cuts in $\mathfrak{m}$, precisely one cut, or no cut at all. 

To see that $\Gamma$ can be put into this form, consider a subsystem $\I$ and the set of cut edges $\cut_{\I}^*$. If this set contains at least two edges, we put the row corresponding to the entropy $S_{\I}$ into the top row block in \cref{eq:simpletreegamma}. In the particular case where no such subsystem $\I$ exists, $\Gamma$ takes the particularly simple form 
\begin{equation}
    \left(
  \begin{array}{cc}
  \mathbb{1} & \mathbb{0}\\
  \end{array}
  \right)
\end{equation}
corresponding to the two bottom rows in \cref{eq:simpletreegamma}. And if all edges participate in at least one cut, we have $\Gamma=\mathbb{1}$. When $\widetilde{\Gamma}$ is non-trivial, the second row in \cref{eq:simpletreegamma} is guaranteed to exist by \cref{lem:mincutdecalt}, since each single edge in $\cut_{\I}^*$ is the min-cut for the corresponding subsystem given by \cref{lem:symdif}. The third row and second column, as well as the third column in \cref{eq:simpletreegamma} may be present or not, depending on the specific case. Their presence or absence however does not affect the essence of the rest of the proof.

Since $\mathbb{S}=\text{Im}\;\Gamma$, we have $\mathbb{S}^{\perp}=\text{Ker}\;\Gamma^{\intercal}$, and the generators of $\mathbb{S}^{\perp}$ are the columns of the matrix 
\begin{equation}
\label{eq:kernel}
\left(
  \begin{matrix}
   \mathbb{1}\\
  -\widetilde{\Gamma}^{\intercal}\\
  \mathbb{0}
  \end{matrix}
  \right)
\end{equation}
where again the last row may or may not be present depending on the specific form of $\Gamma$. In order to show the inclusion $\mathbb{S}\supseteq\mathbb{P}$, it is therefore sufficient to show that each column of this matrix is a linear combination of the columns of the matrix $\Pi(\mathbb{P})$ which specifies the PMI $(\tgm{\N},\mathfrak{m})$. 

In the particular case where $\widetilde{\Gamma}$ is trivial we have $\text{Ker}\;\Gamma^{\intercal}=\{\mathbf{0}\}$, implying that $\mathbb{S}=\mathbb{R}^{\D}$. Since in this case each entropy is computed by a cut of a single edge, by \cref{lem:vanishingI2} none of the instances of the mutual information vanish, therefore $\mathbb{P}=\mathbb{R}^{\D}=\mathbb{S}$ and the theorem holds.

Going back to the more general case, the labeling by polychromatic indices of the rows in each block of \cref{eq:kernel} is fixed by the construction of $\Gamma$. A column $q$ in \cref{eq:kernel} is the vector normal to the hyperplane in $\mathbb{R}^{\D}$ corresponding to the equation
\begin{equation}
\label{eq:entropysum}
    S_{\I}=\sum_{e\in \cut (\mU_{\I})}S_{\K(e)}
\end{equation}
where $\I$ is the polychromatic index labeling the row $q$ in the block $\mathbb{1}$ of \cref{eq:kernel}. In words, these equations say that each entropy $S_{\I}$ such the set $\cut_{\I}^*$ contains more than one edge, is equal to the sum of the entropies of the subsystems ``homologous'' to the edges in the set. This follows already from \cref{lem:mincutdecalt}, but \cref{eq:kernel} says that these equations correspond precisely to the generators of $\mathbb{S}^{\perp}$.

We can then rewrite \cref{eq:entropysum} as follows
\begin{align}
\label{eq:rewrite}
    0=&\;(S_{\K(e_1)}+S_{\K(e_2)}-S_{\K(e_1)\triangle\K(e_2)})\nonumber\\
    +&\;(S_{\K(e_1)\triangle\K(e_2)}+S_{\K(e_3)}-S_{\K(e_1)\triangle\K(e_2)\triangle\K(e_3)})\nonumber\\
    +&\;\ldots
\end{align}
where by \cref{lem:2edgesMI} each term in brackets can be recognized as an instance of the mutual information, cf. \cref{eq:i2instance}. Furthermore, any such instance is guaranteed to belong to the PMI of $(\tgm{\N},\mathfrak{m})$ by \cref{lem:mincutdecalt}, since the entropies that we added and subtracted in \cref{eq:rewrite} all correspond to subsystems homologous to subsets of the edges in $\cut_{\I}^*$. This shows that $\mathbb{S}^{\perp}\subseteq\mathbb{P}^{\perp}$, completing the proof.
\end{proof}

Finally, we extend this result to arbitrary (not necessarily generic) simple tree graphs. The central idea behind this generalization is  again that because of the tree topology, all degeneracy equations correspond to new instances of the mutual information that vanish.
 
\begin{nthm}
\label{thm:VPMIgeneral}
For any simple tree and min-cut structure, the \emph{PMI} and the min-cut subspace coincide.
\end{nthm}
\begin{proof}
As in the proof of \cref{lem:vpmisimple}, we only need to show that $\mathbb{S}\supseteq\mathbb{P}$, since by \cref{lem:vinpmi} $\mathbb{S}\subseteq\mathbb{P}$, and we will again show that $\mathbb{S}^{\perp}\subseteq\mathbb{P}^{\perp}$.

If the min-cut structure is degenerate, the matrix $\Gamma$ is not uniquely specified. However, as we discussed in \cref{sssec:gmdeg}, we can chose a representative min-cut for each degenerate subsystem, and while the specific $\Gamma$ will depend on the representatives, the min-cut subspace will not depend on this choice.

Suppose now that we make a choice of representatives for all degenerate subsystems, and therefore of $\Gamma$. We can imagine to determine a ``partial'' PMI by looking only at a single min-cut for each subsystem as specified by this choice, i.e., we determine the set of instances of the mutual information that vanish according to this subset of the min-cuts. We will denote this partial PMI by $\mathbb{P}_{\Gamma}$ to stress the dependence on this choice. We can then ignore the degeneracies and follow step by step the proof of \cref{lem:vpmisimple} to show that $(\text{Span}\;\Gamma)^{\perp}\subseteq\mathbb{P}_{\Gamma}^{\perp}$.

Because of the degeneracies however, $(\text{Span}\;\Gamma)^{\perp}\subseteq\mathbb{S}^{\perp}$ and $\mathbb{P}_{\Gamma}^{\perp}\subseteq\mathbb{P}^{\perp}$. In order to show that $\mathbb{S}^{\perp}\subseteq\mathbb{P}^{\perp}$, we will show that each degeneracy equation adds a new (not necessarily linearly independent) generator to $(\text{Span}\;\Gamma)^{\perp}$, and that this generator can be written as a linear combination of the generators in $\mathbb{P}^{\perp}$, i.e., that each generator of $\mathbb{S}^{\perp}$ which is not in $(\text{Span}\;\Gamma)^{\perp}$ is also in $\mathbb{P}^{\perp}$.

To see this, we can again follow a similar argument to the one we used in the proof of \cref{lem:vpmisimple}. A degeneracy equation for a subsystem $\I$ is an equation of the form
\begin{equation}
    \sum_{e\,\in\, \cut^{*\alpha}_{\I}}w_e=\sum_{f\,\in\, \cut^{*\beta}_{\I}}w_f
\end{equation}
where we used the more compact notation $\cut^{*\alpha}_{\I}=\cut ({\mU_{\I}}^{\alpha})$. Using \cref{lem:mincutdecalt} we can translate this equation into an equation for the entropies
\begin{equation}
\label{eq:newgenerator}
    \sum_{e\,\in\, \cut^{*\alpha}_{\I}}S_{\J(e)}=\sum_{f\,\in\, \cut^{*\beta}_{\I}}S_{\K(f)}
\end{equation}
which can be seen as a new generator in $\mathbb{S}^{\perp}$. Since both sides of this equation compute $S_{\I}$, we can also think of  \cref{eq:newgenerator} as a combination of the following two equations
\begin{align}
   & S_{\I}=\sum_{e\,\in\, \cut^{*\alpha}_{\I}}S_{\J(e)}\nonumber\\
   & S_{\I}=\sum_{f\,\in\, \cut^{*\beta}_{\I}}S_{\K(f)}
\end{align}
But each of these equations is of the form \cref{eq:entropysum} and can therefore be rewritten as in \cref{eq:rewrite}. All instances of the mutual information which appear in this decomposition belong to the PMI, therefore the new generator of $\mathbb{S}^{\perp}$ corresponding to \cref{eq:newgenerator} is a linear combination of the generators of $\mathbb{P}^{\perp}$. Repeating this construction for each degeneracy equation we obtain that $\mathbb{S}^{\perp}\subseteq\mathbb{P}^{\perp}$, concluding the proof.
\end{proof}

We stress that while the simple tree structure is a nice sufficient condition for the equivalence between PMIs and min-cut subspaces, it is by no means necessary. This equivalence is central in our arguments about the reconstruction of the HEC, and it would be interesting to extend it to a larger class of topological graph models and min-cut structures. While we leave this question for future work, in \cref{sec:logic} we will discuss in more detail what kind of  generalization of \cref{thm:VPMIgeneral} is necessary to achieve the reconstruction, and will see examples of graphs with highly non-trivial topology whose min-cut subspaces and PMIs coincide.

\section{Varying the number of parties}
\label{sec:recolorings}

Up to this point, we have been working with topological graph models and min-cut structures for an arbitrary, but \textit{fixed}, number of parties $\N$. However, in order to uncover some of the deepest relations between min-cut subspaces and PMIs, it is crucial to vary the number of parties and analyse how the various constructs that we have introduced thus far transform under this operation.

As we will see, the subtle behavior of min-cut subspaces and PMIs under recolorings is in stark contrast with the simple behavior of entropy vectors. This should perhaps make even more evident the fundamental differences between an analysis of holographic constraints based on entropy vectors and graph models, from one which purely relies on equivalence classes, like the one advocated here for the HEC, or the one based on proto-entropies for the holographic entropy polyhedron \cite{Hubeny:2018ijt}.

In \cref{subsec:cg} we will analyse \textit{coarse-grainings}, i.e., transformations that reduce the number of parties, and their effect on W-cells, S-cells and min-cut subspaces. The consequences of the opposite transformations, namely \textit{fine-grainings}, will be discussed in \cref{subsec:fg}. Our convention will be to denote by $\N$ the original number of parties and by $\N'$ the new number of parties, both for coarse-grainings and fine-grainings. \textit{Recolorings} of boundary vertices that reduce, respectively increase, the number of parties will be denoted by $\beta^{\downarrow}$ and $\beta^{\uparrow}$.

\subsection{Coarse-grainings of equivalence classes}
\label{subsec:cg}

Given an $\N$-party density matrix $\rho_{\N}$ and a purification $\ket{\psi}_{\N+1}$ in a Hilbert space of the form \cref{eq:purification},  consider a partition of the set $[\N+1]$ into $\N'+1$ (non-empty) parts, with $\N'<\N$.
Each element of the partition is an $\N$-party polychromatic index $\underline{\I}$ which we recolor by a monochromatic index $\ell'\in[\N'+1]$. We capture this coarse-graining by a map $\hat{\phi}: \ell' \mapsto \underline{\I}$, which also tells us which polychromatic indices $\underline{\I}\subseteq[\N+1]$ correspond to the coarse-grained ones $\underline{\I'}\subseteq[\N'+1]$ through
\begin{equation}
\label{eq:coarsegraining}
    \hat{\phi} : \underline{\I'} \;\mapsto\; \bigcup_{\ell'\,\in\,\underline{\I'}}\hat{\phi}(\ell')
\end{equation}

This partition of $[\N+1]$ and redefinition of the polychromatic indices corresponds to a redefinition of the Hilbert space $\underline{\mathcal{H}}$ in \cref{eq:purification} into a new Hilbert space $\underline{\mathcal{H}'}$ with $\N'+1$ factors, each of which is a collection of the original factors in $\underline{\mathcal{H}}$. For an $\N$-party density matrix $\rho$ acting on $\mathcal{H}$ as in \cref{eq:hilbert-space}, this transformation gives a new $\N'$-party density matrix $\rho'$ acting on a Hilbert space $\mathcal{H}'$, obtained from $\underline{\mathcal{H}'}$ by ignoring the new factor that contains the original $\mathcal{H}_{\N+1}$ factor (the purifier) in $\underline{\mathcal{H}'}$. 

We now want to apply such a transformation to an entropy vector, and to do so we need to use non-underlined indices. Notice that in general a coarse-graining as defined in \cref{eq:coarsegraining} can map the $\N$-party purifier $\N+1$ to an arbitrary $\N'$-party color, not necessarily to the $\N'$-party purifier $\N'+1$. Since our convention is that non-underlined polychromatic indices should not include the purifier, we introduce a new map $\phi$ that not only implements a coarse-graining as in \cref{eq:coarsegraining}, but also such that when it acts on an index $\I'$ it replaces $\hat{\phi}(\I')$ with $[\hat{\phi}(\I')]^{\complement}$ if $\hat{\phi}(\I')$ includes the original $\N+1$ subsystem.

The entropy vector of the new density matrix  $\mathbf{S}'(\rho')$ can then be obtained from the entropy vector $\mathbf{S}(\rho)$ of the original density matrix simply as 
\begin{equation}
    S'_{\I'}=S_{\phi(\I')}
\end{equation}
At the level of entropy vectors, we can therefore think of this transformation as a map
\begin{equation}
    \Phi_{\N\rightarrow\N'}:\mathbb{R}^{\D}\rightarrow\mathbb{R}^{\D'}\qquad \mathbf{S}\mapsto\mathbf{S}'=\Phi_{\N\rightarrow\N'} \, \mathbf{S}
\end{equation}
where $\D'=2^{\N'}-1$ and $\Phi_{\N\rightarrow\N'}$ is the $\D'\!\times\!\D$ matrix
\begin{equation}
\label{eq:colorproj}
    (\Phi_{\N\rightarrow\N'})_{\I'\K} \coloneqq \begin{cases}
    1\qquad \text{if}\;\;\K=\phi(\I')\\
    0\qquad \text{otherwise}
    \end{cases}
\end{equation}
This linear map will be referred to as a \textit{color-projection} in what follows.\footnote{\, Technically, the map defined in \cref{eq:colorproj} is not a projection, since it is not an endomorphism of entropy space. However we will still use this terminology since one can simply consider an embedding of $\mathbb{R}^{\D'}$ into $\mathbb{R}^{\D}$, in which case \cref{eq:colorproj} would be a projection.}

A coarse-graining of a graph model is defined similarly. Given $\gm{\N}$ one introduces a \emph{recoloring} by a new coloring map $\beta^{\downarrow} : \partial V \to [\N'+1]$. The recoloring $\beta^{\downarrow}$ induces a coarse-graining of polychromatic indices $\phi$ as described above, and the new entropy vector $\mathbf{S}'(\gm{\N'})$ is then obtained from $\mathbf{S}(\gm{\N})$ via the color-projection defined in \cref{eq:colorproj}
\begin{equation}
    \mathbf{S}'(\gm{\N'})=\Phi_{\N\rightarrow\N'} \, \mathbf{S}(\gm{\N})
\end{equation}
To see this, simply notice that the recoloring does not change the topology of the graph or the weights, and any min-cut of a coarse-grained subsystem $\I'$ is a min-cut for the corresponding original subsystem $\phi(\I')$.

Given the simplicity of the transformation rule for an entropy vector, one may be tempted to conclude that the same transformation also applies to min-cut subspaces. Instead, as we will see, the dimension of a min-cut subspace can even increase under coarse-graining, and this unexpected behavior is another example of a situation where a crucial role is played by the structure of degeneracies.

To understand how this can happen, and in what situations min-cut subspaces do transform analogously to entropy vectors, we need to consider the effect of a coarse-graining on a min-cut structure. As already mentioned above, a recoloring $\beta^{\downarrow}$ only affects the labeling of the boundary vertices of a topological graph model, not its topology, and the space of edge weights therefore is the same, before and after the coarse-graining. For a class $(\tgm{\N},\mathfrak{m})$ and a coarse-graining $\phi$ induced by a recoloring $\beta^{\downarrow}$ , the new min-cut structure $\mathfrak{m}'$ on the new topological graph model $\tgm{\N'}$ can be expressed in terms of the original coloring as
\begin{equation}
\label{eq:mcstructurecg}
    \mathfrak{m}'=\{\mathscr{U}_{\I'}=\mathscr{U}_{\phi(\I')}\; \text{for all}\; \I' \}\subset \mathfrak{m}
\end{equation}
In words, one can think of deriving $\mathfrak{m}'$ from $\mathfrak{m}$ by simply removing all min-cut sets $\mathscr{U}_{\I}$ for all subsystems $\I$ that are ``projected out'' by the coarse-graining, i.e., such that there is no $\I'$ with $\phi(\I')=\I$, and then relabeling the elements of $\mathfrak{m}'$ by the new polychromatic indices. 

Notice that an immediate consequence of \cref{eq:mcstructurecg} is that the relation between the subspaces corresponding to the solutions to the degeneracy equations in $\mathfrak{m}$ and $\mathfrak{m}'$ (cf. \cref{eq:degeneracy_eq}) is given by
\begin{equation}
\label{eq:Wcellinclusion}
    \mathbb{W}'\supseteq\mathbb{W}
\end{equation}
This simply follows from the fact that any equation that appears in $\mathfrak{m}'$ also appears in $\mathfrak{m}$, but in general, not vice versa.

From the transformation rule of the min-cut structures we can also derive the relation between the maps $\Gamma$ and $\Gamma'$ from $\mathbb{R}^3$ to $\mathbb{R}^{\D}$ and $\mathbb{R}^{\D'}$. In particular, since we are especially interested in the behaviour of degeneracies, we need to clarify if and how the relation between these maps can be affected by different choices of representative min-cuts. This is the content of the next lemma

\begin{nlemma}
\label{lem:gammamatricescg}
For any class $(\tgm{\N},\mathfrak{m})$, coarse-graining $\phi$ to a new class $(\tgm{\N'},\mathfrak{m}')$ induced by a recoloring $\beta^{\downarrow}$, and choice of map $\Gamma$ for $\mathfrak{m}$, there exists a choice of map $\Gamma'$ for $\mathfrak{m}'$ such that
\begin{equation}
\label{eq:gammamatricescg}
    \Gamma'=\Phi_{\N\rightarrow\N'}\,\Gamma
\end{equation}
\end{nlemma}
\begin{proof}
    Given a class $(\tgm{\N},\mathfrak{m})$ consider the map $\Gamma_{\{\alpha\}}$ specified by a choice of representative min-cuts ${\mU_{\I}}^{\alpha}$ for all $\N$-party polychromatic indices $\I$. By \cref{eq:mcstructurecg} it follows that
    \begin{equation*}
        {\mU_{\I'}}^{\alpha}\coloneqq{\mU_{\phi(\I')}}^{\!\!\!
        \!\!\!\!\!\alpha}
    \end{equation*}
    is a choice of representative min-cuts for all the coarse-grained polychromatic indices $\I'$. Notice that the index $\alpha$ did not change, since for the subsystems which are not removed by the coarse-graining, any possible choice of min-cut in $\mathfrak{m}$ is a valid choice of min-cut in $\mathfrak{m}'$. 
    
    Recall now the definition of the matrix $\Gamma$ given in \cref{eq:incidence}. Using for $\mathfrak{m}'$ the choice of representatives induced by the choice for $\mathfrak{m}$ just described, the matrix $\Gamma'$ can be obtained from $\Gamma$ as follows. We first delete the rows corresponding to the subsystems $\I$ which are removed by the coarse-graining, and then permute the remaining rows according to the relabeling fixed by $\beta^{\downarrow}$. But this is precisely the transformation performed on $\Gamma$ by the matrix $\Phi_{\N\rightarrow\N'}$ defined in \cref{eq:colorproj}.
\end{proof}

We are now ready to discuss a general situation where the transformation of min-cut subspaces under coarse-grainings is well behaved, and completely determined by the map $\Phi_{\N\rightarrow\N'}$ between entropy spaces. The next result shows that this is the case whenever \cref{eq:Wcellinclusion} is saturated.  
\begin{nthm}[Color-projections of min-cut subspaces]
\label{lem:VspaceCG}
For any class $(\tgm{\N},\mathfrak{m})$ and coarse-graining $\phi$ induced by a recoloring $\beta^{\downarrow}$ to a new class $(\tgm{\N'},\mathfrak{m}')$
\begin{equation}
\label{eq:vspacecg}
    \mathbb{W}'=\mathbb{W}\qquad\Longrightarrow\qquad\mathbb{S}'=\Phi_{\N\rightarrow\N'}\,\mathbb{S}
\end{equation}
\end{nthm}
\begin{proof}
     For a given class $(\tgm{\N},\mathfrak{m})$ consider one of the maps $\Gamma$ defined in \cref{eq:gamma_def} (as usual for some choice of representative min-cuts). For any coarse-graining $\phi$ to a new class $(\tgm{\N'},\mathfrak{m}')$, \cref{lem:gammamatricescg} guarantees that there exists a choice of representative min-cuts for $\mathfrak{m}'$ such that the corresponding map $\Gamma'$ is related to $\Gamma$ via \cref{eq:gammamatricescg}. In general the relation between $\mathbb{W}$ and the min-cut subspace is given by \cref{lem:VfromW}. Therefore by the assumption that $\mathbb{W}'=\mathbb{W}$ we have
     \begin{equation*}
         \mathbb{S}'=\Gamma'(\mathbb{W}')=(\Phi_{\N\rightarrow\N'}\,\Gamma)(\mathbb{W})=\Phi_{\N\rightarrow\N'}\,(\Gamma(\mathbb{W}))=\Phi_{\N\rightarrow\N'}\,\mathbb{S}
     \end{equation*}
     completing the proof.
\end{proof}

A straightforward consequence of this theorem is that min-cut subspaces of generic min-cut structures are always well behaved under coarse-grainings
\begin{ncor}
For any class $(\tgm{\N},\mathfrak{m})$ and coarse-graining $\phi$ induced by a recoloring $\beta^{\downarrow}$, if $\mathfrak{m}$ is generic then \cref{eq:vspacecg} holds.
\end{ncor}
\begin{proof}
    Given a $(\tgm{\N},\mathfrak{m})$, if $\mathfrak{m}$ is generic we have $\mathbb{W}=\mathbb{R}^{\E}$, since the W-cell $\mathcal{W}(\tgm{\N},\mathfrak{m})$ is full dimensional. By \cref{eq:mcstructurecg} a coarse-graining cannot introduce new degeneracies. Therefore the W-cell of the coarse-grained min-cut structure also spans $\mathbb{R}^{\E}$, and \cref{lem:VspaceCG} applies. 
\end{proof}

We have seen that the saturation of \cref{eq:Wcellinclusion} is a sufficient condition for min-cut subspaces to transform according to \cref{eq:vspacecg}. On the other hand, if the inclusion (\ref{eq:Wcellinclusion}) is strict, which can happen if some degeneracy equations are ``lost'' under the coarse-graining, the relation between the min-cut subspaces is more complicated, and depends on the details of the graph, the min-cut structure and the recoloring. In fact, as already mentioned earlier, the dimension of the min-cut subspace can even grow, which can easily happen for coarse-grainings of highly degenerate min-cut structures. 

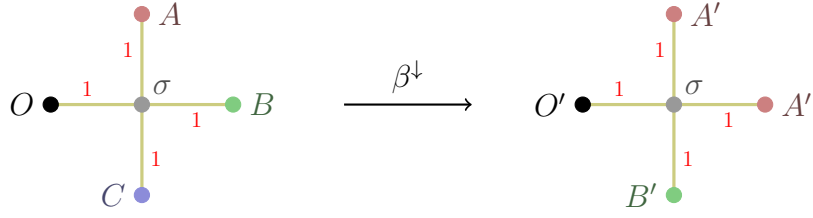
\begin{figure}[tb]
\centering
\begin{tikzpicture}
	\tikzmath{
		\edgelabelsize=0.7;
		\v =0.1;  
		\e=1.2;  
		\shift=7;  
	}
	\coordinate (bv) at (0,0);  
	\coordinate (Av) at (90:\e);  
	\coordinate (Bv) at (0:\e);  
	\coordinate (Cv) at (270:\e);  
	\coordinate (Ov) at (180:\e);  
	\coordinate (bvv) at ($ (0,0)+(\shift,0) $);  
	\coordinate (Avv) at ($ (90:\e)+(\shift,0) $);  
	\coordinate (Bvv) at ($ (0:\e)+(\shift,0) $);  
	\coordinate (Cvv) at ($ (270:\e)+(\shift,0) $);  
	\coordinate (Ovv) at ($ (180:\e)+(\shift,0) $);  

	\draw[edgestyle] (Av) -- node[scale=\edgelabelsize,edgeweightcolor,left,pos=0.4]{$1$} (bv);
	\draw[edgestyle] (Bv) -- node[scale=\edgelabelsize,edgeweightcolor,below,pos=0.4]{$1$} (bv);
	\draw[edgestyle] (Cv) -- node[scale=\edgelabelsize,edgeweightcolor,right,pos=0.4]{$1$}  (bv);
	\draw[edgestyle] (Ov) -- node[scale=\edgelabelsize,edgeweightcolor,above,pos=0.4]{$1$} (bv);
	\filldraw [color=bvcolor] (bv) circle (\v) node[bvcolor!50!black,above right=0pt]{$\sigma$} ;
	\filldraw [color=Acolor] (Av) circle (\v) node[Acolor!50!black,right=2pt]{$A$} ;
	\filldraw [color=Bcolor] (Bv) circle (\v) node[Bcolor!50!black,right=2pt]{$B$} ;
	\filldraw [color=Ccolor] (Cv) circle (\v) node[Ccolor!50!black,left=2pt]{$C$} ;
	\filldraw [color=Ocolor] (Ov) circle (\v) node[Ocolor!50!black,left=2pt]{$O$} ;
	\draw[thick,black,->] (0.38*\shift,0) -- node[above]{$\beta^{\downarrow}$} (0.62*\shift,0);
	\draw[edgestyle] (Avv) -- node[scale=\edgelabelsize,edgeweightcolor,left,pos=0.4]{$1$} (bvv);
	\draw[edgestyle] (Bvv) -- node[scale=\edgelabelsize,edgeweightcolor,below,pos=0.4]{$1$} (bvv);
	\draw[edgestyle] (Cvv) -- node[scale=\edgelabelsize,edgeweightcolor,right,pos=0.4]{$1$}  (bvv);
	\draw[edgestyle] (Ovv) -- node[scale=\edgelabelsize,edgeweightcolor,above,pos=0.4]{$1$} (bvv);
	\filldraw [color=bvcolor] (bvv) circle (\v) node[bvcolor!50!black,above right=0pt]{$\sigma$} ;
	\filldraw [color=Acolor] (Avv) circle (\v) node[Acolor!50!black,right=2pt]{$A'$} ;
	\filldraw [color=Acolor] (Bvv) circle (\v) node[Acolor!50!black,right=2pt]{$A'$} ;
	\filldraw [color=Bcolor] (Cvv) circle (\v) node[Bcolor!50!black,left=2pt]{$B'$} ;
	\filldraw [color=Ocolor] (Ovv) circle (\v) node[Ocolor!50!black,left=2pt]{$O'$} ;
\end{tikzpicture} 
    \caption{The coarse-graining (\ref{eq:cgexample}) of the perfect state graph model $\gm{3}$ from \cref{fig:exampleone}.}
    \label{fig:cgperfectstate}
\end{figure}

As a simple example of this non-trivial behavior, consider the graph models depicted in \cref{fig:cgperfectstate}. The one on the left is the graph that we have already seen before, which realizes the $\N=3$ perfect state, while the one on the right is obtained via the coarse-graining to $\N'=2$ specified by
\begin{equation}
\label{eq:cgexample}
    \hat{\phi}(A')=AB\qquad \hat{\phi}(B')=C\qquad \hat{\phi}(O')=O
\end{equation}
The entropy ray obtained from this coarse-graining of $\gm{3}$ can be directly computed from the ray in \cref{eq:perfectray}, obtaining
\begin{equation}
\label{eq:cgperfectstate}
     \mathbf{S}(\gm{2}) = \lambda \, (2,1,1),\qquad \lambda>0
\end{equation}
However we do not want to simply derive the new entropy ray, but also the min-cut subspace of the new topological graph model and min-cut structure. After the recoloring, the only degenerate min-cut that remains from \cref{eq:wcellexample} is
\begin{equation}
\begin{aligned}
    \mathscr{U}_{A'} &= \{ \{A'_1,A'_2\}, \{A'_1,A'_2,\sigma\} \} \\
\end{aligned}
\end{equation}
(where $A'_1,A'_2$ denote the two original $A,B$ boundary vertices that have now been recolored) and the only degeneracy equation left from \cref{eq:degeqexample} is
\begin{equation}
\begin{aligned}
    w_{_{A'_1\sigma}}+w_{_{A'_2\sigma}} &= w_{_{B'\sigma}}+w_{_{O'\sigma}} 
\end{aligned}
\end{equation}
The solution to this degeneracy equation is the $3$-dimensional subspace
\begin{equation}
    \mathbb{W}'=(1,1,-1,-1)^{\perp}\subset \mathbb{R}^{4}
\end{equation}
where we ordered the weights according to $(A'_1\sigma,A'_2\sigma,B'\sigma,O'\sigma)$. The W-cell is the interior of the following polyhedral cone in $\mathbb{W}'$ (written as embedded in $\mathbb{R}^4_{>0}$)
\begin{equation}
\label{eq:cgperfectwcell}
    \text{cone}\, \{(1,0,0,1),\, (0,1,0,1),\, (0,1,1,0),\, (1,0,1,0) \}
\end{equation}
Its image under the map
\begin{equation}
\label{eq:gammaperfectcg}
    \Gamma=\left(
    \begin{array}{cccc}
    1 & 1 & 0 & 0 \\
    0 & 0 & 1 & 0 \\
    0 & 0 & 0 & 1 \\
    \end{array}
    \right)
\end{equation}
fixed by the choice of representative $\mU_{\!A'}=\{A'_1,A'_2\}$ is the S-cell, which is the interior of the following polyhedral cone in $\mathbb{R}^3$
\begin{equation}
\label{eq:facetbycg}
    \text{cone}\, \{(1,0,1),\, (1,1,0) \}
\end{equation}
This can easily be recognized as the facet of the $\text{HEC}_2=\text{SAC}_2$ supported by the min-cut subspace (or equivalently, PMI), ${\sf I}(B':O')=0$. Notice that the straightforward projection of the perfect state entropy ray from \cref{eq:perfectray} given in \cref{eq:cgperfectstate} is just one of the rays on this facet and has no particular meaning.

The example we just discussed also shows that, like for the min-cut subspace, the dimension of the PMI of a topological graph model and min-cut structure can similarly grow under coarse-graining. Indeed, the reader can easily verify that the PMI of the graph model $\gm{3}$ of \cref{fig:exampleone} is $1$-dimensional, while the coarse-grained graph has a $2$-dimensional PMI.

\subsection{Fine-grainings of equivalence classes}
\label{subsec:fg}

For fine-grainings of density matrices and graph models we can proceed similarly as for coarse-grainings. We consider the case of a state $\ket{\psi}_{\N+1}$ in a Hilbert space $\underline{\mathcal{H}}$ as in \cref{eq:purification}, such that some of the factors admit a tensor product structure into ``finer components'', giving a new Hilbert space $\underline{\mathcal{H}'}$. An $\N$-party density matrix $\rho$ acting on $\mathcal{H}$ in  \cref{eq:hilbert-space} can then be seen as a new $\N'$-party density matrix $\rho'$, with $\N'>\N$, acting on $\mathcal{H}'$. From $\rho'$ we can then compute the entropy vector $\mathbf{S}'(\rho')$, which in general will depend on the details of the initial density matrix $\rho$. For any fine-graining however, we can always chose an appropriate coarse-graining $\phi$ that ``undoes'' it. No matter what the details of the initial density matrix $\rho$ are, the entropy vectors will then be related by this coarse-graining as usual 
\begin{equation}
\label{eq:entropyvecfg}
    \mathbf{S}=\Phi_{\N'\rightarrow\N}\, \mathbf{S}'
\end{equation}

As in the case of coarse-grainings, fine-grainings of graph models can be defined similarly to fine-grainings of density matrices. For a given graph model $\gm{\N}$ realizing an entropy vector $\mathbf{S}$ we consider a boundary recoloring $\beta^{\uparrow} : \partial V \to [\N'+1]$, which will specify a new \textit{fine-grained graph model} $\gm{\N'}$ from which we obtain a new entropy vector $\mathbf{S}'$. We can then ``undo'' the fine-graining with a coarse-graining $\phi$ specified by a boundary recoloring $\beta^{\downarrow}$
\begin{equation}
\label{eq:fgundoing}
    \beta^{\downarrow}=\beta
\end{equation}
where $\beta$ is simply the initial coloring of $\gm{\N}$. The entropy vectors before and after the fine-graining induced by $\beta^{\uparrow}$ will then be related as in \cref{eq:entropyvecfg}. 

For fine-grainings of min-cut structures on topological graph models on the other hand, we again need to be more careful. Given a class $(\tgm{\N},\mathfrak{m})$, a boundary recoloring $\beta^{\uparrow}$ gives a new topological graph model $\tgm{\N'}$, but in general this does not automatically induce a fine-graining, since the new min-cut structure might be incomplete. The reason is that in general $\mathfrak{m}$ does not specify the min-cuts for all subsystems $\I'$, and the fine-grainings induced by $\beta^{\uparrow}$ on different graph models $\gm{\N}\in(\tgm{\N},\mathfrak{m})$ can correspond to different min-cut structures.

To take this indeterminacy into account, we define the set $\mathfrak{m}\big\uparrow\!_{_{\N'}}$ of min-cut structures $\mathfrak{m}'$ for $\tgm{\N'}$ that reduce to the original min-cut structure $\mathfrak{m}$ on $\tgm{\N}$ when we undo the fine-graining, specifically 
\begin{equation}
\label{eq:mstructurelifts}
    \mathfrak{m}\big\uparrow\!_{_{\N'}} \coloneqq \{\mathfrak{m}':\;\mathscr{U}_{\I}=\mathscr{U}_{\phi(\I)}\; \text{for all}\; \I \}
\end{equation}
where $\phi$ is the coarse-graining induced by the recoloring $\beta^{\downarrow}$ given in \cref{eq:fgundoing}. Notice that for any $(\tgm{\N},\mathfrak{m})$ and recoloring $\beta^{\uparrow}$, this set is non-empty. The existence of at least one min-cut structure in $\mathfrak{m}\big\uparrow\!_{_{\N'}}$ is obvious: simply consider any graph model $\gm{\N}\in(\tgm{\N},\mathfrak{m})$, apply the recoloring $\beta^{\uparrow}$ and determine the min-cut structure $\mathfrak{m}'$ of the new graph model $\gm{\N'}$. We then define a fine-graining of $(\tgm{\N},\mathfrak{m})$ as follows:
\begin{ndefi}[Fine-graining of a class $(\tgm{\N},\mathfrak{m})$]
\label{def:eqclassfg}
A fine-graining of a class $(\tgm{\N},\mathfrak{m})$ is a pair $(\beta^{\uparrow},\mathfrak{m}')$ of a recoloring $\beta^{\uparrow}$ and a min-cut structure from $\mathfrak{m}\big\uparrow\!_{_{\N'}}$.
\end{ndefi}

We stress the fundamental difference between a fine-graining of a graph model as defined above, and a fine-graining of an equivalence class as defined here. Given a class $(\tgm{\N},\mathfrak{m})$ and a choice of representative $\gm{\N}$, any recoloring $\beta^{\uparrow}$ will automatically specify a fine-graining of $\gm{\N}$, but in general this is only \textit{one of the possible} fine-grainings of the class $(\tgm{\N},\mathfrak{m})$ according to \cref{def:eqclassfg}. In fact, for a fixed recoloring $\beta^{\uparrow}$, the induced fine-graining of different choices of representatives will in general correspond to different fine-grainings of $(\tgm{\N},\mathfrak{m})$. Importantly, this can even happen when the min-cut subspace of $(\tgm{\N},\mathfrak{m})$ is $1$-dimensional, and one may be inclined to think that a graph model and its equivalence class are essentially the same object. As we will see (and exemplify) in what follows, the reason for this unexpected behavior is that even when the min-cut subspace is $1$-dimensional, the W-cell can be higher dimensional. 

As usual, a choice of representative of an equivalence class is a convenient way to specify a min-cut structure, and in \cref{sec:logic} we will often resort to fine-grainings of graph models. However, the reader should always remain aware of the fundamental difference highlighted here, and of the possibility of alternative fine-grainings of an equivalence class specified by a choice of graph model.

Let us now continue our analysis of the set of possible fine-grainings of a topological graph model and a min-cut structure. We have seen in the previous subsection that in certain situations min-cut subspaces are well behaved under coarse-grainings, in the sense that they transform under the same projection which determines a coarse-grained entropy vector. Given a class $(\tgm{\N},\mathfrak{m})$, and a recoloring $\beta^{\uparrow}$, it is then natural to ask under what conditions there exists a fine-graining, i.e., a choice of a new min-cut structure $\mathfrak{m}'\in\mathfrak{m}\big\uparrow\!_{_{\N'}}$, such that \cref{lem:VspaceCG} applies. 

To answer this question, and understand the origin of some of the properties of fine-grainings mentioned above, we first need to analyse in more detail the effect of a recoloring $\beta^{\uparrow}$ on the set of W-cells for a topological graph model $\tgm{\N}$. As we have seen, a recoloring does not change the topology of $\tgm{\N}$, and in particular its set of edges, therefore the space of edge weights $\mathbb{R}^{\E}_{>0}$ is unaffected by the recoloring and we can compare W-cells before and after it. Their relation is captured by the following lemma 

\begin{nlemma}[Refinement of W-cells]
\label{lem:Wrefinementfg}
Given a class $(\tgm{\N},\mathfrak{m})$ and any recoloring $\beta^{\uparrow}$, the set of \emph{W}-cells $\mathcal{W}_{\mathfrak{m}'}$ for all $\mathfrak{m}'\in\mathfrak{m}\big\uparrow\!_{_{\N'}}$ is a partition of $\mathcal{W}_{\mathfrak{m}}$.
\end{nlemma}
\begin{proof}
Given a class $(\tgm{\N},\mathfrak{m})$ and any recoloring $\beta^{\uparrow}$, consider a min-cut structure $\mathfrak{m}'\in\mathfrak{m}\big\uparrow\!_{_{\N'}}$ and the corresponding W-cell $\mathcal{W}_{\mathfrak{m}'}$. Denoting by $\phi$ the coarse-graining induced by the recoloring $\beta^{\downarrow}$ determined by \cref{eq:fgundoing}, it follows from \cref{eq:mstructurelifts} that all inequalities (\ref{eq:open_contraints2}) and degeneracy equations (\ref{eq:degeneracy_eq}) that specify the W-cell $\mathcal{W}_{\mathfrak{m}}$ of the original min-cut structure $\mathfrak{m}$ also belong to the set of inequalities and degeneracy equations that specify $\mathcal{W}_{\mathfrak{m}'}$. This implies that $\mathcal{W}_{\mathfrak{m}'}\subseteq\mathcal{W}_{\mathfrak{m}}$, for all $\mathfrak{m}'\in\mathfrak{m}\big\uparrow\!_{_{\N'}}$. Furthermore, by \cref{lem:Wpartition}, the different W-cells $\mathcal{W}_{\mathfrak{m}'}$ do not intersect. Therefore it only remains to prove that the union of all W-cells $\mathcal{W}_{\mathfrak{m}'}$ for all $\mathfrak{m}'\in\mathfrak{m}\big\uparrow\!_{_{\N'}}$ is $\mathcal{W}_{\mathfrak{m}}$. Consider a graph model $\gm{\N}\in(\tgm{\N},\mathfrak{m})$, specified by a weight vector $\mathbf{w}\in\mathcal{W}_{\mathfrak{m}}$. The recoloring $\beta^{\uparrow}$ induces a fine-graining of $\gm{\N}$ to a new graph model $\gm{\N'}\in(\tgm{\N'},\mathfrak{m}_{\mathbf{w}})$ where $\mathfrak{m}_{\mathbf{w}}$ is one of the min-cut structures in $\mathfrak{m}\big\uparrow\!_{_{\N'}}$, completing the proof.
\end{proof}

We are now ready to establish for what classes $(\tgm{\N},\mathfrak{m})$ and recolorings $\beta^{\uparrow}$, there exists a fine-graining $(\beta^{\uparrow},\check{\mathfrak{m}})$ such that the min-cut subspace of $(\tgm{\N},\mathfrak{m})$ can be obtained from that of $(\tgm{\N'},\check{\mathfrak{m}})$ via the color projection associated to the coarse-graining $\phi$ specified by the recoloring $\beta^{\downarrow}$ given in \cref{eq:fgundoing}. The next lemma clarifies that this is always the case.

\begin{nlemma}
\label{lem:canocialfg}
    For any class $(\tgm{\N},\mathfrak{m})$ and recoloring $\beta^{\uparrow}$, there exists a choice of min-cut structure $\check{\mathfrak{m}}\in\mathfrak{m}\big\uparrow\!_{_{\N'}}$ such that 
    \begin{equation}
    \label{eq:canonicalfg}
        \mathbb{S}(\tgm{\N},\mathfrak{m})=\Phi_{\N'\rightarrow\N}\,\mathbb{S}(\tgm{\N'},\check{\mathfrak{m}})
    \end{equation}
\end{nlemma}
\begin{proof}
    We only need to prove that \cref{lem:VspaceCG} applies, and for this we only need to show that there exists $\check{\mathfrak{m}}\in\mathfrak{m}\big\uparrow\!_{_{\N'}}$ such that $\mathbb{W}_{\check{\mathfrak{m}}}=\mathbb{W}_{\mathfrak{m}}$, where $\mathbb{W}_{\check{\mathfrak{m}}}=\text{Span}\,(\mathcal{W}_{\check{\mathfrak{m}}})$ and $\mathbb{W}_{\mathfrak{m}}=\text{Span}\,(\mathcal{W}_{\mathfrak{m}})$. This follows from \cref{lem:Wrefinementfg}, since the W-cells $\mathcal{W}_{\mathfrak{m}'}$ for all $\mathfrak{m}'\in\mathfrak{m}\big\uparrow\!_{_{\N'}}$ form a \textit{finite} partition of $\mathcal{W}_{\mathfrak{m}}$, and therefore there must exist at least one min-cut structure $\check{\mathfrak{m}}\in\mathfrak{m}\big\uparrow\!_{_{\N'}}$ such that $\text{Span}\,(\mathcal{W}_{\check{\mathfrak{m}}})=\text{Span}\,(\mathcal{W}_{\mathfrak{m}})$.
\end{proof}

While this lemma proves the existence of at least one fine-graining such that \cref{eq:canonicalfg} applies, it should be clear from the proof that this fine-graining is in general non-unique. Since these fine-grainings will play a crucial role in the next section, we introduce the following definition 
\begin{ndefi}[Minimally-degenerate fine-graining]
\label{def:mmdgfg}
For any class $(\tgm{\N},\mathfrak{m})$ and recoloring $\beta^{\uparrow}$, a minimally-degenerate fine-graining $(\beta^{\uparrow},\check{\mathfrak{m}})$ is any fine-graining such that 
\begin{equation}
\label{eq:mmdgfg}
    \text{\emph{Span}}\,(\mathcal{W}_{\check{\mathfrak{m}}})= \text{\emph{Span}}\,(\mathcal{W}_{\mathfrak{m}})
\end{equation}
\end{ndefi}
As explained in the proof of \cref{lem:canocialfg}, the condition in \cref{eq:mmdgfg} is sufficient to guarantee that \cref{eq:canonicalfg} holds, but a priori it is not necessary. Therefore it is in principle possible that there exists fine-grainings which are not minimally-degenerate according to \cref{def:mmdgfg}, but which are still well behaved in the sense of the coarse-graining in \cref{eq:canonicalfg}. Notice that in the particular case of a class $(\tgm{\N},\mathfrak{m})$ with a $1$-dimensional W-cell, any recoloring $\beta^{\uparrow}$ automatically specifies a minimally-degenerate fine-graining for the class. This simply follows from the fact that the W-cell $\mathcal{W}_{\mathfrak{m}}$ cannot be partitioned into smaller components, and there is therefore only one possible fine-grained min-cut structure $\check{\mathfrak{m}}$. Furthermore, since $\mathcal{W}_{\mathfrak{m}}=\mathcal{W}_{\check{\mathfrak{m}}}$, \cref{eq:mmdgfg} is trivially satisfied. Equivalently, in this case the class $(\tgm{\N},\mathfrak{m})$ has only a single representative $\gm{\N}$ (up to an irrelevant global rescaling of the weights), and the (unique) fine-graining of the class induced by $\beta^{\uparrow}$ necessarily coincides with the fine-graining of $\gm{\N}$.

Everything we have discussed thus far for fine-grainings of a class $(\tgm{\N},\mathfrak{m})$, can be iterated until one reaches the largest number of different colors that can be assigned to the boundary vertices of $\tgm{\N}$. This number is of course given by the cardinality of $\partial V$, and we denote it by $\sf{V}_{\partial}$. We define a \textit{maximal recoloring} of $\tgm{\N}$ as any recoloring that attains this bound, taking the form
\begin{equation}
\label{eq:maxrecoloring}
    \overline{\beta^{\uparrow}}:\; \partial V \rightarrow\; [\sf{V}_{\partial}]
\end{equation}
Given a class $(\tgm{\N},\mathfrak{m})$, any fine-graining of the form $(\overline{\beta^{\uparrow}},\mathfrak{m}')$, with $\mathfrak{m}'\in\mathfrak{m}\big\uparrow\!_{_{\N'}}$, will be referred to as a \textit{maximal fine-graining} of $(\tgm{\N},\mathfrak{m})$.

An interesting and useful application of maximal fine-grainings and \cref{lem:canocialfg} is the generalization of our analysis from \cref{subsec:simpletreepmi} concerning the relation between min-cut subspaces and PMIs for tree graphs. The main result of that section was \cref{thm:VPMIgeneral}, which says that for any simple tree and min-cut structure, the min-cut-subspace and the PMI coincide. Given an arbitrary topological graph model $\tgm{\N}$ with tree topology, we can now turn it into a simple tree via a maximal recoloring $\overline{\beta^{\uparrow}}$. And for any maximal recoloring, \cref{lem:canocialfg} guarantees the existence of a minimally-degenerate fine-graining such that \cref{eq:canonicalfg} applies. We have thus proved the following 
general result for arbitrary topological graph models with tree topology and min-cut structures on them

\begin{nthm}
\label{thm:nonsimpletree}
For any topological graph model $\tgm{\N}$ with tree topology, min-cut structure $\mathfrak{m}$ on $\tgm{\N}$, and minimally-degenerate maximal fine-graining $(\overline{\beta^{\uparrow}}, \check{\mathfrak{m}})$, the min-cut subspace $\mathbb{S}(\tgm{\N'},\check{\mathfrak{m}})$ is given by
\begin{equation}
   \mathbb{S}(\tgm{\N},\mathfrak{m})=\Phi_{\N'\rightarrow\N}\; \pi(\tgm{\N'},\check{\mathfrak{m}})
\end{equation}
\end{nthm}

This general result shows that for fixed $\N$ many\footnote{\, Notice that the trees in \cref{thm:nonsimpletree} are allowed to have an arbitrarily large number of vertices for each color, and therefore could in principle encode the same information content as graphs with more complicated topology.} min-cut subspaces can be seen as color-projections of PMIs for a larger number of parties, and it is interesting to ask if all min-cut subspaces can be obtained in this way. As we will discuss in \cref{sec:logic}, this question is of particular relevance for the reconstruction of the HEC from the solution to the HMIP.

We conclude this section with an example of fine-grainings in the particular case of an equivalence class whose min-cut subspace is $1$-dimensional, illustrating how even in this case there might exists alternative fine-grainings, as well as the interplay between the fine-graining procedure and the application of \cref{lem:key} and the reduction of \cref{ssec:vanishing}.

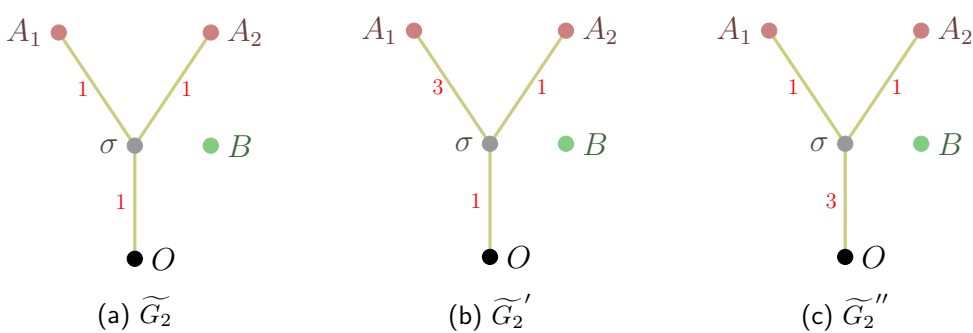
\begin{figure}[tb]
    \centering
    \begin{subfigure}{0.3\textwidth}
    \centering
    \begin{tikzpicture}

	\tikzmath{
	\edgelabelsize=0.7;
		\v =0.1;  
	}
	\coordinate (A1v) at (1,3.5);  
	\coordinate (A2v) at (3,3.5);  
	\coordinate (Bv) at (3,2);  
	\coordinate (Ov) at (2,0.5);  
	\coordinate (bv) at (2,2);  

	\draw[edgestyle] (A1v) -- node[scale=\edgelabelsize,edgeweightcolor,left,midway]{$1$} (bv);
	\draw[edgestyle] (Ov) -- node[scale=\edgelabelsize,edgeweightcolor,left,midway]{$1$} (bv);
	\draw[edgestyle] (A2v) -- node[scale=\edgelabelsize,edgeweightcolor,right,midway]{$1$} (bv);
	
	\filldraw [color=bvcolor] (bv) circle (\v) node[bvcolor!50!black,left=2pt]{$\sigma$} ;
	\filldraw [color=Acolor] (A1v) circle (\v) node[Acolor!50!black,left=2pt]{$A_1$} ;
	\filldraw [color=Acolor] (A2v) circle (\v) node[Acolor!50!black,right=2pt]{$A_2$} ;
	\filldraw [color=Bcolor] (Bv) circle (\v) node[Bcolor!50!black,right=2pt]{$B$} ;
	\filldraw [color=Ocolor] (Ov) circle (\v) node[Ocolor!50!black,right=2pt]{$O$} ;

\end{tikzpicture} 
\subcaption[]{$\gm{2}$}
\label{fig:Afgexample}
    \end{subfigure}
    \begin{subfigure}{0.3\textwidth}
    \centering
    \begin{tikzpicture}

	\tikzmath{
	\edgelabelsize=0.7;
		\v =0.1;  
	}
	\coordinate (A1v) at (1,3.5);  
	\coordinate (A2v) at (3,3.5);  
	\coordinate (Bv) at (3,2);  
	\coordinate (Ov) at (2,0.5);  
	\coordinate (bv) at (2,2);  

	\draw[edgestyle] (A1v) -- node[scale=\edgelabelsize,edgeweightcolor,left,midway]{$3$} (bv);
	\draw[edgestyle] (Ov) -- node[scale=\edgelabelsize,edgeweightcolor,left,midway]{$1$} (bv);
	\draw[edgestyle] (A2v) -- node[scale=\edgelabelsize,edgeweightcolor,right,midway]{$1$} (bv);
	
	\filldraw [color=bvcolor] (bv) circle (\v) node[bvcolor!50!black,left=2pt]{$\sigma$} ;
	\filldraw [color=Acolor] (A1v) circle (\v) node[Acolor!50!black,left=2pt]{$A_1$} ;
	\filldraw [color=Acolor] (A2v) circle (\v) node[Acolor!50!black,right=2pt]{$A_2$} ;
	\filldraw [color=Bcolor] (Bv) circle (\v) node[Bcolor!50!black,right=2pt]{$B$} ;
	\filldraw [color=Ocolor] (Ov) circle (\v) node[Ocolor!50!black,right=2pt]{$O$} ;

\end{tikzpicture}
\subcaption[]{$\gm{2}'$}
\label{fig:Bfgexample}
    \end{subfigure}
        \begin{subfigure}{0.3\textwidth}
    \centering
    \begin{tikzpicture}

	\tikzmath{
	\edgelabelsize=0.7;
		\v =0.1;  
	}
	\coordinate (A1v) at (1,3.5);  
	\coordinate (A2v) at (3,3.5);  
	\coordinate (Bv) at (3,2);  
	\coordinate (Ov) at (2,0.5);  
	\coordinate (bv) at (2,2);  

	\draw[edgestyle] (A1v) -- node[scale=\edgelabelsize,edgeweightcolor,left,midway]{$1$} (bv);
	\draw[edgestyle] (Ov) -- node[scale=\edgelabelsize,edgeweightcolor,left,midway]{$3$} (bv);
	\draw[edgestyle] (A2v) -- node[scale=\edgelabelsize,edgeweightcolor,right,midway]{$1$} (bv);
	
	\filldraw [color=bvcolor] (bv) circle (\v) node[bvcolor!50!black,left=2pt]{$\sigma$} ;
	\filldraw [color=Acolor] (A1v) circle (\v) node[Acolor!50!black,left=2pt]{$A_1$} ;
	\filldraw [color=Acolor] (A2v) circle (\v) node[Acolor!50!black,right=2pt]{$A_2$} ;
	\filldraw [color=Bcolor] (Bv) circle (\v) node[Bcolor!50!black,right=2pt]{$B$} ;
	\filldraw [color=Ocolor] (Ov) circle (\v) node[Ocolor!50!black,right=2pt]{$O$} ;

\end{tikzpicture} 
\subcaption[]{$\gm{2}''$}
\label{fig:Cfgexample}
    \end{subfigure}
    \caption{A choice of three graph models with the same underlying topology. $\gm{2}$ and $\gm{2}'$ have the same min-cut structure and are therefore different representatives of the same class, while $\gm{2}''$ has a different min-cut structure. Both min-cut structures are generic.}
    \label{fig:fgexample}
\end{figure}

Consider the two graph models in \cref{fig:Afgexample} and \cref{fig:Bfgexample}. Both of them have the same generic min-cut structure, which takes the form
\begin{equation}
    \mathscr{U}_{A} =\{\{A_1,A_2,\sigma\}\},\qquad \mathscr{U}_{B}=\{\{B\}\}, \qquad \mathscr{U}_{AB} =\{\{A_1,A_2,\sigma,B\}\}
\end{equation}
The W-cell for this min-cut structure is the interior of the following polyhedral cone in $\mathbb{R}^3_{>0}$
\begin{equation}
\label{eq:largeWcelleg}
    \text{cone}\,\{(1,0,0),(0,1,0),(1,0,1),(0,1,1)\}
\end{equation}
where we ordered the weights according to $(A_1\sigma,A_2\sigma,O\sigma)$. Its image under the map
\begin{equation}
    \Gamma=\left(
    \begin{tabular}{ccc}
        0 & 0 & 1 \\
        0 & 0 & 0 \\
        0 & 0 & 1 \\
    \end{tabular}
    \right)
\end{equation}
is the $1$-dimensional S-cell
\begin{equation}
\label{eq:fgexampleScell}
    \mathcal{S}=\lambda\, (1,0,1),\qquad \lambda>0
\end{equation}
which is indeed the entropy ray of both graph models $\gm{2}$ and $\gm{2}'$. 

The min-cut structure of the graph model in \cref{fig:Cfgexample} is again generic and takes the form 
\begin{equation}
    \mathscr{U}_{A} =\{\{A_1,A_2\}\},\qquad \mathscr{U}_{B}=\{\{B\}\}, \qquad \mathscr{U}_{AB} =\{\{A_1,A_2,B\}\}
\end{equation}
The W-cell is the interior of another polyhedral cone in $\mathbb{R}^3_{>0}$
\begin{equation}
\label{eq:smallWcelleg}
    \text{cone}\,\{(1,0,1),(0,1,1),(0,0,1)\}
\end{equation}
and its image under the map
\begin{equation}
    \Gamma=\left(
    \begin{tabular}{ccc}
        1 & 1 & 0 \\
        0 & 0 & 0 \\
        1 & 1 & 0 \\
    \end{tabular}
    \right)
\end{equation}
is the same $1$-dimensional S-cell given in \cref{eq:fgexampleScell} for the other two graph models.

\begin{figure}[t]
    \centering
        \begin{tikzpicture}

	\tikzmath{
	\edgelabelsize=0.7;
		\v =0.1;  
	}
	\coordinate (A) at (-3,0);  
	\coordinate (B) at (0,0); 
	\coordinate (C) at (3,0);  
	\coordinate (D) at (-1.5,2.598);
	\coordinate (E) at (1.5,2.598);
	\coordinate (F) at (0,5.196);
	\coordinate (G1) at (0,1.732);
	\coordinate (G2) at (0,3.464);
	\coordinate (G3) at (1.5,0.866);
	
	\filldraw[fill=orange,opacity=0.3] (A) -- (B) -- (E) -- (F) -- (D) -- cycle;
	\filldraw[fill=cyan,opacity=0.3] (B) -- (C) -- (E) -- cycle;

	\draw[black] (A) -- (B);
	\draw[black] (B) -- (C);
	\draw[black] (C) -- (E);
	\draw[black] (A) -- (D);
	\draw[black,dashed] (D) -- (E);
	\draw[black] (E) -- (F);
	\draw[black] (D) -- (F);
	\draw[black,dashed] (B) -- (D);
	\draw[black] (B) -- (E);
	
	\filldraw [color=black, fill=white, very thick] (A) circle (\v) node[black,left=2pt]{{\footnotesize $(0,1,0)$}} ;
	\filldraw [color=black] (B) circle (\v) node[black,below=2pt]{{\footnotesize $(0,1,1)$}} ;
	\filldraw [color=black, fill=white, very thick] (C) circle (\v) node[black,right=2pt]{{\footnotesize $(0,0,1)$}} ;
	\filldraw [color=gray, fill=white, very thick] (D) circle (\v) node[gray,left=2pt]{{\footnotesize $(1,1,0)$}} ;
	\filldraw [color=black] (E) circle (\v) node[black,right=2pt]{{\footnotesize $(1,0,1)$}} ;
	\filldraw [color=black, fill=white, very thick] (F) circle (\v) node[black,above=2pt]{{\footnotesize $(1,0,0)$}} ;
	\filldraw [color=orange] (G1) circle (\v) node[black,below=2pt]{{\footnotesize $\gm{2}$}} ;
	\filldraw [color=orange] (G2) circle (\v) node[black,above=2pt]{{\footnotesize $\gm{2}'$}} ;
	\filldraw [color=cyan] (G3) circle (\v) node[black,above=2pt]{{\footnotesize $\gm{2}''$}} ;

\end{tikzpicture}
    \caption{A cross section of $\mathbb{R}^3_{>0}$ by an affine $\mathbb{R}^2$ plane orthogonal to the vector $(1,1,1)$. This cross section shows the partition of $\mathbb{R}^3_{>0}$ and its boundary into W-cells for the topological graph model of the graphs in \cref{fig:fgexample}, before and after the recoloring from \cref{eq:recoloringexample}. See the main text for more details.}
    \label{fig:setWcelleg}
\end{figure}
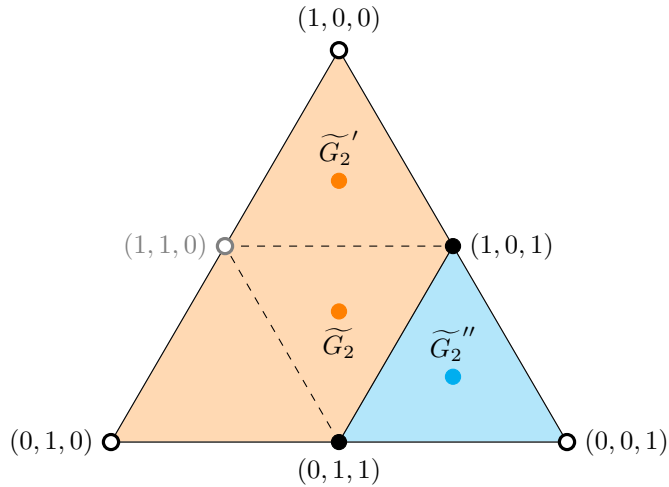

Figure \ref{fig:setWcelleg} shows the partition into W-cells of the space of edge weights $\mathbb{R}^3_{>0}$ for the topological graph model underlying the three graph models of \cref{fig:fgexample}. The regions shaded in orange and cyan correspond respectively to the two full-dimensional W-cells given in \cref{eq:largeWcelleg} and \cref{eq:smallWcelleg}. We leave it as an exercise for the reader to verify that the solid edge separating the these two W-cells does indeed correspond, as suggested by the figure, to a degenerate min-cut structure with a $1$-dimensional W-cell (see \cref{fig:Aegnewweights} for a choice of representative of this class).

The regions on the boundary of $\mathbb{R}^3_{>0}$ instead, correspond to W-cells for min-cut structures on different topological graph models obtained by appropriate deletion of edges, following the reduction described in \cref{ssec:vanishing}. The vertices shown in solid black in \cref{fig:setWcelleg} correspond to extreme rays shared by the two full-dimensional W-cells (see \cref{fig:Begnewweights} for an example). They are $1$-dimensional W-cells whose associated S-cells are again given by \cref{eq:fgexampleScell}, in agreement with \cref{lem:key}. On the other hand, the unfilled vertices corresponding to the canonical bases vectors of $\mathbb{R}^3$, are examples of the situation where \cref{lem:key} cannot be used. As explained in the proof of the same lemma, it can happen that an extreme ray of a W-cell is mapped to the null vector by the corresponding $\Gamma$ map, which is precisely what happens here (cf.\ \cref{fig:Cegnewweights}). All other solid edges in \cref{fig:setWcelleg} can also be checked to be $2$-dimensional W-cells, again after an appropriate deletion of edges (for now, the dashed edges and the $(1,1,0)$ vertex should be ignored, and the segment connecting $(0,1,0)$ to $(1,0,0)$ should be seen as a single face). Finally, notice that again in agreement with \cref{lem:key}, each W-cell (for all dimensions higher than $1$) has at least an extreme ray which is not mapped to the null vector by the corresponding $\Gamma$ matrix. The only exception seems to be the face generated by $\{(0,1,0),(1,0,0)\}$, however in this case even if the W-cell is $2$-dimensional, \cref{lem:key} does not apply, since the entire W-cell is mapped to the null vector and the min-cut subspace is not $1$-dimensional.

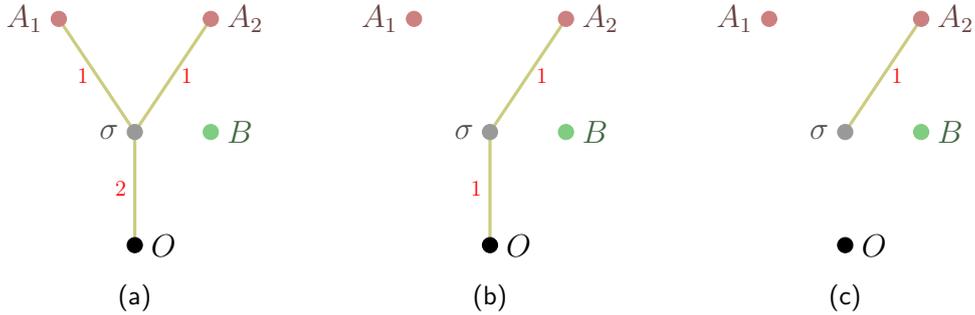
\begin{figure}[b]
    \centering
    \begin{subfigure}{0.3\textwidth}
    \centering
    \begin{tikzpicture}

	\tikzmath{
	\edgelabelsize=0.7;
		\v =0.1;  
	}
	\coordinate (A1v) at (1,3.5);  
	\coordinate (A2v) at (3,3.5);  
	\coordinate (Bv) at (3,2);  
	\coordinate (Ov) at (2,0.5);  
	\coordinate (bv) at (2,2);  

	\draw[edgestyle] (A1v) -- node[scale=\edgelabelsize,edgeweightcolor,left,midway]{$1$} (bv);
	\draw[edgestyle] (Ov) -- node[scale=\edgelabelsize,edgeweightcolor,left,midway]{$2$} (bv);
	\draw[edgestyle] (A2v) -- node[scale=\edgelabelsize,edgeweightcolor,right,midway]{$1$} (bv);
	
	\filldraw [color=bvcolor] (bv) circle (\v) node[bvcolor!50!black,left=2pt]{$\sigma$} ;
	\filldraw [color=Acolor] (A1v) circle (\v) node[Acolor!50!black,left=2pt]{$A_1$} ;
	\filldraw [color=Acolor] (A2v) circle (\v) node[Acolor!50!black,right=2pt]{$A_2$} ;
	\filldraw [color=Bcolor] (Bv) circle (\v) node[Bcolor!50!black,right=2pt]{$B$} ;
	\filldraw [color=Ocolor] (Ov) circle (\v) node[Ocolor!50!black,right=2pt]{$O$} ;

\end{tikzpicture} 
\subcaption[]{}
\label{fig:Aegnewweights}
    \end{subfigure}
    \begin{subfigure}{0.3\textwidth}
    \centering
    \begin{tikzpicture}

	\tikzmath{
	\edgelabelsize=0.7;
		\v =0.1;  
	}
	\coordinate (A1v) at (1,3.5);  
	\coordinate (A2v) at (3,3.5);  
	\coordinate (Bv) at (3,2);  
	\coordinate (Ov) at (2,0.5);  
	\coordinate (bv) at (2,2);  

	\draw[edgestyle] (Ov) -- node[scale=\edgelabelsize,edgeweightcolor,left,midway]{$1$} (bv);
	\draw[edgestyle] (A2v) -- node[scale=\edgelabelsize,edgeweightcolor,right,midway]{$1$} (bv);
	
	\filldraw [color=bvcolor] (bv) circle (\v) node[bvcolor!50!black,left=2pt]{$\sigma$} ;
	\filldraw [color=Acolor] (A1v) circle (\v) node[Acolor!50!black,left=2pt]{$A_1$} ;
	\filldraw [color=Acolor] (A2v) circle (\v) node[Acolor!50!black,right=2pt]{$A_2$} ;
	\filldraw [color=Bcolor] (Bv) circle (\v) node[Bcolor!50!black,right=2pt]{$B$} ;
	\filldraw [color=Ocolor] (Ov) circle (\v) node[Ocolor!50!black,right=2pt]{$O$} ;

\end{tikzpicture}
\subcaption[]{}
\label{fig:Begnewweights}
    \end{subfigure}
        \begin{subfigure}{0.3\textwidth}
    \centering
    \begin{tikzpicture}

	\tikzmath{
	\edgelabelsize=0.7;
		\v =0.1;  
	}
	\coordinate (A1v) at (1,3.5);  
	\coordinate (A2v) at (3,3.5);  
	\coordinate (Bv) at (3,2);  
	\coordinate (Ov) at (2,0.5);  
	\coordinate (bv) at (2,2);  

	\draw[edgestyle] (A2v) -- node[scale=\edgelabelsize,edgeweightcolor,right,midway]{$1$} (bv);
	
	\filldraw [color=bvcolor] (bv) circle (\v) node[bvcolor!50!black,left=2pt]{$\sigma$} ;
	\filldraw [color=Acolor] (A1v) circle (\v) node[Acolor!50!black,left=2pt]{$A_1$} ;
	\filldraw [color=Acolor] (A2v) circle (\v) node[Acolor!50!black,right=2pt]{$A_2$} ;
	\filldraw [color=Bcolor] (Bv) circle (\v) node[Bcolor!50!black,right=2pt]{$B$} ;
	\filldraw [color=Ocolor] (Ov) circle (\v) node[Ocolor!50!black,right=2pt]{$O$} ;

\end{tikzpicture} 
\subcaption[]{}
\label{fig:Cegnewweights}
    \end{subfigure}
    \caption{Graph models specified by particular choices of weights discussed in the main text.}
    \label{fig:egnewweights}
\end{figure}

Let us now consider the maximal recoloring of the graphs models in \cref{fig:fgexample} specified by 
\begin{equation}
\label{eq:recoloringexample}
    A_1\rightarrow A,\quad A_2\rightarrow C,\quad B\rightarrow B,\quad O\rightarrow O
\end{equation}
The new min-cut structures corresponding to the recolored graph models are still generic, and can conveniently be described by their corresponding $\Gamma$ matrices. They are respectively
\begin{equation}
    \Gamma_{\gm{2}}=\left(
    \begin{tabular}{ccc}
        1 & 0 & 0 \\
        0 & 0 & 0 \\
        0 & 1 & 0 \\
        1 & 0 & 0 \\
        0 & 0 & 1 \\
        0 & 1 & 0 \\
        0 & 0 & 1 \\
    \end{tabular}
    \right)\qquad
    \Gamma_{\gm{2}'}=\left(
    \begin{tabular}{ccc}
        0 & 1 & 1 \\
        0 & 0 & 0 \\
        0 & 1 & 0 \\
        0 & 1 & 1 \\
        0 & 0 & 1 \\
        0 & 1 & 0 \\
        0 & 0 & 1 \\
    \end{tabular}
    \right)\qquad
    \Gamma_{\gm{2}''}=\left(
    \begin{tabular}{ccc}
        1 & 0 & 0 \\
        0 & 0 & 0 \\
        0 & 1 & 0 \\
        1 & 0 & 0 \\
        1 & 1 & 0 \\
        0 & 1 & 0 \\
        1 & 1 & 0 \\
    \end{tabular}
    \right)
\end{equation}
As one can immediately see, after recoloring, the graph model in \cref{fig:Afgexample} has a $3$-dimensional min-cut subspace, while the min-cut subspaces of the graph models in \cref{fig:Bfgexample} are $2$-dimensional, even if the two original graph models belonged to the same W-cell and had a $1$-dimensional min-cut subspace. 

The W-cells of the first two fine-grained graph models are
\begin{align}
    &\mathcal{W}_{\gm{2}} = \text{cone}\,\{(1,1,0),(0,1,1),(1,0,1)\}\nonumber\\
    &\mathcal{W}_{\gm{2}'} = \text{cone}\,\{(1,1,0),(1,0,0),(1,0,1)\}
\end{align}
These new W-cells are again shown in \cref{fig:setWcelleg}, where the dashed edges, and the new vertex $(1,1,0)$, are new W-cells that correspond to new min-cut structures that emerge from the recoloring of the original topological graph model. As one can easily guess from the figure, there is an additional full-dimensional W-cell
\begin{equation}
    \text{cone}\,\{(0,1,0),(1,1,0),(0,1,1)\}
\end{equation}
which can immediately be obtained by swapping $A$ and $C$ in the recolored graph model of \cref{fig:Bfgexample}. Combined, all these W-cells form a partition of the original W-cell for the graph models in \cref{fig:Afgexample} and \cref{fig:Bfgexample} given in \cref{eq:largeWcelleg}, in agreement with \cref{lem:Wrefinementfg}. On the other hand, as also shown in \cref{fig:setWcelleg}, the recoloring of the graph model in \cref{fig:Cfgexample} is not associated to any partition of the original W-cell. The full set of W-cells after the recoloring of the topological graph model can also be seen in \cref{fig:flow}, where we explicitly show the transition between different W-cells as a function of two independent weights (the third can be fixed by rescaling). Notice that there is an additional W-cell that is not visible in \cref{fig:triangle}, namely the $0$-dimensional W-cell corresponding to the origin of $\mathbb{R}^3$.
 
Finally, let us briefly comment again on the application of \cref{lem:key} to see how it relates to fine-grainings. Starting from $\gm{2}$ we can apply \cref{lem:key} and select an extreme ray of the closure of the W-cell given in \cref{eq:largeWcelleg} such that the min-cut subspace of the new graph model is again the one generated by the S-cell given in \cref{eq:fgexampleScell}. Suppose that we choose $(0,1,1)$ (cf.\ \cref{fig:Cegnewweights}) and then apply the recoloring from \cref{eq:recoloringexample}. The W-cell of the resulting graph model is the same as the original one, and the new min-cut structure is degenerate. With the choice of representative $\mathscr{U}_{C}=\{\{C\}\}$, the $\Gamma$ matrix is
\begin{equation}
    \Gamma=\left(
    \begin{tabular}{ccc}
        0 & 0 & 0 \\
        0 & 0 & 0 \\
        0 & 1 & 0 \\
        0 & 0 & 0 \\
        0 & 1 & 0 \\
        0 & 1 & 0 \\
        0 & 1 & 0 \\
    \end{tabular}
    \right)
\end{equation}
and the (necessarily $1$-dimensional) S-cell is
\begin{equation}
\label{eq:fgScellexample}
\mathcal{S}=\lambda\, (0,0,1,0,1,1,1),\qquad \lambda>0
\end{equation}
The coarse-graining associated to the initial coloring of the graph (before the recoloring that induced the fine-graining) specifies the color projection (cf.\ \cref{eq:colorproj})
\begin{equation}
\label{eq:colorprojexample}
\Phi_{3\rightarrow2}=\left(
    \begin{tabular}{ccccccc}
        0 & 0 & 0 & 0 & 1 & 0 & 0 \\
        0 & 1 & 0 & 0 & 0 & 0 & 0 \\
        0 & 0 & 0 & 0 & 0 & 0 & 1 \\
    \end{tabular}
    \right)
\end{equation}
and one can immediately verify that once applied to \cref{eq:fgScellexample} this transformation gives the original S-cell from \cref{eq:fgexampleScell}.


\begin{figure}[tb]
    \centering
    \begin{subfigure}{0.49\textwidth}
    \centering
    \begin{tikzpicture}

	\tikzmath{
	\edgelabelsize=0.7;
		\v =0.1;  
	}
	\coordinate (A) at (-6,0);  
	\coordinate (B) at (-6,1); 
	\coordinate (C) at (-6,6);  
	\coordinate (D) at (-5,0);
	\coordinate (E) at (0,0);
	\coordinate (F) at (-4,3);
	\coordinate (G) at (-3,2);
	
	
	\filldraw [color=white, fill=white] (E) circle (\v) node[black,below=2pt]{{\footnotesize $\lambda_1$}} ;
	\filldraw [color=white, fill=white] (C) circle (\v) node[black,left=2pt]{{\footnotesize $\lambda_2$}} ;

	\draw[->] (A) -- (C);
	\draw[->] (A) -- (E);
	\draw[-] (B) -- (F);
	\draw[-] (D) -- (G);
	\draw[-] (B) -- (D);
	
	\filldraw [color=black, fill=black] (B) circle (\v) node[black,left=2pt]{{\footnotesize $1$}} ;
	\filldraw [color=black, fill=black] (D) circle (\v) node[black,below=2pt]{{\footnotesize $1$}} ;
	\filldraw [color=black, fill=white, very thick] (A) circle (\v) node[black,above=2pt]{} ;
	\node[rotate=45] at (-4.5,2.8) {\footnotesize $\lambda_2=\lambda_1+1$};
	\node[rotate=45] at (-3.5,1.8) {\footnotesize $\lambda_1=\lambda_2+1$};
	\node[] at (-0.85,0.55) {\footnotesize $\lambda_1+\lambda_2=1$};
	\draw[<-,dotted] (-5.45,0.55) -- (-1.85,0.55);

    \filldraw[fill=gray,opacity=0.1,draw=none] (-3,6) -- (0,6) -- (0,3) -- (-3,3) -- cycle;

	\coordinate (a) at (-2.5,5.5);  
	\coordinate (b) at (-0.5,4); 
	\coordinate (c) at (-0.5,5.5);  
	\coordinate (o) at (-1.5,3.5);
	\coordinate (s) at (-1.5,4.5);
	\draw[-] (a) -- (s);
	\draw[-] (c) -- (s);
	\draw[-] (o) -- (s);
	\node[align=left] at (-2.3,5) {\footnotesize $\lambda_1$};
	\node[align=left] at (-0.6,5) {\footnotesize $\lambda_2$};
	\node[align=left] at (-2.6,5.7) {\footnotesize $A$};
	\node[align=left] at (-0.4,4) {\footnotesize $B$};
	\node[align=left] at (-0.4,5.7) {\footnotesize $C$};
	\node[align=left] at (-1.5,3.3) {\footnotesize $O$};
	\node[align=left] at (-1.7,4) {\footnotesize $1$};
	\node at (-0.6,4) [circle,fill=black,inner sep=0.5pt]{};
	\node[align=left] at (-1.3,4.4) {\footnotesize $\sigma$};

\end{tikzpicture} 
\subcaption[]{}
    \end{subfigure}
    \begin{subfigure}{0.49\textwidth}
    \centering
    \begin{tikzpicture}

	\tikzmath{
	\edgelabelsize=0.7;
		\v =0.1;  
	}
	\coordinate (A) at (-6,0);  
	\coordinate (B) at (-6,6);  
	\coordinate (C) at (0,0);
	\coordinate (E) at (-4,2);
	

\filldraw [color=white, fill=white] (C) circle (\v) node[black,below=2pt]{{\footnotesize $\lambda_1$}} ;
\filldraw [color=white, fill=white] (B) circle (\v) node[black,left=2pt]{{\footnotesize $\lambda_2$}} ;
\node[rotate=45] at (-4.5,1.8) {\footnotesize $\lambda_1=\lambda_2$};

	\draw[->] (A) -- (B);
	\draw[->] (A) -- (C);
	\draw[-] (A) -- (E);

	\filldraw [color=black, fill=white, very thick] (A) circle (\v) node[black,above=2pt]{} ;

    \filldraw[fill=gray,opacity=0.1,draw=none] (-3,6) -- (0,6) -- (0,3) -- (-3,3) -- cycle;

	\coordinate (a) at (-2.5,5.5);  
	\coordinate (b) at (-0.5,4); 
	\coordinate (c) at (-0.5,5.5);  
	\coordinate (o) at (-1.5,3.5);
	\coordinate (s) at (-1.5,4.5);
	\draw[-] (a) -- (s);
	\draw[-] (c) -- (s);
	\node[align=left] at (-2.3,5) {\footnotesize $\lambda_1$};
	\node[align=left] at (-0.6,5) {\footnotesize $\lambda_2$};
	\node[align=left] at (-2.6,5.7) {\footnotesize $A$};
	\node[align=left] at (-0.4,4) {\footnotesize $B$};
	\node[align=left] at (-0.4,5.7) {\footnotesize $C$};
	\node[align=left] at (-1.5,3.3) {\footnotesize $O$};
	\node at (-0.6,4) [circle,fill=black,inner sep=0.5pt]{};
	\node at (-1.5,3.5) [circle,fill=black,inner sep=0.5pt]{};
	\node[align=left] at (-1.3,4.4) {\footnotesize $\sigma$};

\end{tikzpicture}
\subcaption[]{}
    \end{subfigure}
    \caption{The W-cells for the topological graph model underlying the graphs in \cref{fig:fgexample}, as well as variations of it obtained via the deletion of the appropriate edges. The figure shows the explicit dependence on the values of the weights as well as the degeneracy equations. The shaded insets show the topological graph models associated to the two diagrams, which are distinguished by the presence (left) or absence (right) of the $(\sigma,O)$ edge.}
    \label{fig:flow}
\end{figure}
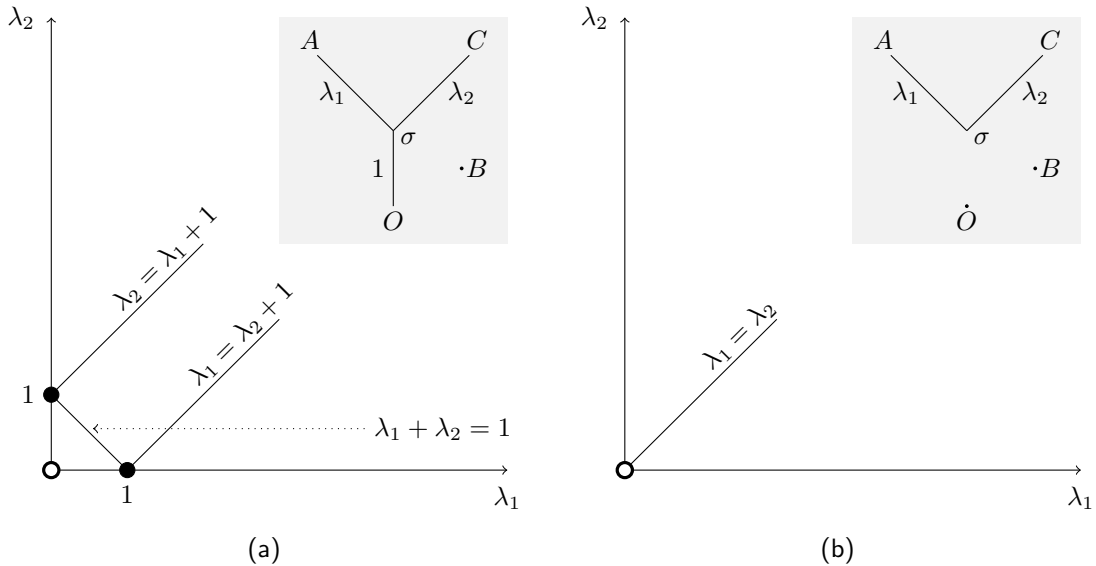

\section{The HEC from marginal independence}
\label{sec:logic}

In this section we will use the machinery that we have developed thus far to relate the construction of the HEC$_{\N}$ to the knowledge of set of holographic PMIs. 

In \cref{subsec:extreme_rays} we will first look at various graph models which realize the known extreme rays of the HEC$_{\N}$ up to $\N=6$, focusing in particular on their PMIs. The reader should not be surprised by the fact that we will look directly at graph models, rather than equivalence classes, since by \cref{cor:ext} each of the graph models discussed here is by itself an equivalence class (up to a global rescaling of the weights). In these cases, an explicit choice of weights should merely be seen as a more compact way of specifying a min-cut structure on the underlying topological graph model. 

Based on the evidence presented in \cref{subsec:extreme_rays}, in \cref{subsec:reconstruction} we will then formulate various conjectures about certain properties of the graph models that realize the extreme rays of the HEC$_{\N}$ for arbitrary $\N$. We will discuss in detail how these conjectures are related to each other, and their implications for the explicit reconstruction of the HEC$_{\N}$ from the set of holographic PMIs.

In most of the following discussions about extreme rays of the HEC$_{\N}$, including the formulation of our conjectures (in particular \ref{con:c1} and \ref{con:c4}), we will always implicitly assume that we are focusing only on the restricted set of ``genuine $\N$-party'' extreme rays, i.e., on extreme rays which are not embeddings into $\N$-party entropy space of extreme rays for fewer parties. The convenience of this assumption lies in the fact that the components of such extreme rays are strictly positive, which allows us to ignore unnecessary complications related to the connectivity of the graphs. Indeed, if some components of an $\N$-party entropy vector $\mathbf{S}$ vanish, it is always possible to use purity\footnote{\, Recall that the von Neumann entropy of a density matrix vanishes if and only if it corresponds to a pure state. Similarly, for graphs, the entropy of a subsystem $\I$ vanishes if and only if the min-cut for $\I$ is disconnected from its complement.} to effectively reduce the number of parties to some $\N'<\N$, and distill the essential information contained in $\mathbf{S}$ into a new $\N'$-party entropy vector $\mathbf{S}'$.

\subsection{Graph models and PMIs for known extreme rays of the HEC}
\label{subsec:extreme_rays}

In \cref{subsec:properties} we proved that any graph model that realizes an extreme ray of the HEC$_{\N}$ has a $1$-dimensional min-cut subspace, and we have shown an explicit example for $\N=3$. In that example, the topological graph model $\tgm{3}$ associated to $\gm{3}$ was a simple tree graph (cf. \cref{def:simple}). By \cref{thm:VPMIgeneral} it follows then that its PMI is also $1$-dimensional (it is the min-cut subspace), and that the ray in \cref{eq:perfectray} is also an extreme ray of the SAC$_3$. As it turns out, the same type of situation also occurs for all graph models which realize the extreme rays of the HEC$_{\N}$, for $\N\in\{2,3,4\}$ \cite{Bao:2015bfa}, as the reader can easily verify.

On the other hand, for $\N=5$, not all realizations of the extreme rays given in \cite{HernandezCuenca:2019wgh,Avis:2021xnz} were trees. The key observation of this subsection is that one can nonetheless find alternative graph models, which are trees, and realize the same extreme rays. These trees however are no longer simple and this non-simplicity is not just a consequence of particular choices of graph realizations, but rather of the fact that the corresponding extreme rays of the HEC$_5$ are not extreme rays of the SAC$_5$ (their PMIs are higher dimensional, see the next subsection for more details). 

In \cref{fig:N5trees}, we provide a tree graph realization for every HEC$_5$ extreme ray which was not realized by a tree graph in \cite{Avis:2021xnz}. Most of these new trees can immediately be obtained from the original graph models from \cite{Avis:2021xnz} by simply splitting every degree-$k$ boundary vertex into $k$ vertices of the same color.\footnote{\, This is not always sufficient: for example the last graph in \cref{fig:N5trees} requires a judicious iterative application of the operations described in \cref{sec:gops}, particularly the ${\sf \Delta}$-{\sf Y} exchange operation shown in \cref{fig:triangle}.} For each of these tree graph models we can then specify a maximal recoloring and consider the associated possible fine-grainings of the corresponding class. As exemplified in the previous section, one possible such fine-gaining is simply obtained via the fine graining of the graph model, i.e., by only applying the recoloring to the boundary vertices, without varying the weights. Despite the fact that in general the dimension of the min-cut subspace can grow after a fine-graining, it turns out that in each case the resulting min-cut subspace is $1$-dimensional.\footnote{\, At this point the reader may already wonder whether we could have used \cref{lem:key} in case the resulting min-cut subspace had a higher dimensionality. Indeed this is the case, and we will use this strategy in what follows.} Since each of these new graph models is now specified on a simple tree (because the recoloring was maximal), it then follows by \cref{thm:VPMIgeneral} that the PMI is also $1$-dimensional, and these graph models thus realize extreme rays of the SAC$_{\N'}$ for the corresponding $\N'={\sf V}_{\partial}-1\geq5$. For instance, via this procedure, the last graph model in \cref{fig:N5trees} lifts to a simple tree graph model for $\N'=10$, which one can check realizes an extreme ray of the SAC$_{10}$.\footnote{\, To check extremality of an entropy vector for $\N$ parties, one can simply check that it satisfies all instances of SA for that $\N$, and that at least $\D-1$ linearly independent instances are saturated.} Notice that even if the fine grainings specified here may\footnote{\, We have not explicitly checked the dimensionality of the W-cells, since it is not crucial here.} not be minimally-degenerate according to \cref{def:mmdgfg}, the min-cut subspaces before and after the fine-graining are still related by a color-projection, since both S-cells are just single rays.

\begin{figure}[tb] 
\centering
\includegraphics[scale=0.655]{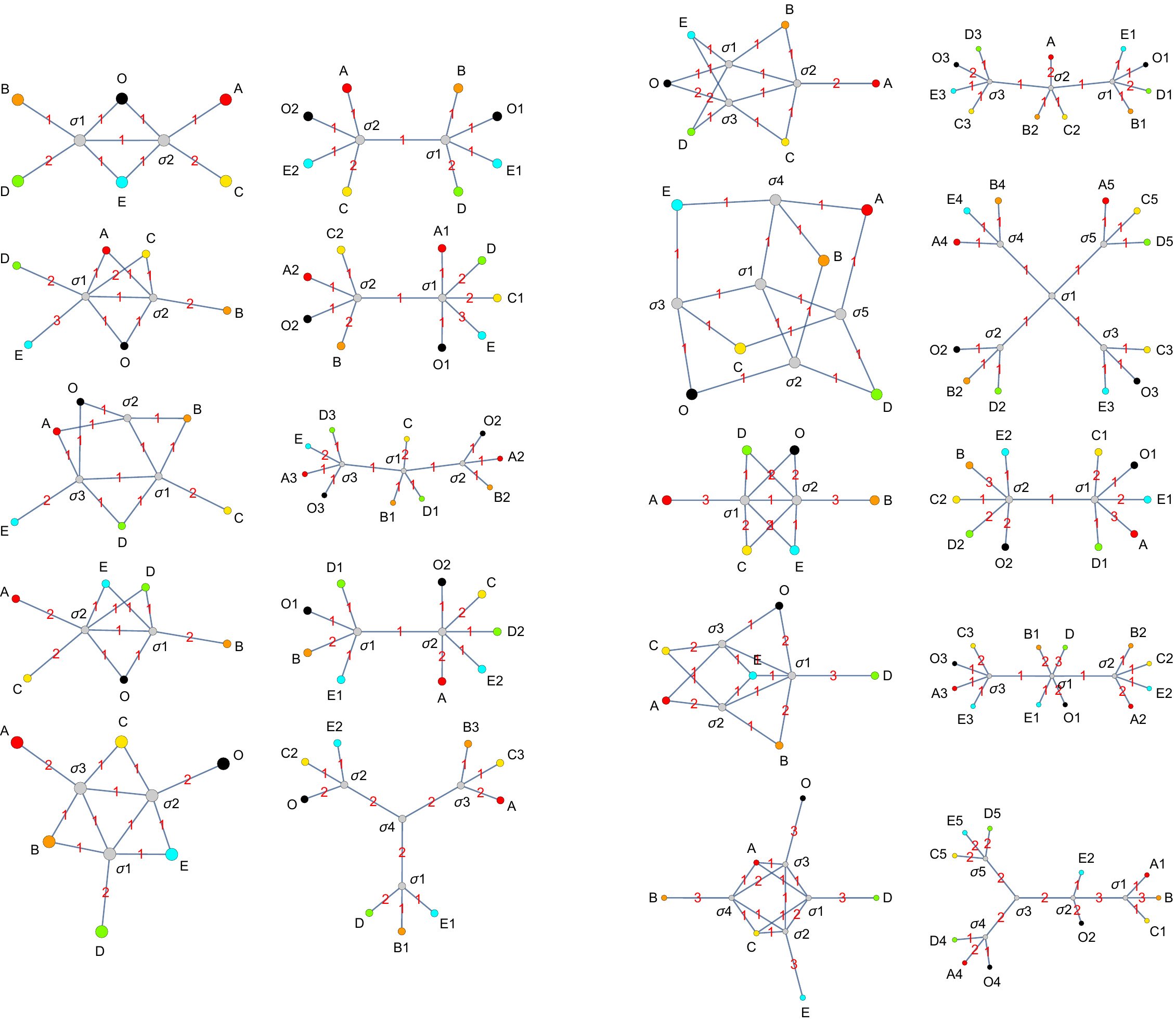} \\
\caption{The first and third columns show graph realizations of the extreme rays of the HEC$_5$ which cannot be realized by simple trees. Rather than their original form from \cite{HernandezCuenca:2019wgh}, here we show the graphs of minimal number of vertices obtained by \cite{Avis:2021xnz}. To the right of each graph, we provide a non-simple tree realization of the same extreme ray, obtained through repeated application of the graph operations described in \cref{sec:gops}.}
\label{fig:N5trees}
\end{figure} 

For $\N\geq6$, the HEC$_{\N}$ is still unknown, but a large number of its extreme rays is known for $\N=6$ \cite{n6wip}, and we can again explore whether we can realize them with graph models with tree topology. If we succeed, we can then repeat the same procedure we just described for the $\N=5$ case, to make these tree graph models simple via a maximal recoloring, and check if the resulting min-cut subspace is $1$-dimensional. If this is the case, these graph models then realize extreme rays of the SAC$_{\N'}$ for some $\N'\geq 6$. 

At the time of writing, a total of $4122$ distinct (orbits of) extreme rays have been discovered. Although only $24$ of these can be realized by simple trees, it turns out that as many as $3905$ of the others can be realized by graph models that can be immediately turned into non-simple trees by just splitting their boundary vertices (as explained before for the $\N=5$ case). In other words, only $193$ extreme rays are realized by graph models that contain cycles involving only bulk vertices, which can thus not be broken by splitting boundary vertices. However, most of these just contain a single bulk $3$-cycle, which is straightforward to break as we show in \cref{sec:gops}. Ultimately, there only remain $14$ graphs models which cannot be obviously turned into trees, as they involve larger cycles or more than a single $3$-cycle. Figures \ref{fig:N6er1} and \ref{fig:N6er3} depict graph models realizations containing a bulk $4$-cycle for two of the HEC$_6$ extreme rays. After splitting the boundary vertices and maximally recoloring them, the topological graph models are the same for these two cases.  Nevertheless, it is easy to convert the graph of \cref{fig:N6er1} into a tree graph (by splitting the bulk cycle at the $\sigma2$ vertex; cf.\ the bottom panel), whereas an analogous transformation does not work for the graph of \cref{fig:N6er3}, even though the original graph looks a bit simpler. For all simple tree graph models obtained via this procedure, the resulting min-cut subspace (and therefore PMI) is indeed $1$-dimensional.

In total there are only a few (some of the $14$ mentioned above) extreme rays of the HEC$_6$ which are realized by graph models that we were not immediately able to convert into tree graphs using these simple operations. While it is possible that some of these extreme rays cannot be realized by any tree graph, we have not attempted a systematic and exhaustive search, and it seems likely that the difficulty in obtaining tree realizations for these extreme rays is just technical rather than fundamental.

\begin{figure}[tb] 
\begin{center}
\includegraphics[width=2.4in]{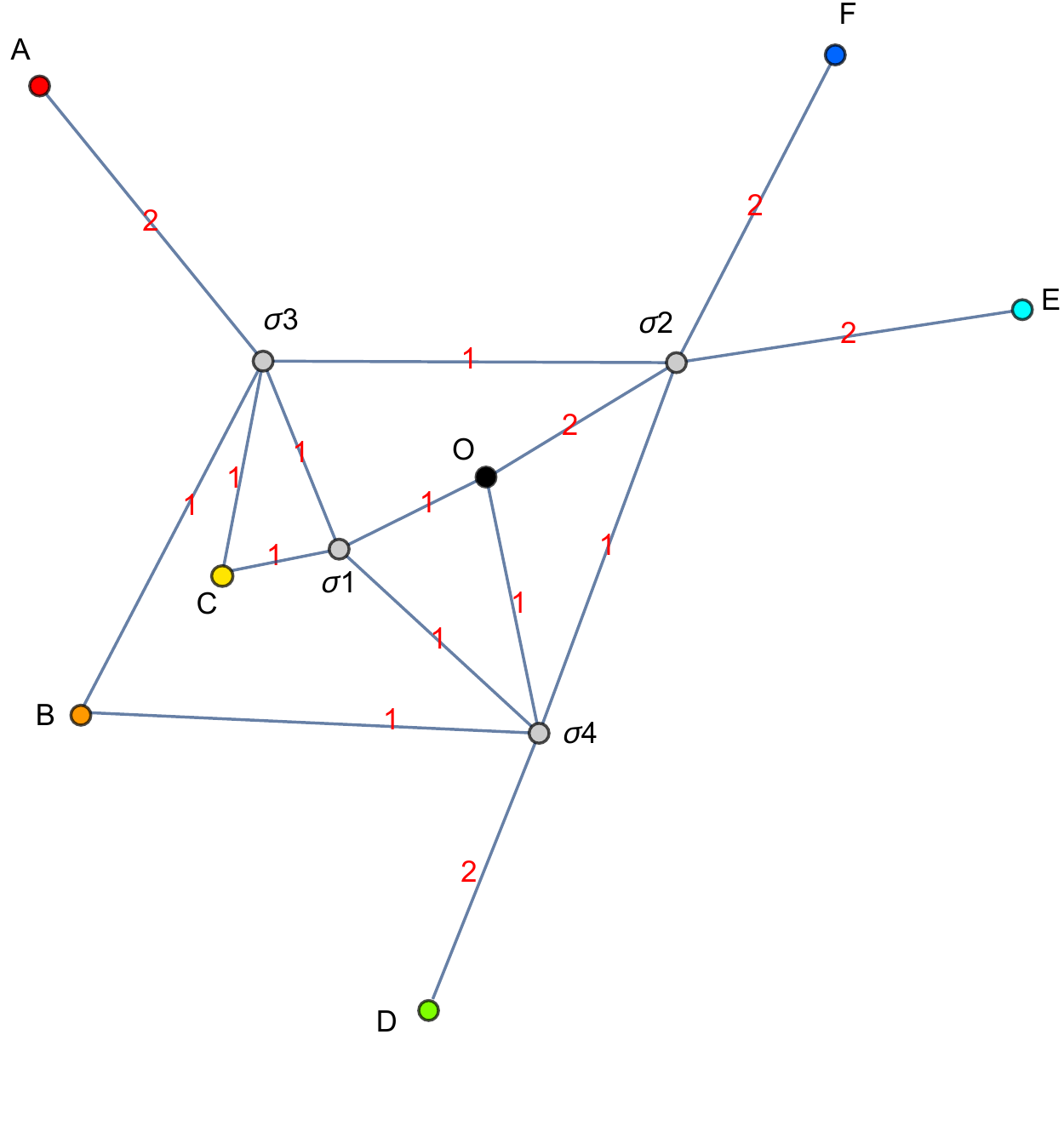}
\hspace{1cm}
\includegraphics[width=2.4in]{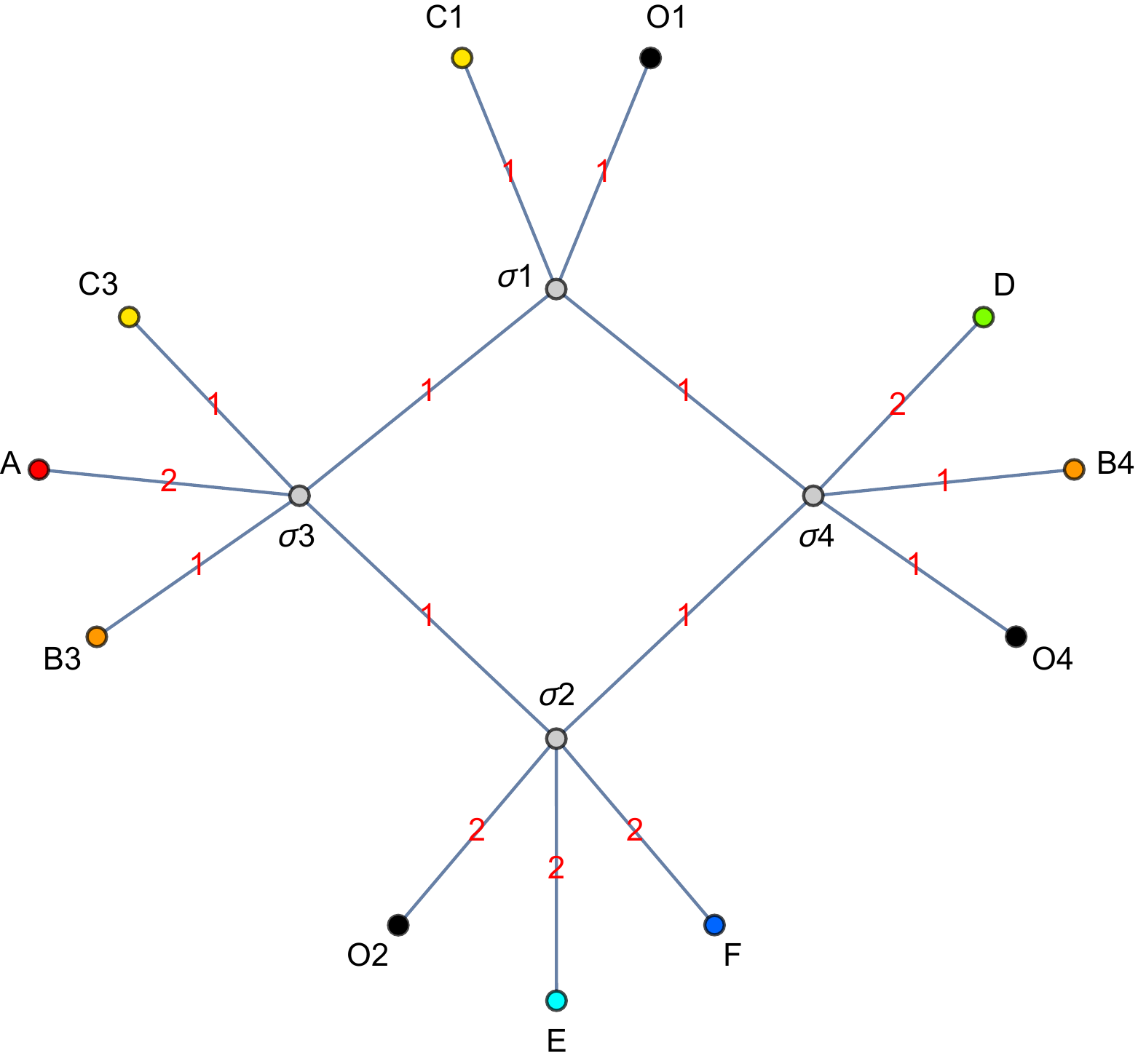}
\includegraphics[width=2.4in]{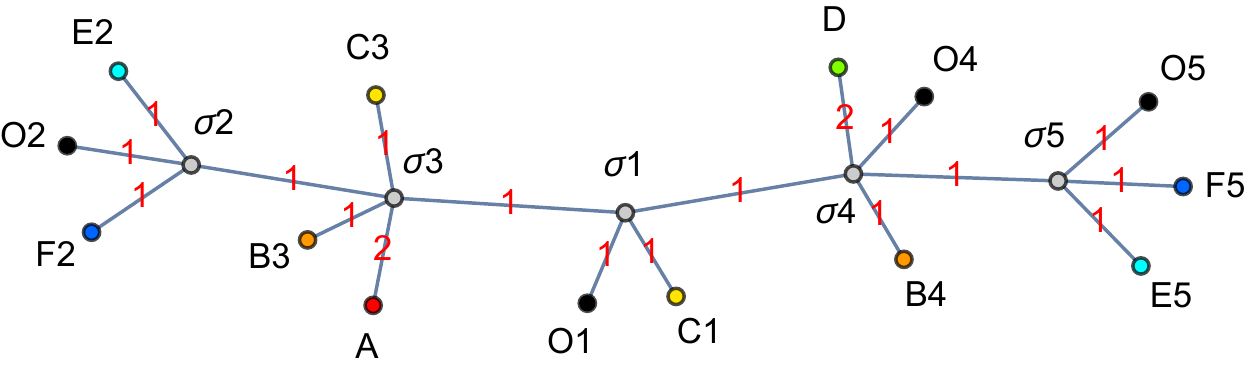}
\caption{Alternative graph realizations of one of the extreme rays of the HEC$_6$. Notice that one of these realizations is a tree graph.}
\label{fig:N6er1}
\end{center}
\end{figure} 

\begin{figure}[tb] 
\begin{center}
\includegraphics[width=2.4in]{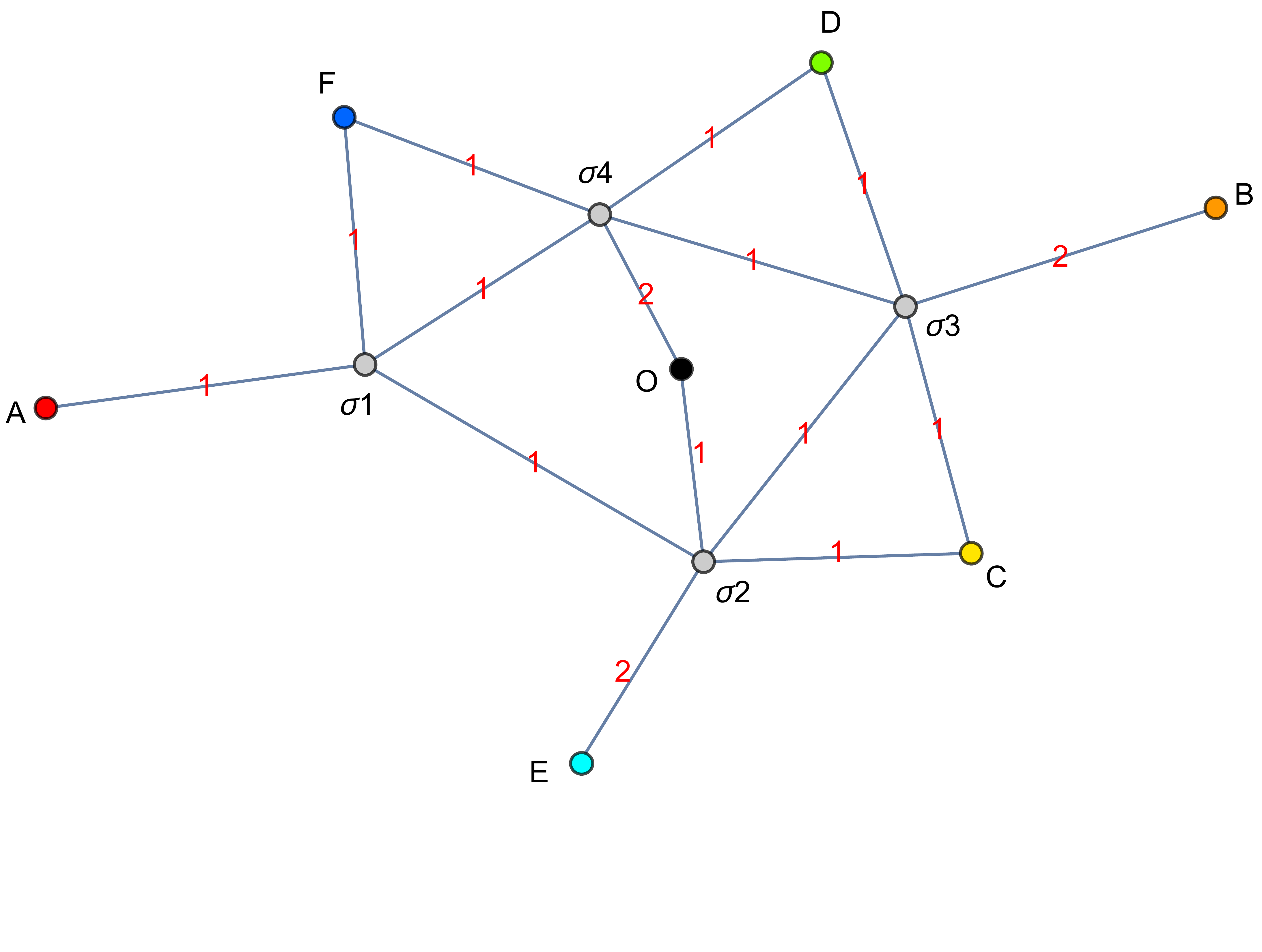}
\hspace{1cm}
\includegraphics[width=2.4in]{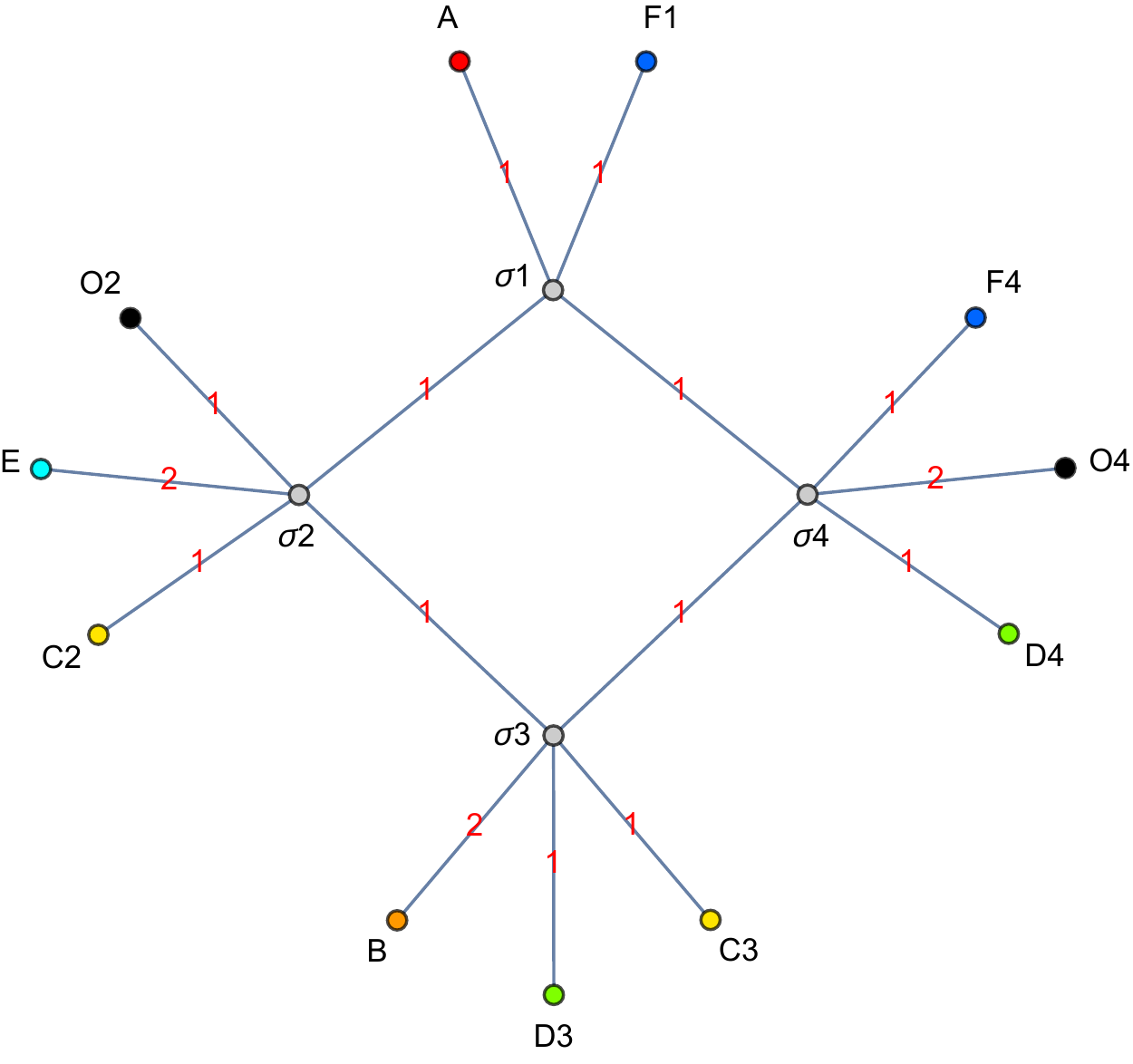}
\caption{Alternative graph realizations of another extreme ray of the HEC$_6$. Notice that the graph on the right has the same topology as the second graph in \cref{fig:N6er1}.}
\label{fig:N6er3}
\end{center}
\end{figure} 

\subsection{Reconstruction of the HEC from marginal independence}
\label{subsec:reconstruction}

We are now ready to discuss the main claim of this work, namely how the reconstruction of the holographic entropy cone is related to the solution of the holographic marginal independence problem. We will not be able to provide a definite proof that such a reconstruction is possible. However, motivated by the observations presented in the previous subsection, we will formulate a series of conjectures each of which will imply that this is indeed the case. We will organize these conjectures following a hierarchy, from the weakest to the strongest, analysing their relations and implications. As we will see, all of these conjectures, except only for the weakest \ref{con:c1} (see below), will imply that this reconstruction has a particularly nice form and is intimately related to $1$-dimensional holographic PMIs, i.e., holographic extreme rays of the SAC.

Let us begin by explaining more explicitly how the \textit{reconstruction} would work. Since the HEC is a polyhedral cone, it can equivalently be described by providing either the full set of its extreme rays, or the full set of its facets, specified by non-redundant entropy inequalities. Given one description, one can in principle\footnote{\, There is no known efficient algorithm to perform this conversion, and in practice this is undoable already for $\N=6$.} obtain the other using well-known conversion algorithms \cite{Fukuda:2015}, and we will focus on the description in terms of extreme rays. Suppose now that we are interested in the HEC$_{\N^*}$ for a specific number of parties $\N^*$, and we are given the full solution to the HMIP, i.e., we are given the full set of all holographic PMIs (recall \cref{def:hologpmi}) for all possible values of $\N$, can we derive all the extreme rays of the HEC$_{\N^*}$ from this data?

As a warm up, let us first discuss the simpler situation where $\N^*\leq 4$. Suppose that we do not know the extreme rays of the HEC$_{\N^*}$, but instead we want to try to extract them from the solution to the HMIP for $\N=\N^*$. Such solution is the collection of all PMIs that can be realized holographically for $\N=\N^*$, and it will contain PMIs of different dimensions $1\leq d\leq\D^*$. For each $1$-dimensional PMI it is straightforward to determine an entropy ray that generates it, as one simply needs to pick a ray oriented towards the positive orthant of entropy space. Furthermore, any such ray is automatically an extreme ray of the HEC$_{\N^*}$ because it is a holographic extreme ray of the SAC$_{\N^*}$.\footnote{\, Notice that any class $(\tgm{\N},\mathfrak{m})$ realizing this ray (which exists by assumption) is guaranteed by \cref{lem:vinpmi} to have a $1$-dimensional min-cut subspace, in agreement with \cref{cor:ext}.} On the other hand, we will ignore all higher dimensional PMIs, since it is not possible to uniquely determine an entropy ray out of them. We can then build a candidate for the HEC$_{\N^*}$ by taking the conical hull of the rays thus obtained. The equivalence of the HEC$_{\N^*}$ with this cone however is just a coincidence related to the fact that we are only considering $\N^*\leq 4$, and as we reviewed in the previous subsection, in this case all extreme rays of the HEC$_{\N^*}$ correspond to $1$-dimensional PMIs.

Let us now focus on $\N^*\geq 5$, and consider an extreme ray $\mathbf{S}$ of the HEC$_{\N^*}$ which is \emph{not} an extreme ray of the SAC$_{\N^*}$, and a class $(\tgm{\N^*},\mathfrak{m})$ realizing it, i.e., such that
\begin{equation}
    \mathbb{S}(\tgm{\N^*},\mathfrak{m})=\mathbb{S}
\end{equation}
where $\mathbb{S}$ is the subspace generated by $\mathbf{S}$. Suppose that there exists a minimally-degenerate fine-graining (cf.\  \cref{def:mmdgfg}) to a class $(\tgm{\N},\check{\mathfrak{m}})$, with $\N>\N^*$ and $\check{\mathfrak{m}}\in\mathfrak{m}\big\uparrow\!_{_{\N}}$, such that the min-cut subspace and the PMI coincide. Denoting by $\Phi_{\N\rightarrow\N^*}$ the corresponding coarse-graining we then have
\begin{equation}
\label{eq:pmiprojection}
    \mathbb{S}=\Phi_{\N\rightarrow\N^*}\,\mathbb{S}(\tgm{\N},\check{\mathfrak{m}})=\Phi_{\N\rightarrow\N^*}\,\pi(\tgm{\N},\check{\mathfrak{m}})
\end{equation}
This implies that $\mathbf{S}$ can be obtained from the solution to the HMIP. The reason is that this solution includes $\pi(\tgm{\N},\check{\mathfrak{m}})$, from $\pi(\tgm{\N},\check{\mathfrak{m}})$ we can obtain $\mathbb{S}$ via a color-projection, and $\mathbf{S}$ can be obtained from $\mathbb{S}$ as described before, since $\mathbb{S}$ is $1$-dimensional. This motivates us to formulate our first conjecture 

\begin{nconj}
\label{con:c1}
For any $\N^*$ and any extreme ray $\mathbf{S}$ of the \emph{HEC}$_{\N^*}$, there exists a class $(\tgm{\N^*},\mathfrak{m})$ realizing $\mathbf{S}$ and a minimally-degenerate fine-graining to a class $(\tgm{\N},\check{\mathfrak{m}})$ for some finite $\N\geq\N^*$ such that
\begin{equation}
    \mathbb{S}(\tgm{\N},\check{\mathfrak{m}})=\pi(\tgm{\N},\check{\mathfrak{m}})
\end{equation}
and $\pi(\tgm{\N},\check{\mathfrak{m}})$ is nowhere-zero.\footnote{\, Here by a nowhere-zero PMI $\mathbb{P}$ we mean that for any choice of generators of $\mathbb{P}$, none of the components vanishes.}
\end{nconj}

If this conjecture holds, in principle all one needs to do to construct the HEC$_{\N^*}$ is to consider all possible color-projections from $\N$ to $\N^*$ of all nowhere-zero PMIs that are realizable holographically, for each $\N\geq\N^*$, pick a non-negative ray from each $1$-dimensional subspace obtained by this procedure, and take the conical hull of the resulting set. Notice that not all entropy rays obtained in this way will be extremal, however the ones that are non-extremal are guaranteed to be contained in the HEC$_{\N^*}$ by the nowhere-zero condition. This follows from the fact that if a holographic PMI is nowhere-zero, it contains at least a (not necessarily extremal) nowhere-zero entropy vector $\mathbf{S}'$ realized by a graph model. Since we are only considering PMIs that are color-projected to $1$-dimensional subspaces, and by being nowhere-zero $\mathbf{S}'$ cannot be mapped to the null vector by such a projection, the result of the projection must be the subspace generated by the projection of $\mathbf{S}'$, which is inside the HEC$_{\N^*}$. On the other hand, \ref{con:c1} will guarantee that all extreme rays of the HEC$_{\N^*}$ will be included in the resulting set.

At first sight, this procedure seems daunting, since we are demanding to consider all possible values of $\N\geq\N^*$. In practice however this is never necessary, since for any $\N^*$ the HEC$_{\N^*}$ only has a finite number of extreme rays (it is a polyhedral cone) and can therefore be reconstructed from the solution to the HMIP for $\N=\N_{\text{max}}(\N^*)$ given by
\begin{equation}
    \N_{\text{max}}(\N^*)=\max\; \{\, \N_{\text{min}}(\mathbf{S}),\;\forall\,\mathbf{S} \}
\end{equation}
where for each extreme ray $\mathbf{S}$ of the HEC$_{\N^*}$, $\N_{\text{min}}(\mathbf{S})$ denotes the smallest $\N\geq\N^*$ for which \cref{con:c1} holds. The reason is that for each extreme ray $\mathbf{S}$ of the HEC$_{\N^*}$ and PMI $\mathbb{P}$ in the solution to the HMIP for $\N=\N_{\text{min}}(\mathbf{S})$ such that
\begin{equation}
    \mathbb{S}=\Phi_{\N_{\text{min}}(\mathbf{S})\rightarrow\N^*}\mathbb{P},
\end{equation}
there is a lift\footnote{\, This is a straightforward consequence of the standard lift construction described in \cref{subsec:disconnected}, and of the fact that the min-cut subspace of a class determines its PMI (cf.\ \cref{cor:vfixespmi}).} to a new PMI $\mathbb{P}'$ in the solution to the HMIP for $\N=\N_{\text{max}}$ such that
\begin{equation}
    \mathbb{P}=\Phi_{\N_{\text{max}}\rightarrow\N_{\text{min}}}\mathbb{P}'
\end{equation}
Therefore we can simply obtain $\mathbb{S}$ from $\mathbb{P}'$ as follows
\begin{equation}
    \mathbb{S}=\Phi_{\N_{\text{max}}\rightarrow\N^*}\mathbb{P}'
\end{equation}
where
\begin{equation}
\label{eq:colorprojcomposition}
    \Phi_{\N_{\text{max}}\rightarrow\N^*}=\Phi_{\N_{\text{min}}\rightarrow\N^*}\,\Phi_{\N_{\text{max}}\rightarrow\N_{\text{min}}}
\end{equation}

Notice that we are only arguing for the existence of $\N_{\text{max}}(\N^*)$, but we are not providing a way to determine its value for a given $\N^*$. In any concrete reconstruction of the HEC$_{\N^*}$ one would instead need to know the value of $\N_{\text{max}}(\N^*)$ to make sure that one is considering the solution to the HMIP for a value of $\N$ large enough such that it is sufficient to derive the complete set of extreme rays.\footnote{\, In principle one could also proceed as follows. Start from the solution to the HMIP for $\N=\N^*$ and construct a cone as we explained. Then derive its facets and check if they are valid holographic entropy inequalities using the contraction maps of \cite{Bao:2015bfa}. If all facets are valid inequalities, the reconstruction is complete, otherwise one can increase the value of $\N$ and repeat the process. Once $\N_{\text{max}}(\N^*)$ has been reached, the reconstruction is guaranteed to be complete. The problem with this procedure however is that the conversion from extreme rays to facets is inefficient, and that it is currently unknown if all valid holographic inequalities can be proven via contraction maps.} However, it should by now be clear that the focus of the present work is not on any explicit reconstruction of the HEC, but rather on the equivalence of this problem to the solution of the HMIP.

For small $\N$ our evidence strongly supports \ref{con:c1} and allows us to derive some tight upper bounds on $\N_{\text{max}}(\N^*)$. For $\N\leq4$, we have already seen that $\N_{\text{max}}(\N^*)=\N^*$, so the first non-trivial case is $\N=5$. The graph models shown in the first and third columns of \cref{fig:N5trees} correspond to PMIs of dimensions ranging from $2$ (e.g., the first one) to $6$ (e.g., the last one). To explore \cref{con:c1}, one may consider incremental fine-grainings\footnote{\, Again, these are particular choices of fine-grainings of the min-cut structure of a graph model specified by the same choice of weights as in the original graph.} of these graph models obtained by first splitting boundary vertices of degree higher than one, and then recoloring the boundary vertices. In all cases, one observes that the min-cut subspace remains $1$-dimensional upon fine-graining, while the PMI exhibits a non-increasing dimensionality. Remarkably, in most cases one obtains a graph model with a $1$-dimensional PMI after only a couple of steps, and in all cases three steps always suffice.\footnote{\, In fact, by taking the currently known extreme rays of the HEC$_6$ which are also extremal in the SAC$_6$ (equivalently, extreme rays of the SAC$_6$ which are holographic), and coarse-graining down to $5$ parties, we observe that as many as $16$ out of the $19$ orbits of extreme rays of the HEC$_5$ are obtained. In other words, remarkably, $16$ of the extreme rays of the HEC$_5$ descend from the SAC$_{\N}$ already at $\N=6$.} In summary, this allows us to conclude that $\N_{\text{max}}(5)\leq 8$, and that this bound is likely saturated. Although our data for $\N^*=6$ is incomplete and thus insufficient to derive a similar bound, we highlight that the same qualitative features are observed for all graph models realizing the currently known extreme rays of the HEC$_6$ (unlike in the previous subsection, here we are not attempting to realize these extreme rays with tree graph models). The splitting of boundary vertices and the straightforward fine-graining of a graph model specified by a maximal recoloring is sufficient to obtain a new graph model such that the corresponding min-cut subspace and PMI coincide.\footnote{\, Technically this observation is not an exact confirmation that \ref{con:c1} holds in these cases, since a priori the fine-grainings of the min-cut structures may not be minimally-degenerate. However, we can also imagine to first use \cref{lem:key} before the recoloring, to obtain a $1$-dimensional W-cell. The fine-graining induced by the recoloring is then guaranteed to be minimally-degenerate, and \ref{con:c1} holds.}

Notice that while \ref{con:c1} in principle allows for fine-grainings to graph models with higher dimensional min-cut subspaces, in all known cases discussed above the equivalence with PMIs was achieved with $1$-dimensional min-cut subspaces. This suggests that the reconstruction of the HEC could actually be even simpler, and reliant only on the knowledge of $1$-dimensional holographic PMIs. This motivates the following stronger conjecture, consistent with all our data to date:

\begin{nconj}
\label{con:c2}
For any $\N^*$ and any extreme ray $\mathbf{S}$ of the \emph{HEC}$_{\N^*}$, there exists an extreme ray $\mathbf{R}$ of the \emph{SAC}$_{\N}$ for some finite $\N\geq\N^*$, a graph model realizing $\mathbf{R}$ and a recoloring $\beta^{\downarrow}$ such that from the corresponding coarse-graining $\phi$ we have
\begin{equation}
\label{eq:C2colorprojection}
    \mathbf{S}=\Phi_{\N\rightarrow\N^*}\,\mathbf{R}
\end{equation}
\end{nconj}

If this conjecture holds, it is not necessary to know the full solution to the HMIP to reconstruct the HEC. Instead it is sufficient to know only the $1$-dimensional PMIs that are holographic or, equivalently, the extreme rays of the SAC that can be realized by graph models.
We could then redefine $\N_{\text{max}}(\N^*)$ as the smallest number of parties from which all HEC$_{\N^*}$ extreme rays can be recovered as in \cref{eq:C2colorprojection}.
Explicitly, the HEC$_{\N^*}$ would then be given by 

\begin{equation}
\label{eq:hecformula}
    \text{HEC}_{\N^*}=\text{cone}\,\{\Phi_{\N_{\text{max}}\rightarrow\N^*}\,\mathbf{R},\;\forall\,\mathbf{R}\in\mathscr{R}_{_{\text{H}}}^{\N_{\text{max}}},\;\forall\,\Phi_{\N_{\text{max}}\rightarrow\N^*}\}
\end{equation}
where $\mathscr{R}_{_{\text{H}}}^{\N}$ is the set of holographic extreme rays of the SAC$_{\N}$. For instance, from the bound on $\N_{\text{max}}$ mentioned above for $\N^*=5$ we see that the HEC$_5$ can be obtained from the conical hull of just those extreme rays of the SAC$_8$ which are holographic. Notice that from a structural point of view this a highly non-trivial statement, since the projections $\Phi_{\N_{\text{max}}\rightarrow\N^*}$ that appear in \cref{eq:hecformula} are not arbitrary projections, but color-projections associated to coarse-grainings. Finally, notice that the reconstruction formula \cref{eq:hecformula} also gives a full solution to the HMIP for $\N=\N^*$, as one can recover all holographic PMIs from the knowledge of the extreme rays.

Due to its generality however, proving \ref{con:c2} might be quite challenging, in particular because the topology of the graph models that realize the holographic extreme rays of the SAC is not restricted in any way. On the other hand, the observations at the end of the previous subsection suggest that there might be an additional simplification. We have seen that for all extreme rays of the HEC$_5$ we could find realizations by graph models with tree topology. And that for all these tree graph models, the fine-graining specified by a maximal recoloring was automatically giving a new (simple tree) graph model with a $1$-dimensional PMI. For instance, the tree graphs exhibited in \cref{fig:N5trees} show that all the extreme rays of the HEC$_5$ can be obtained by coarse-graining extreme rays of the SAC$_{\N}$ realizable by simple tree graph models for some $5<\N\leq 11$.\footnote{\, Specifically, this upper bound comes from the tree counterpart to the only non-planar graph in \cref{fig:N5trees}, which requires the highest lift of all graphs for obtaining a simple tree: from $\N^*=5$ to $\N=11$.} Similar facts seem to hold for the known extreme rays of the HEC$_6$: by just applying the simple entropy-preserving graph manipulations described in \cref{sec:gops}, we are able to turn almost all graph models realizing these rays into new graph models with tree topology. One of just a handful of exceptions is the graph in \cref{fig:N6er3}. However, we stress that despite the fact that we have not found a tree form for it to support \ref{con:c3}, one can easily check that after a splitting of e.g. vertex $C$ alone, and a maximal recoloring, one already obtains a graph model realizing an extreme ray of the SAC$_7$, in agreement with \ref{con:c2}.

Based on these observations, it is therefore interesting to explore the implications of the following conjecture.

\begin{nconj}
\label{con:c3}
For any $\N^*$, every (nowhere-zero) extreme ray of the \emph{HEC}$_{\N^*}$ is realizable by a graph model with tree topology.
\end{nconj}

The reader will probably already guess how this new conjecture is related to \ref{con:c2} and the reconstruction problem, since we have already seen several examples of this situation. However let us explain more carefully why \ref{con:c2} implies that the reconstruction is given by \cref{eq:hecformula}. 

Consider an extreme ray $\mathbf{S}$ of the HEC$_{\N^*}$, and suppose that \ref{con:c3} holds. Then there exists a class $(\tgm{\N^*},\mathfrak{m})$ whose S-cell is the ray generated by $\mathbf{S}$, and we denote by $\mathbb{S}$ its ($1$-dimensional) min-cut subspace. In general the W-cell of $(\tgm{\N^*},\mathfrak{m})$ could be higher dimensional, and an arbitrary fine-graining could produce a new class whose min-cut subspace has dimension greater than one (recall the example at the end of \cref{subsec:fg}). However, if this is the case, we can always use \cref{lem:key} to find a new class $(\htgm{\N^*},\hat{\mathfrak{m}})$ such that the min-cut subspace is still $\mathbb{S}$, while the S-cell is just a single entropy ray. Notice that, as we have already seen in various examples, when we use \cref{lem:key} we may be forced to apply the reduction described in \cref{ssec:vanishing} and delete some of the edges. We can now apply \cref{thm:nonsimpletree} to $(\htgm{\N^*},\hat{\mathfrak{m}})$ to get yet a new class $(\check{\tgm{\N}},\check{\mathfrak{m}})$, for some $\N>\N^*$, where $\check{\tgm{\N}}$ is now a simple tree and the  color-projection $\Phi_{\N\rightarrow\N^*}$ of its PMI is  $\mathbb{S}$. Finally, since the W-cell of $(\htgm{\N^*},\hat{\mathfrak{m}})$ was $1$-dimensional, it is guaranteed by \cref{lem:Wrefinementfg} to remain $1$-dimensional after the fine-graining, and the PMI of $(\check{\tgm{\N}},\check{\mathfrak{m}})$ is therefore $1$-dimensional as well. 

In summary, assuming \ref{con:c3}, we have shown how to construct a new class that satisfies \ref{con:c2}, showing that \ref{con:c3} is sufficient for the reconstruction of the HEC via \cref{eq:hecformula}.

We conclude by mentioning one more conjecture which is not specific to any particular set of faces of the HEC or the SAC, and could be an interesting general property of all min-cut structure on arbitrary topological graph models. 

\begin{nconj}
\label{con:c4}
For any $\N$, topological graph model $\tgm{\N}$ and min-cut structure $\mathfrak{m}$, there exists a topological graph model $\tgm{\N}'$ \emph{with tree topology} and a min-cut structure $\mathfrak{m}'$ such that
\begin{equation}
\mathbb{S}(\tgm{\N}',\mathfrak{m}')=\mathbb{S}(\tgm{\N},\mathfrak{m})
\end{equation}
\end{nconj}

A priori one may have doubts about this conjecture, considering how special tree graphs are compared to arbitrary graphs. Notice however that the trees in \ref{con:c1} are not required to be simple, and can have an arbitrary number of boundary vertices for each color. Finally should be clear that if one could prove \ref{con:c4}, then \ref{con:c3} would follow immediately. By the argument presented above, \ref{con:c3} would then imply \ref{con:c2}, and the reconstruction of the HEC$_{\N^*}$ from the extreme rays of the SAC$_{\N_{\text{max}}}$ for all $\N^*$ would be given by \cref{eq:hecformula}.

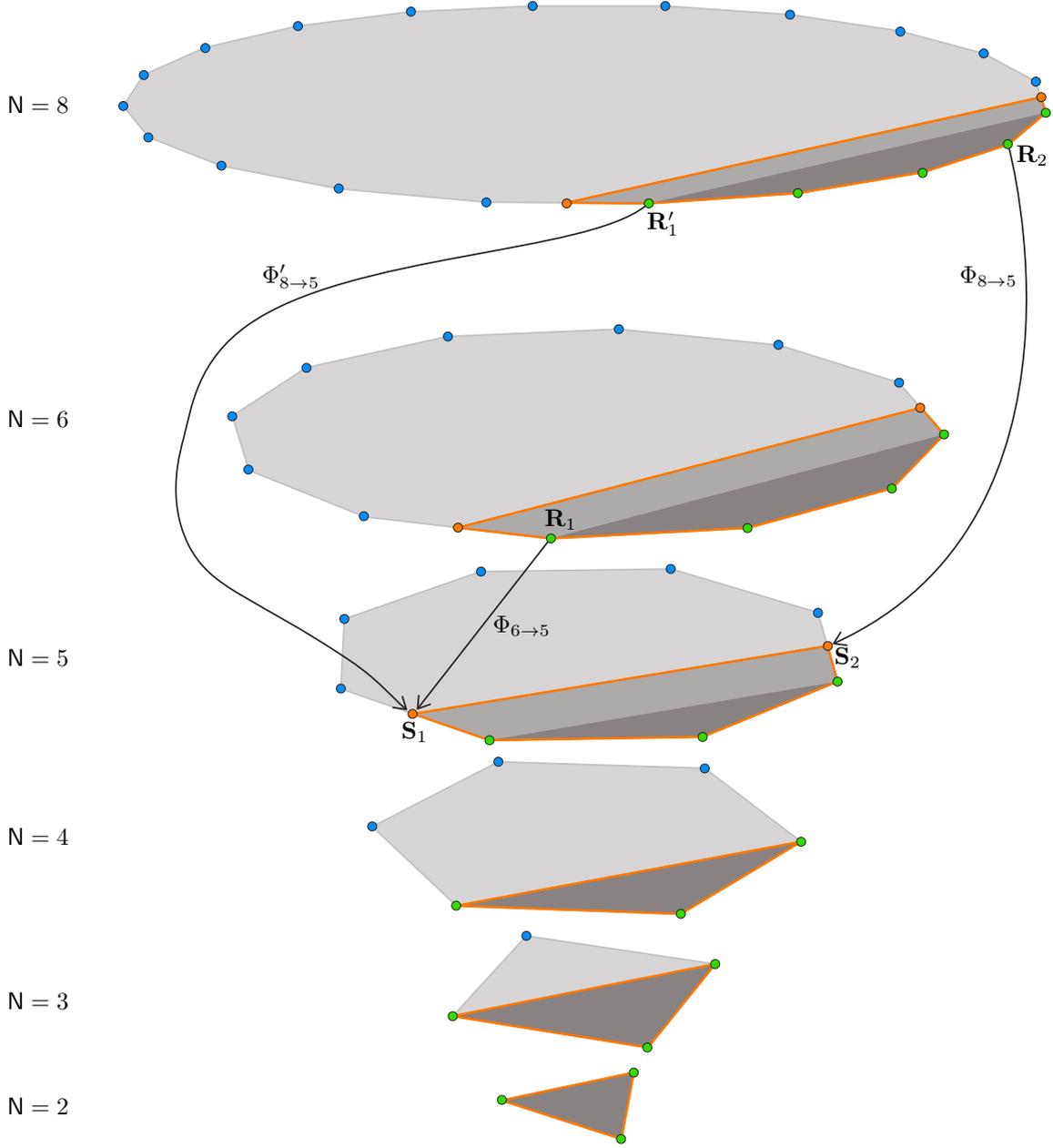
\begin{figure}
\centering
\hspace*{1.625cm}\includegraphics[scale=0.65]{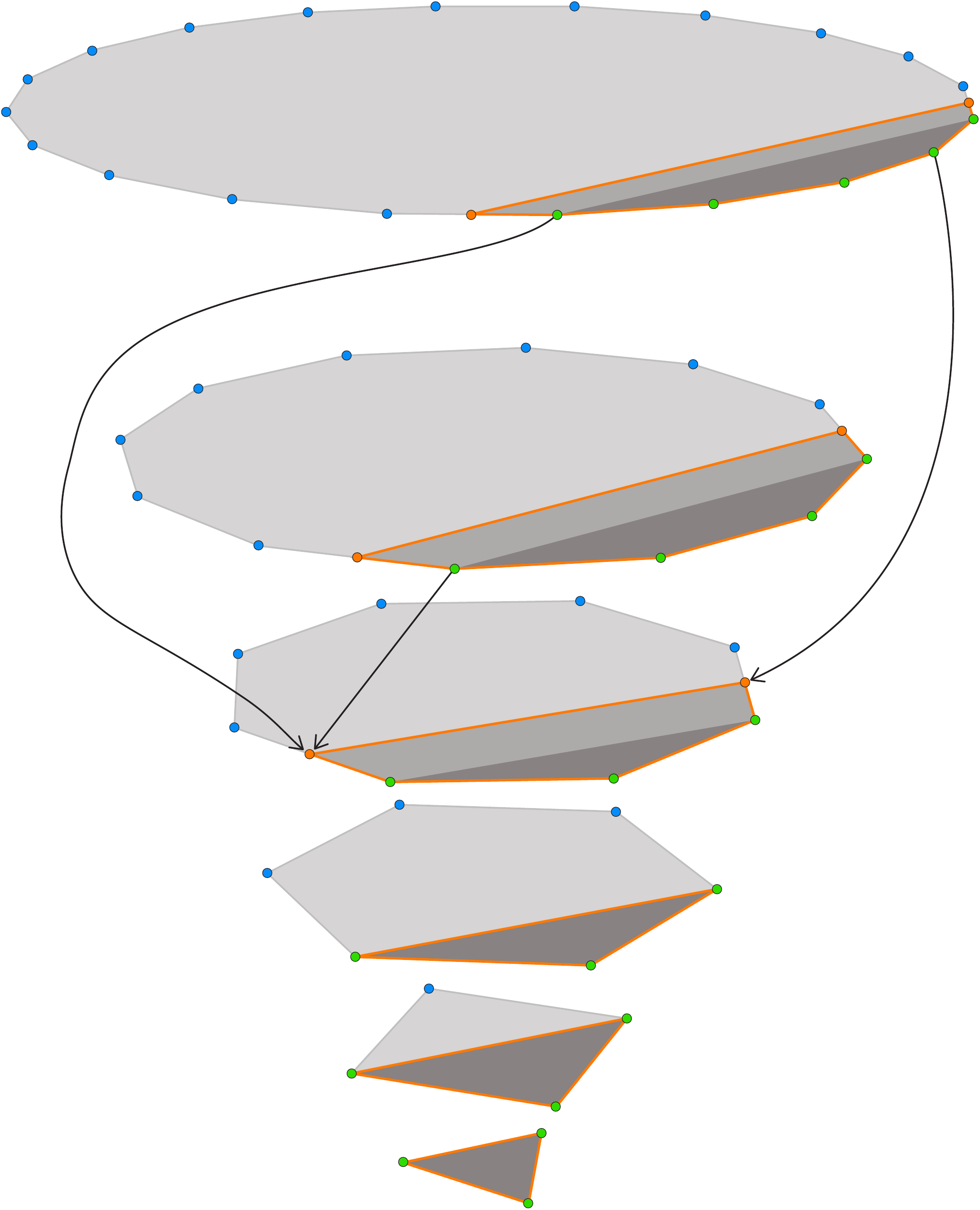}
\put(-417,431){\makebox(0,0){{\footnotesize $\N=8$}}}
\put(-417,301){\makebox(0,0){{\footnotesize $\N=6$}}}
\put(-417,202){\makebox(0,0){{\footnotesize $\N=5$}}}
\put(-417,128){\makebox(0,0){{\footnotesize $\N=4$}}}
\put(-417,60){\makebox(0,0){{\footnotesize $\N=3$}}}
\put(-417,16){\makebox(0,0){{\footnotesize $\N=2$}}}
\put(-218,215){\makebox(0,0){{\footnotesize $\Phi_{6\rightarrow 5}$}}}
\put(-26,360){\makebox(0,0){{\footnotesize $\Phi_{8\rightarrow 5}$}}}
\put(-313,360){\makebox(0,0){{\footnotesize $\Phi'_{8\rightarrow 5}$}}}
\put(-262,171){\makebox(0,0){{\footnotesize $\mathbf{S}_1$}}}
\put(-84,202){\makebox(0,0){{\footnotesize $\mathbf{S}_2$}}}
\put(-160,382){\makebox(0,0){{\footnotesize $\mathbf{R}'_1$}}}
\put(-202,259){\makebox(0,0){{\footnotesize $\mathbf{R}_1$}}}
\put(-8,410){\makebox(0,0){{\footnotesize $\mathbf{R}_2$}}}
\vspace{0.7cm}
    \caption{A schematic cartoon of a cross-section of the SAC$_{\N}$ and the HEC$_{\N}$ for different values of $\N$. The extreme rays of the SAC$_{\N}$ are represented by the green and blue vertices. The green ones are also extreme rays of the HEC$_{\N}$, and their conical hull is the darkest shaded region. The blue ones are outside of the HEC$_{\N}$. The boundary of the HEC$_{\N}$ is highlighted in orange, and the orange vertices represent extreme rays of the HEC$_{\N}$ which are \emph{not} extreme rays of the SAC$_{\N}$. The arrows indicate color-projections associated with coarse-grainings that implement the reconstruction, mapping the green vertices to the orange. For $\N=5$ the figure shows two extreme rays $\mathbf{S}_1$ and $\mathbf{S}_2$ that are not extreme rays of the SAC$_5$ and need to be reconstructed. While $\N=6$ is sufficient to obtain $\mathbf{S}_1$ as $\mathbf{S}_1=\Phi_{6\rightarrow 5}\mathbf{R}_1$, it is not sufficient to reconstruct $\mathbf{S}_2$. On the other hand, $\N=8$ is sufficient to reconstruct both rays, $\mathbf{S}_1=\Phi'_{8\rightarrow 5}\mathbf{R}'_1$ and $\mathbf{S}_2=\Phi_{8\rightarrow 5}\mathbf{R}_2$.
    To avoid clutter, we are suppressing all irrelevant projections/lifts of other extreme rays. The shapes of the cone cross-sections (as well as the numbers of facets and extreme rays of various kinds) are not meant to be taken literally.
    }
    \label{fig:stackofcones}
\end{figure}

\section{Discussion}
\label{sec:discussion}

In this section we briefly summarize our results and what remains to be proven to definitely conclude that the solution to the holographic marginal independence problem contains sufficient information for the reconstruction of the holographic entropy cone. Moreover, we comment on several interesting directions for future investigations which are suggested by the possibility of such a reconstruction, and by the framework that we have developed throughout this work. 

\paragraph{Summary:} The main goal of this work was to argue that the holographic entropy cone can be reconstructed, for an arbitrary number of parties, from the solution to the holographic marginal independence problem. We have shown how such a reconstruction would work, and argued that it takes a particularly simple form. Indeed, we argued that to be able to reconstruct the HEC one does not even need to know the full solution to the HMIP, but only the set of $1$-dimensional holographic
PMIs. Ultimately, what this implies is the following
\begin{claim}
The knowledge of all holographic entropy inequalities, for all values of $\N$, is equivalent to the knowledge of all holographic extreme rays of the \emph{SAC}, for all values of $\N$.
\end{claim}
A summarizing cartoon of this equivalence and of the reconstruction procedure is provided in \cref{fig:stackofcones}. We stress that, as we explained above, this claim does not hold for a fixed value of $\N$ when
$\N\geq 5$, i.e., knowing the holographic extreme rays of the SAC for $\N$ parties does not in general provide enough information to derive all the $\N$-party holographic entropy inequalities.

\paragraph{Completing the proof:} To complete the proof of our main claim above it remains to prove \ref{con:c2}. This problem however might be challenging, since \ref{con:c2} is specific to extreme rays, which are unknown (and are precisely what we want to reconstruct), and the graph models realizing them are allowed to have unrestricted topology. We believe that the best way to proceed is to prove \ref{con:c3}. Although we were not able to check that this conjecture holds for all known extreme rays of the HEC$_6$, we also did not attempt an exhaustive search, and it seems likely that one simply needs to consider more complicated tree graphs.

While in principle it is possible that \ref{con:c3} only holds for extreme rays, it also seems reasonable to expect that it is a more general fact about $1$-dimensional min-cut subspaces, in which case proving this conjecture would presumably be easier. In fact, the realizability of min-cut subspaces via tree graphs might be even more general, and one could also attempt to directly prove \ref{con:c4}. An interesting way to proceed in this direction would be to better understand the structure of the partition of the space of edge weights into W-cells for a given topological graph model. Since this structure also determines the maps to entropy space that produce the min-cut subspaces for the various min-cut structures, if one could directly relate this structure to the topology of the graph, one might then be able to prove that any min-cut subspace can be obtained from a tree.

Finally, we should also contemplate the possibility that our main claim above is false. In such a case it might still be possible to reconstruct the HEC from the solution to the HMIP, but it would not be sufficient to only know the holographic extreme rays of the SAC. To prove that the reconstruction is possible one should then try to prove \ref{con:c1}, although this conjecture suffers from the same complications we just mentioned for \ref{con:c2}, since it is specific to extreme rays and the graphs have unrestricted topology. However, one may also wonder if for each equivalence class of graph models (not just those which realize extreme rays), there always exists a fine-graining to a new one such that the min-cut subspace and the PMI coincide. While stronger then \ref{con:c1}, this property might be easier to prove by investigating more deeply the relation between min-cut subspaces and PMIs for arbitrary topological graph models and min-cut structures.

\paragraph{Finding the holographic extreme rays of the SAC:} 

If our main claim above holds, the problem of determining the HEC is mapped to the problem of determining which extreme rays of the SAC are holographic. While this is beyond the scope of the present work, we briefly comment on the importance of this question and some intriguing possible answers.

For a fixed number of parties $\N$, consider the following sets of extreme rays of the SAC$_{\N}$ (dropping the $\N$-dependence to simplify the notation) 
\begin{align}
    \er & \coloneqq \{\text{extreme rays of the SAC}\} \nonumber\\
    \erssa & \coloneqq \{\text{extreme rays of the SAC that satisfy SSA}\}\nonumber\\
    \erq & \coloneqq \{\text{extreme rays of the SAC that can be realized by quantum states}\}\nonumber\\
    \erh & \coloneqq \{\text{extreme rays of the SAC that can be realized by graph models}\} 
\end{align}
These sets clearly satisfy the following chain of inclusions
\begin{equation}
\label{eq:inclusions}
    \erh \subseteq \erq \subseteq \erssa \subseteq \er
\end{equation}
and we have $\erssa=\er$ only for $\N=2$, while typically $\abs{\erssa}\ll\abs{\er}$.

If our main claim holds, the deep question about the physical origin of the HEC and its structure is distilled to the question about the relation between the sets $\erh$ and $\erq$. If $\erh\!\!=\!\!\erq$, all extreme rays of the SAC compatible with quantum mechanics would participate in the construction of the HEC, and there would be no other holographic constraint to resolve. While all data currently available points in this direction, it is conceivable that for larger values of $\N$ there exist extreme rays of the SAC that can be realized by quantum states but not by geometric states. This would hint at more fundamental holographic constraints which operate at the level of the SAC and would be very interesting to investigate.

Answering this question however could be quite difficult, since very little is known about the structure of the QEC, and to the best of our knowledge, there is no systematic construction of quantum states that realize the extreme rays of the SAC. A more approachable question on the other hand could be whether $\erh$ and $\erssa$ coincide. Interestingly, for $\N\leq 5$, for which the elements of $\erssa$ can be determined using standard algorithms \cite{Fukuda:2015}, this turns out to be the case. If $\erh\!=\erssa$, the chain of inclusions \cref{eq:inclusions} immediately implies that $\erh\!=\erq$ and $\erq\!=\erssa$, and we would not only obtain a full characterization of the HEC, but also new information about the QEC for arbitrary $\N$. 
In order to establish if $\erh\!=\erssa$, the structure of $\erssa$ is currently being investigated in \cite{he:2022}. The goal is to obtain a sufficiently detailed characterization of the elements of this set, such that they could then be matched with explicit holographic realizations in terms of graph models.

\paragraph{Generalization to HRT:} 
In this work we focused on the static version of the holographic entropy cone, since this is the situation where configurations of HRRT surfaces can be described by graph models. The connection between the HEC and the solution to the HMIP however seems to suggest that the same structure would pertain also to dynamical spacetimes. Intuitively, the reason is that the crucial information about the configurations of HRRT surfaces is not rooted in the area of the individual surfaces, but rather the connectivity of the entanglement wedges of the various subsystems. It would be interesting to explore this possibility in more detail.\footnote{\, For two-dimensional CFT, it was proven in \cite{Czech:2019lps} that any holographic entropy inequality which holds for RT surfaces also holds for HRT surfaces. Despite the limitations related to the special properties of a $3$-dimensional bulk spacetime, it would be useful to clarify the relation between the present work and the result of \cite{Czech:2019lps}, where it was SSA rather than SA that played a crucial role.}

The most convenient contingency would be one wherein one can devise a graph model for the general time-dependent situation, with the subsystem entanglement entropy again given by the sum of the edge weights of the cut edges for a suitably-defined min-cut, in which case the graph side of the HEC construction would be identical.  The immediate obstacle with this approach arises from the fact that parts of the HRT surfaces for various regions can be timelike-separated, and more importantly that the surfaces in question have their areas extremized rather than simply minimized.  One possibility of circumventing this complication is examined in \cite{ghh:minimax}.

\paragraph{The holographic entropy polyhedron (HEP):}
Many of the tools that we used here have a close relationship with those introduced in \cite{Hubeny:2018ijt,Hubeny:2018trv} in the context of the \textit{holographic entropy arrangement}. The reader who is familiar with these works will probably notice for example the similarity between the min-cut structure of a topological graph model and the proto-entropy vector of a holographic configuration. Analogously, the min-cut subspace is closely related to the set of constraints that must be satisfied by the coefficients of an unspecified ``information quantity'' (defined as an arbitrary linear combination of entropies, or equivalently as an element of the dual of entropy space) such that it vanishes on a full S-cell.
When a min-cut subspace $\mathbb{S}$ has dimension $\D-1$ in fact, $\mathbb{S}^{\perp}$ is the vector of coefficients of the only information quantity that vanishes on all vectors in the S-cell whose linear span is $\mathbb{S}$. This is the type of information quantity that was dubbed a ``primitive'' in \cite{Hubeny:2018ijt,Hubeny:2018trv}. The holographic entropy arrangement was then defined as the set of all hyperplanes in $\N$-party entropy space which correspond to these primitive quantities, and the HEP as the polyhedron carved by the entropy inequalities associated to the primitive quantities that have a definite sign for geometric states. 

In light of the results of the present work, it is tempting to identify the holographic entropy arrangement with the set of codimension-$1$ min-cut subspaces, and the HEP with the HEC. Indeed, as we have shown here, each extreme ray of the HEC corresponds to a $1$-dimensional min-cut subspace, and the min-cut subspace of a graph obtained from the disjoint union of multiple graphs is the sum of the min-cut subspaces of the individual graphs (cf. \cref{subsec:disconnected}). Using this construction one could then combine the graphs realizing the extreme rays of the HEC to obtain new graphs whose min-cut subspaces are the supporting subspaces of the facets of the HEC. All information quantities corresponding to holographic entropy inequalities would then be primitive in the sense of \cite{Hubeny:2018ijt,Hubeny:2018trv}, and the holographic entropy cone and polyhedron would coincide.

This simple reasoning however obfuscates some of the subtleties that were at the core of the motivations of \cite{Hubeny:2018trv} for introducing the HEP in the first place. Throughout this work we have constantly seen that a crucial role in determining the structure, or certain transformation properties, of the objects that we introduced has been played by the ``pattern of degeneracies'' of an equivalence class of graph models. This is especially true for the realization of the extreme rays of the HEC via graph models, which as we have seen, have a $1$-dimensional min-cut subspace and are therefore ``maximally degenerate''. In terms of holographic configurations and extremal surfaces, these degeneracies seem to translate to situations where the configurations are so fine-tuned that the entropies of multiple subsystems are computed by several coexisting surfaces of equal area. In many other contexts however, like bulk reconstruction, one is typically inclined to ignore such fine-tuned cases to avoid worrying about subtleties regarding order of limits \cite{Dong:2016hjy} by focusing instead on generic situations. Indeed, this was precisely the approach followed by \cite{Hubeny:2018trv}, which defined primitive information quantities by restricting to generic configurations. In the language of this work, resolving this issue seems to boil down to the question of whether the min-cut subspaces which support the facets of the HEC can be realized by generic equivalence classes, rather than via the aforementioned construction based on extreme rays. The precise details about the connection between the two formulations however could be more subtle, and requires more careful investigations \cite{Hernandez-Cuenca:2022a}.

Finally, a similar fine-tuning is related to the concrete holographic realizations of graph models, which is obtained via $3$-dimensional multiboundary wormhole geometries \cite{Bao:2015bfa} where the subsystems are chosen to cover the entire boundaries. This choice in fact allows for a fine-tuning of the entanglement among the CFTs that live on the different boundaries, which has no counterpart for subsystems of a single boundary, where the entropies are divergent and cannot be meaningfully regulated. For this reason, \cite{Hubeny:2018trv} defined the HEP by further restricting the set of primitive quantities to those that can be realized by single-boundary configurations. Work is in progress to shed light also on this issue, and to resolve the subtle differences between the HEC and the HEP \cite{Hernandez-Cuenca:2022a}.

\paragraph{Quantum entropies:}
Our discussion in this paper has focused on entropies realizable by geometric states in holography, but it would be interesting to explore if some of the tools developed here can also be utilized to analyze properties of quantum entropies for other classes of states. For example, models of entanglement capturing richer patterns of quantum correlations through hypergraphs and topological links have been explored in \cite{Bao:2020zgx,Bao:2020mqq,Walter:2020zvt,Bao:2021gzu} and much of the structure we studied here admits a natural extrapolation to those settings. It is then interesting to ask if any of our results and conjectures could be suitably generalized. For instance, are min-cuts the only ingredient at the heart of why the HEC descends from the SAC in the sense of \cref{con:c2}? If so, it would be reasonable to expect an analogous phenomenon with hypergraphs, namely, that the hypergraph entropy cone be obtainable from extreme rays of the SAC that are realizable by hypergraphs. This poses a non-trivial question already at $4$ parties, when the hypergraph entropy cone includes an extreme ray which is not an extreme ray of the SAC$_4$, and it is unclear whether it can be obtained as a color projection of an extreme ray of the SAC$_{\N}$ for some $\N>4$ \cite{Bao:2020zgx}.

\paragraph{Quantum Corrections:}
The graph models studied here precisely capture the properties of the RT prescription in holography. What can our results say about holographic entanglement entropy beyond the strictly classical limit where entropies are purely geometric? It would be interesting to explore how quantum contributions from matter fields affect our results. The application of the combinatorial machinery to the study of entropies computed by the quantum extremal surface prescription was already shown to be fruitful in \cite{Akers:2021lms}. A similar reasoning could be used here to discern the imprint that the area term leaves on the generalized entropy and the discrete structure potentially emerging from it.

\acknowledgments
 
It is a pleasure to thank Mukund Rangamani for initial collaboration and many discussions on this project.
We also thank Temple He and Matt Headrick for useful discussions.
We would like to thank KITP, with support from the National Science Foundation under Grant No. NSF PHY-1748958, for its hospitality during the program ``Gravitational Holography''.
During the completion of this project, SHC has been supported by NSF grants PHY-1801805 and PHY-2107939, by a Len DeBenedictis Graduate Fellowship, and by funds from UCSB.
VH has been supported in part by the U.S.\ Department of Energy grant DE-SC0009999 and in part by the U.S.\ Department of Energy grant DE-SC0020360 under the HEP-QIS QuantISED program.
MR has been supported by the University of Amsterdam, via the ERC Consolidator Grant QUANTIVIOL, and by the Stichting Nederlandse Wetenschappelijk Onderzoek Instituten (NWO-I).

\appendix

\section{Graph operations}
\label{sec:gops}

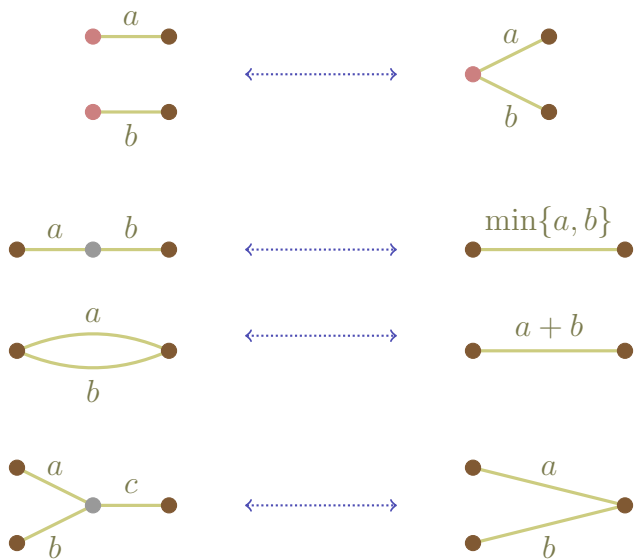
\begin{figure}[b]
    \centering
    \begin{tikzpicture}
    	\tikzmath{\v=0.1;\e=1;\a=2;\g=1;}
    	\coordinate (A1v) at (0,0.5);
    	\coordinate (A2v) at (0,-0.5);
    	\coordinate (Av) at (\e+\a+2*\g,0);
    		
		\draw[edgestyle] (A1v) -- node[above,edgeweightstyle]{$a$} +(1,0);
		\draw[edgestyle] (A2v) -- node[below,edgeweightstyle]{$b$} +(1,0); 
		\filldraw [color=Acolor] (A1v) circle (\v);
		\filldraw [color=Acolor] (A2v) circle (\v);
		\filldraw [color=Fcolor] ($(A1v)+(1,0)$) circle (\v);
		\filldraw [color=Fcolor] ($(A2v)+(1,0)$) circle (\v);
		
		\draw[bfstyle] (\e+\g,0) -- +(\a,0);
		
		\draw[edgestyle] (Av) --  node[above,edgeweightstyle]{$a$}  +(1,0.5);
		\draw[edgestyle] (Av) --  node[below,edgeweightstyle]{$b$}  +(1,-0.5);
		\filldraw [color=Acolor] (Av) circle (\v);
		\filldraw [color=Fcolor] ($(A1v)+(Av)+(1,0)$) circle (\v);
		\filldraw [color=Fcolor] ($(A2v)+(Av)+(1,0)$) circle (\v);

        \useasboundingbox (0,1.2);
    \end{tikzpicture}
    
    \begin{tikzpicture}
    	\tikzmath{\v=0.1;\e=1;\a=2;\g=1;}
    	\coordinate (bv) at (0,0);
        	
		\draw[edgestyle] (bv) -- node[above,edgeweightstyle]{$a$} +(-1,0);
		\draw[edgestyle] (bv) -- node[above,edgeweightstyle]{$b$} +(1,0); 
		\filldraw [color=bvcolor] (bv) circle (\v);
		\filldraw [color=Fcolor] (bv)+(-1,0) circle (\v);
		\filldraw [color=Fcolor] (bv)+(1,0) circle (\v);

		\draw[bfstyle] (\e+\g,0) -- +(\a,0);
        	
		\draw[edgestyle] (Av) --  node[above,edgeweightstyle]{$\min\{a,b\}$}  +(2,0);
		
		\filldraw [color=Fcolor] (Av) circle (\v);
		\filldraw [color=Fcolor] (Av)+(2,0) circle (\v);
        
        \useasboundingbox (0,1.2);
    \end{tikzpicture}
    
    \begin{tikzpicture}
    	\tikzmath{\v=0.1;\ce=2;\a=2;\g=1;}
    	\coordinate (bvL) at (0,0);
    	\coordinate (bvR) at (2*\g+\ce+\a,0);
        	
    	\draw[edgestyle] (bvL) .. controls (\ce*0.34,0.3) and (\ce*0.66,0.3) .. 
    			node[above,edgeweightstyle]{$a$}  (\ce,0);
    	\draw[edgestyle] (bvL) .. controls (\ce*0.34,-0.3) and (\ce*0.66,-0.3) ..
    			node[below,edgeweightstyle]{$b$}  (\ce,0);
    	\filldraw [color=Fcolor] (bv) circle (\v);
    	\filldraw [color=Fcolor] (bv)+(\ce,0) circle (\v);
        
        \draw[bfstyle] (\ce+\g,0.2) -- +(\a,0);
            
        \draw[edgestyle] (bvR) --  node[above,edgeweightstyle]{$a+b$}  +(\ce,0);
		\filldraw [color=Fcolor] (bvR) circle (\v);
		\filldraw [color=Fcolor] (bvR)+(\ce,0) circle (\v);
        
        \useasboundingbox (0,1.2);
    \end{tikzpicture}

    \begin{tikzpicture}
        \tikzmath{\v=0.1;\e=1;\a=2;\g=1;}
    	\coordinate (bv) at (0,0);
    	\coordinate (bvR) at (2*\g+2*\e+\a+1,0);

		\draw[edgestyle] (bv) -- node[above,edgeweightstyle]{$a$}  (-\e,0.5);
		\draw[edgestyle] (bv) -- node[below,edgeweightstyle]{$b$}  (-\e,-0.5);
		\draw[edgestyle] (bv) -- node[above,edgeweightstyle]{$c$}  (\e,0);
		\filldraw [color=bvcolor] (bv) circle (\v);
		\filldraw [color=Fcolor] (bv)+(\e,0) circle (\v);
		\filldraw [color=Fcolor] (bv)+(-\e,0.5) circle (\v);
		\filldraw [color=Fcolor] (bv)+(-\e,-0.5) circle (\v);
		
		\draw[bfstyle] (\e+\g,0) -- +(\a,0);

		\draw[edgestyle] (bvR) -- node[above,edgeweightstyle]{$a$}  +(-2*\e,0.5);
		\draw[edgestyle] (bvR) -- node[below,edgeweightstyle]{$b$}  +(-2*\e,-0.5);
		
		\filldraw [color=Fcolor] (bvR)+(-2*\e,0.5) circle (\v);
		\filldraw [color=Fcolor] (bvR)+(-2*\e,-0.5) circle (\v);
		\filldraw [color=Fcolor] (bvR) circle (\v);

        \useasboundingbox (0,1.2);
    \end{tikzpicture}
    
    \caption{Basic entropy-preserving graph operations. The last one requires $c\geq a+b$. Brown vertices are to be kept fixed under these operations. Boundary vertices are colored in red, and bulk ones in gray. In a general graph, both brown and red vertices may connect to arbitrarily many other edges, whereas gray ones should only appear as shown.}
    \label{fig:gops}
\end{figure}

\begin{figure}[t]
    \centering
        \begin{tikzpicture}
        	\tikzmath{\v=0.1;\e=1;\a=2;\g=1;}
        	\coordinate (bvL1) at (90:0.8*\e);
        	\coordinate (bvL2) at (210:0.8*\e);
        	\coordinate (bvL3) at (330:0.8*\e);
        	\coordinate (bvR) at (2*\g+2.4*\e+\a,0);

    		\draw[edgestyle] (bvL1) -- node[left,edgeweightstyle]{$a$}  (bvL2);
    		\draw[edgestyle] (bvL2) -- node[below,edgeweightstyle]{$c$}  (bvL3);
    		\draw[edgestyle] (bvL3) -- node[right,edgeweightstyle]{$b$}  (bvL1);
    		\filldraw [color=Fcolor] (bvL1) circle (\v);
    		\filldraw [color=Fcolor] (bvL2) circle (\v);
    		\filldraw [color=Fcolor] (bvL3) circle (\v);

    		\draw[bfstyle] (\e+\g,0.2) -- +(\a,0);

    		\draw[edgestyle] (bvR) -- node[right,near end,edgeweightstyle]{$a+b$}  +(90:\e);
    		\draw[edgestyle] (bvR) -- node[above,sloped,edgeweightstyle]{$a+c$}  +(210:\e);
    		\draw[edgestyle] (bvR) -- node[above,sloped,edgeweightstyle]{$b+c$}  +(323:\e);
    		\filldraw [color=bvcolor] (bvR) circle (\v);
    		\filldraw [color=Fcolor] (bvR) +(90:\e) circle (\v);
    		\filldraw [color=Fcolor] (bvR)+(210:\e) circle (\v);
    		\filldraw [color=Fcolor] (bvR)+(323:\e) circle (\v);
    		
            \useasboundingbox (0,1.2);
        \end{tikzpicture}
    \caption{The new entropy-preserving graph operation ${\sf \Delta}$-{\sf Y}. Vertices are color-coded as in \cref{fig:gops}.}
    \label{fig:triangle}
\end{figure}
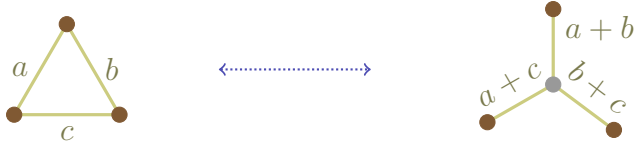

Throughout this work we have been focusing on min-cut structures on topological graph models rather than on graph models and entropy vectors. However, as we have seen, there is a particularly important situation where the two concepts essentially coincide, namely the extreme rays of the HEC. In these cases it is useful to perform certain manipulations which simplify the graph or make manifest some underlying properties. In particular, we have used these operations to obtain realizations of all extreme rays of the HEC$_5$ and most of the HEC$_6$ by tree graphs, in order to establish a connection with the extreme rays of the SAC$_{\N}$. In this appendix we provide a brief presentation of the operations we have used.

The kind of graph operations one is interested in are those which can be used to locally change the vertices, edges, and weights of a general graph while preserving its entropy vector. A few such basic operations were listed in Figure 6 of \cite{Bao:2015bfa}. It turns out these can all be easily broken down into simpler operations, which we reproduce in \cref{fig:gops}.

In the examples we discussed, however, these operations were not always sufficient to convert a graph model into a tree, and there is another graph operation which has proven remarkably useful (e.g. in obtaining the tree graphs at the bottom of \cref{fig:N5trees}), the ${\sf \Delta}$-{\sf Y} exchange shown in \cref{fig:triangle}. Since this is a non-trivial operation, we include a proof that it indeed does not alter the entropy vector of a graph model:
    \begin{table}[b]
        \centering
        \begin{tabular}{c|c|c}
            $\gm{\N}$ & weight & $\gm{\N}'$ \\\hline
            $\es, \{i,j,k\}$ & $0$ & $\es, \{\sigma,i,j,k\}$ \\
            $\{i\},\{j,k\}$ & $w_{ij}+w_{ik} = w_{\sigma i}$ & $\{i\}, \{\sigma,j,k\}$ \\
            $\{j\},\{i,k\}$ & $w_{ij}+w_{jk} = w_{\sigma j}$ & $\{j\}, \{\sigma,i,k\}$  \\
            $\{k\},\{i,j\}$ & $w_{ik}+w_{jk} = w_{\sigma k}$ & $\{k\}, \{\sigma,i,j\}$  \\
        \end{tabular}
        \caption{Weight contributions to general candidate min-cuts on $\gm{\N}$ and $\gm{\N}'$, depending on exactly which subset of the vertices involved in the graph operation in \cref{fig:triangle} is included. The agreement between the two graphs proves \cref{lem:triop}.}
        \label{tab:cuts}
    \end{table}

\begin{nlemma}
\label{lem:triop}
    The ${\sf \Delta}$-{\sf Y} exchange operation preserves the entropies.
\end{nlemma}
\begin{proof}
Consider two graph models $\gm{\N}$ and $\gm{\N}'$ related by the operation in \cref{fig:triangle}, such that $\gm{\N}$ (on the left) contains the $3$-cycle (${\sf \Delta}$), while $\gm{\N}'$ (on the right) contains the degree-$3$ vertex ({\sf Y}). In $\gm{\N}$, let the pertinent vertices be $\{i,j,k\}$, joined by edges with weights $\{w_{ij},w_{ik},w_{jk}\}$
(in \cref{fig:triangle} labeled more compactly as $\{a,b,c\}$ for simplicity of notation). To obtain $\gm{\N}'$, we preserve the same vertices, delete the edges, add a new vertex $\sigma$, and connect it to each of the former vertices by new edges of weights $\{w_{\sigma i},w_{\sigma j},w_{\sigma k}\}$ given by
\begin{equation*}
    w_{\sigma i}=w_{ij}+w_{ik}\qquad
    w_{\sigma j}=w_{ij}+w_{jk}\qquad
    w_{\sigma k}=w_{ik}+w_{jk}
\end{equation*}
An arbitrary vertex cut on $\gm{\N}$ may contain any of the  $8$ subsets of the vertices $\{i,j,k\}$ that make up the $3$-cycle. In each case, the contribution from the edges that form the cycle to the weight of the corresponding cut is given in \cref{tab:cuts}. In $\gm{\N}'$ there are instead $16$ possible subsets of $\{\sigma,i,j,k\}$ which an arbitrary cut may contain. However, $8$ of them can be immediately ruled out by the fact that they cannot achieve minimum weight. To see this, consider for example a cut containing precisely $\{\sigma,i\}$. This would receive a weight $w_{\sigma j}+w_{\sigma k}=w_{ij}+w_{ik}+2w_{jk}$, which is strictly greater than the weight that $\{i\}$ alone would give, namely, $w_{\sigma i}=w_{ij}+w_{ik}$. In general, one easily observes that for a cut to be of minimum weight, the newly added vertex $\sigma$ should only participate when the cut contains at least two of $\{i,j,k\}$, thus giving only $8$ possibilities, as in $\gm{\N}$. Looking at each case as in \cref{tab:cuts}, we arrive at the desired result that min-cut weights on the two graphs indeed match.
\end{proof}


\providecommand{\href}[2]{#2}\begingroup\raggedright\endgroup


\begin{thebibliography}{10}

\bibitem{Maldacena:1997re}
J.~M. Maldacena, \emph{{The Large N limit of superconformal field theories and
  supergravity}}, \href{https://doi.org/10.1023/A:1026654312961}{\emph{Adv.
  Theor. Math. Phys.} {\bfseries 2} (1998) 231}
  [\href{https://arxiv.org/abs/hep-th/9711200}{{\ttfamily hep-th/9711200}}].

\bibitem{Gubser:1998bc}
S.~S. Gubser, I.~R. Klebanov and A.~M. Polyakov, \emph{{Gauge theory
  correlators from noncritical string theory}},
  \href{https://doi.org/10.1016/S0370-2693(98)00377-3}{\emph{Phys. Lett. B}
  {\bfseries 428} (1998) 105}
  [\href{https://arxiv.org/abs/hep-th/9802109}{{\ttfamily hep-th/9802109}}].

\bibitem{Witten:1998qj}
E.~Witten, \emph{{Anti-de Sitter space and holography}},
  \href{https://doi.org/10.4310/ATMP.1998.v2.n2.a2}{\emph{Adv. Theor. Math.
  Phys.} {\bfseries 2} (1998) 253}
  [\href{https://arxiv.org/abs/hep-th/9802150}{{\ttfamily hep-th/9802150}}].

\bibitem{Almheiri:2014lwa}
A.~Almheiri, X.~Dong and D.~Harlow, \emph{{Bulk Locality and Quantum Error
  Correction in AdS/CFT}},
  \href{https://doi.org/10.1007/JHEP04(2015)163}{\emph{JHEP} {\bfseries 04}
  (2015) 163} [\href{https://arxiv.org/abs/1411.7041}{{\ttfamily 1411.7041}}].

\bibitem{Jafferis:2015del}
D.~L. Jafferis, A.~Lewkowycz, J.~Maldacena and S.~J. Suh, \emph{{Relative
  entropy equals bulk relative entropy}},
  \href{https://doi.org/10.1007/JHEP06(2016)004}{\emph{JHEP} {\bfseries 06}
  (2016) 004} [\href{https://arxiv.org/abs/1512.06431}{{\ttfamily
  1512.06431}}].

\bibitem{Dong:2016eik}
X.~Dong, D.~Harlow and A.~C. Wall, \emph{{Reconstruction of Bulk Operators
  within the Entanglement Wedge in Gauge-Gravity Duality}},
  \href{https://doi.org/10.1103/PhysRevLett.117.021601}{\emph{Phys. Rev. Lett.}
  {\bfseries 117} (2016) 021601}
  [\href{https://arxiv.org/abs/1601.05416}{{\ttfamily 1601.05416}}].

\bibitem{Cotler:2017erl}
J.~Cotler, P.~Hayden, G.~Penington, G.~Salton, B.~Swingle and M.~Walter,
  \emph{{Entanglement Wedge Reconstruction via Universal Recovery Channels}},
  \href{https://doi.org/10.1103/PhysRevX.9.031011}{\emph{Phys. Rev. X}
  {\bfseries 9} (2019) 031011}
  [\href{https://arxiv.org/abs/1704.05839}{{\ttfamily 1704.05839}}].

\bibitem{Headrick:2014cta}
M.~Headrick, V.~E. Hubeny, A.~Lawrence and M.~Rangamani, \emph{{Causality \&
  holographic entanglement entropy}},
  \href{https://doi.org/10.1007/JHEP12(2014)162}{\emph{JHEP} {\bfseries 12}
  (2014) 162} [\href{https://arxiv.org/abs/1408.6300}{{\ttfamily 1408.6300}}].

\bibitem{Hubeny:2006yu}
V.~E. Hubeny, H.~Liu and M.~Rangamani, \emph{{Bulk-cone singularities \&
  signatures of horizon formation in AdS/CFT}},
  \href{https://doi.org/10.1088/1126-6708/2007/01/009}{\emph{JHEP} {\bfseries
  01} (2007) 009} [\href{https://arxiv.org/abs/hep-th/0610041}{{\ttfamily
  hep-th/0610041}}].

\bibitem{Hammersley:2006cp}
J.~Hammersley, \emph{{Extracting the bulk metric from boundary information in
  asymptotically AdS spacetimes}},
  \href{https://doi.org/10.1088/1126-6708/2006/12/047}{\emph{JHEP} {\bfseries
  12} (2006) 047} [\href{https://arxiv.org/abs/hep-th/0609202}{{\ttfamily
  hep-th/0609202}}].

\bibitem{Folkestad:2021kyz}
{\AA}.~Folkestad and S.~Hern\'andez-Cuenca, \emph{{Conformal rigidity from
  focusing}}, \href{https://doi.org/10.1088/1361-6382/ac27ef}{\emph{Class.
  Quant. Grav.} {\bfseries 38} (2021) 215005}
  [\href{https://arxiv.org/abs/2106.09037}{{\ttfamily 2106.09037}}].

\bibitem{Maldacena:2015iua}
J.~Maldacena, D.~Simmons-Duffin and A.~Zhiboedov, \emph{{Looking for a bulk
  point}}, \href{https://doi.org/10.1007/JHEP01(2017)013}{\emph{JHEP}
  {\bfseries 01} (2017) 013}
  [\href{https://arxiv.org/abs/1509.03612}{{\ttfamily 1509.03612}}].

\bibitem{Engelhardt:2016wgb}
N.~Engelhardt and G.~T. Horowitz, \emph{{Towards a Reconstruction of General
  Bulk Metrics}},
  \href{https://doi.org/10.1088/1361-6382/34/1/015004}{\emph{Class. Quant.
  Grav.} {\bfseries 34} (2017) 015004}
  [\href{https://arxiv.org/abs/1605.01070}{{\ttfamily 1605.01070}}].

\bibitem{Engelhardt:2016crc}
N.~Engelhardt and G.~T. Horowitz, \emph{{Recovering the spacetime metric from a
  holographic dual}},
  \href{https://doi.org/10.4310/ATMP.2017.v21.n7.a2}{\emph{Adv. Theor. Math.
  Phys.} {\bfseries 21} (2017) 1635}
  [\href{https://arxiv.org/abs/1612.00391}{{\ttfamily 1612.00391}}].

\bibitem{Hernandez-Cuenca:2020ppu}
S.~Hern\'andez-Cuenca and G.~T. Horowitz, \emph{{Bulk reconstruction of metrics
  with a compact space asymptotically}},
  \href{https://doi.org/10.1007/JHEP08(2020)108}{\emph{JHEP} {\bfseries 08}
  (2020) 108} [\href{https://arxiv.org/abs/2003.08409}{{\ttfamily
  2003.08409}}].

\bibitem{Hubeny:2012ry}
V.~E. Hubeny, \emph{{Extremal surfaces as bulk probes in AdS/CFT}},
  \href{https://doi.org/10.1007/JHEP07(2012)093}{\emph{JHEP} {\bfseries 07}
  (2012) 093} [\href{https://arxiv.org/abs/1203.1044}{{\ttfamily 1203.1044}}].

\bibitem{Bilson:2010ff}
S.~Bilson, \emph{{Extracting Spacetimes using the AdS/CFT Conjecture: Part
  II}}, \href{https://doi.org/10.1007/JHEP02(2011)050}{\emph{JHEP} {\bfseries
  02} (2011) 050} [\href{https://arxiv.org/abs/1012.1812}{{\ttfamily
  1012.1812}}].

\bibitem{Bao:2019bib}
N.~Bao, C.~Cao, S.~Fischetti and C.~Keeler, \emph{{Towards Bulk Metric
  Reconstruction from Extremal Area Variations}},
  \href{https://doi.org/10.1088/1361-6382/ab377f}{\emph{Class. Quant. Grav.}
  {\bfseries 36} (2019) 185002}
  [\href{https://arxiv.org/abs/1904.04834}{{\ttfamily 1904.04834}}].

\bibitem{Ryu_2006}
S.~Ryu and T.~Takayanagi, \emph{Holographic derivation of entanglement entropy
  from the anti–de sitter space/conformal field theory correspondence},
  \href{https://doi.org/10.1103/physrevlett.96.181602}{\emph{Physical Review
  Letters} {\bfseries 96} (2006) }.

\bibitem{Hubeny_2007}
V.~E. Hubeny, M.~Rangamani and T.~Takayanagi, \emph{A covariant holographic
  entanglement entropy proposal},
  \href{https://doi.org/10.1088/1126-6708/2007/07/062}{\emph{Journal of High
  Energy Physics} {\bfseries 2007} (2007) 062–062}.

\bibitem{VanRaamsdonk:2010pw}
M.~Van~Raamsdonk, \emph{{Building up spacetime with quantum entanglement}},
  \href{https://doi.org/10.1142/S0218271810018529}{\emph{Gen. Rel. Grav.}
  {\bfseries 42} (2010) 2323}
  [\href{https://arxiv.org/abs/1005.3035}{{\ttfamily 1005.3035}}].

\bibitem{Takayanagi:2017knl}
T.~Takayanagi and K.~Umemoto, \emph{{Entanglement of purification through
  holographic duality}},
  \href{https://doi.org/10.1038/s41567-018-0075-2}{\emph{Nature Phys.}
  {\bfseries 14} (2018) 573}
  [\href{https://arxiv.org/abs/1708.09393}{{\ttfamily 1708.09393}}].

\bibitem{Dutta:2019gen}
S.~Dutta and T.~Faulkner, \emph{{A canonical purification for the entanglement
  wedge cross-section}},
  \href{https://doi.org/10.1007/JHEP03(2021)178}{\emph{JHEP} {\bfseries 03}
  (2021) 178} [\href{https://arxiv.org/abs/1905.00577}{{\ttfamily
  1905.00577}}].

\bibitem{Hayden:2004}
P.~Hayden, R.~Jozsa, D.~Petz and A.~Winter, \emph{Structure of states which
  satisfy strong subadditivity of quantum entropy with equality},
  \href{https://doi.org/10.1007/s00220-004-1049-z}{\emph{Communications in
  Mathematical Physics} {\bfseries 246} (2004) 359–374}.

\bibitem{Casini:2017roe}
H.~Casini, E.~Teste and G.~Torroba, \emph{{Modular Hamiltonians on the null
  plane and the Markov property of the vacuum state}},
  \href{https://doi.org/10.1088/1751-8121/aa7eaa}{\emph{J. Phys. A} {\bfseries
  50} (2017) 364001} [\href{https://arxiv.org/abs/1703.10656}{{\ttfamily
  1703.10656}}].

\bibitem{Hayden:2011ag}
P.~Hayden, M.~Headrick and A.~Maloney, \emph{{Holographic Mutual Information is
  Monogamous}}, \href{https://doi.org/10.1103/PhysRevD.87.046003}{\emph{Phys.
  Rev. D} {\bfseries 87} (2013) 046003}
  [\href{https://arxiv.org/abs/1107.2940}{{\ttfamily 1107.2940}}].

\bibitem{Bao:2015bfa}
N.~Bao, S.~Nezami, H.~Ooguri, B.~Stoica, J.~Sully and M.~Walter, \emph{{The
  Holographic Entropy Cone}},
  \href{https://doi.org/10.1007/JHEP09(2015)130}{\emph{JHEP} {\bfseries 09}
  (2015) 130} [\href{https://arxiv.org/abs/1505.07839}{{\ttfamily
  1505.07839}}].

\bibitem{HernandezCuenca:2019wgh}
S.~Hern\'andez~Cuenca, \emph{{Holographic entropy cone for five regions}},
  \href{https://doi.org/10.1103/PhysRevD.100.026004}{\emph{Phys. Rev. D}
  {\bfseries 100} (2019) 2} [\href{https://arxiv.org/abs/1903.09148}{{\ttfamily
  1903.09148}}].

\bibitem{Hayden:2016cfa}
P.~Hayden, S.~Nezami, X.-L. Qi, N.~Thomas, M.~Walter and Z.~Yang,
  \emph{{Holographic duality from random tensor networks}},
  \href{https://doi.org/10.1007/JHEP11(2016)009}{\emph{JHEP} {\bfseries 11}
  (2016) 009} [\href{https://arxiv.org/abs/1601.01694}{{\ttfamily
  1601.01694}}].

\bibitem{pippenger2003inequalities}
N.~Pippenger, \emph{The inequalities of quantum information theory},
  {\emph{IEEE Transactions on Information Theory} {\bfseries 49} (2003) 773}.

\bibitem{n6wip}
D.~Avis and S.~Hern\'andez-Cuenca, \emph{{The Six-Party Holographic Entropy
  Cone}}, {\emph{Work in progress} (2022) }.

\bibitem{Avis:2021xnz}
D.~Avis and S.~Hern\'andez-Cuenca, \emph{{On the foundations and extremal
  structure of the holographic entropy cone}},
  \href{https://arxiv.org/abs/2102.07535}{{\ttfamily 2102.07535}}.

\bibitem{Hubeny:2018trv}
V.~E. Hubeny, M.~Rangamani and M.~Rota, \emph{{Holographic entropy relations}},
  \href{https://doi.org/10.1002/prop.201800067}{\emph{Fortsch. Phys.}
  {\bfseries 66} (2018) 1800067}
  [\href{https://arxiv.org/abs/1808.07871}{{\ttfamily 1808.07871}}].

\bibitem{Hubeny:2018ijt}
V.~E. Hubeny, M.~Rangamani and M.~Rota, \emph{{The holographic entropy
  arrangement}}, \href{https://doi.org/10.1002/prop.201900011}{\emph{Fortsch.
  Phys.} {\bfseries 67} (2019) 1900011}
  [\href{https://arxiv.org/abs/1812.08133}{{\ttfamily 1812.08133}}].

\bibitem{Hernandez-Cuenca:2019jpv}
S.~Hern\'andez-Cuenca, V.~E. Hubeny, M.~Rangamani and M.~Rota, \emph{{The
  quantum marginal independence problem}},
  \href{https://arxiv.org/abs/1912.01041}{{\ttfamily 1912.01041}}.

\bibitem{Sorce:2019zce}
J.~Sorce, \emph{{Holographic entanglement entropy is cutoff-covariant}},
  \href{https://doi.org/10.1007/JHEP10(2019)015}{\emph{JHEP} {\bfseries 10}
  (2019) 015} [\href{https://arxiv.org/abs/1908.02297}{{\ttfamily
  1908.02297}}].

\bibitem{Bao:2020mqq}
N.~Bao, N.~Cheng, S.~Hern\'andez-Cuenca and V.~P. Su, \emph{{A Gap Between the
  Hypergraph and Stabilizer Entropy Cones}},
  \href{https://arxiv.org/abs/2006.16292}{{\ttfamily 2006.16292}}.

\bibitem{Witten:2018zxz}
E.~Witten, \emph{{APS Medal for Exceptional Achievement in Research: Invited
  article on entanglement properties of quantum field theory}},
  \href{https://doi.org/10.1103/RevModPhys.90.045003}{\emph{Rev. Mod. Phys.}
  {\bfseries 90} (2018) 045003}
  [\href{https://arxiv.org/abs/1803.04993}{{\ttfamily 1803.04993}}].

\bibitem{Marolf:2017shp}
D.~Marolf, M.~Rota and J.~Wien, \emph{{Handlebody phases and the polyhedrality
  of the holographic entropy cone}},
  \href{https://doi.org/10.1007/JHEP10(2017)069}{\emph{JHEP} {\bfseries 10}
  (2017) 069} [\href{https://arxiv.org/abs/1705.10736}{{\ttfamily
  1705.10736}}].

\bibitem{Wall:2012uf}
A.~C. Wall, \emph{{Maximin Surfaces, and the Strong Subadditivity of the
  Covariant Holographic Entanglement Entropy}},
  \href{https://doi.org/10.1088/0264-9381/31/22/225007}{\emph{Class. Quant.
  Grav.} {\bfseries 31} (2014) 225007}
  [\href{https://arxiv.org/abs/1211.3494}{{\ttfamily 1211.3494}}].

\bibitem{Rota:2017ubr}
M.~Rota and S.~J. Weinberg, \emph{{New constraints for holographic entropy from
  maximin: A no-go theorem}},
  \href{https://doi.org/10.1103/PhysRevD.97.086013}{\emph{Phys. Rev. D}
  {\bfseries 97} (2018) 086013}
  [\href{https://arxiv.org/abs/1712.10004}{{\ttfamily 1712.10004}}].

\bibitem{Bao:2018wwd}
N.~Bao and M.~Mezei, \emph{{On the Entropy Cone for Large Regions at Late
  Times}},  \href{https://arxiv.org/abs/1811.00019}{{\ttfamily 1811.00019}}.

\bibitem{Czech:2019lps}
B.~Czech and X.~Dong, \emph{{Holographic Entropy Cone with Time Dependence in
  Two Dimensions}}, \href{https://doi.org/10.1007/JHEP10(2019)177}{\emph{JHEP}
  {\bfseries 10} (2019) 177}
  [\href{https://arxiv.org/abs/1905.03787}{{\ttfamily 1905.03787}}].

\bibitem{Hernandez-Cuenca:2022a}
S.~Hern\'andez-Cuenca, V.~Hubeny and M.~Rota, \emph{{Is the holographic entropy
  polyhedron the entropy cone? Work in progress}}, .

\bibitem{Cu85}
W.~H. Cunningham, \emph{On submodular function minimization},
  \href{https://doi.org/10.1007/BF02579361}{\emph{Comb.} {\bfseries 5} (1985)
  185}.

\bibitem{Fu91}
S.~Fujishige, \emph{Submodular Functions and Optimization}, vol.~58 of
  \emph{Ann. Discrete Math.} Elsevier, 2005,
  \href{https://doi.org/10.1016/S0167-5060(13)71057-4}{10.1016/S0167-5060(13)71057-4}.

\bibitem{Fukuda:2015}
K.~Fukuda, ``{Lecture: Polyhedral Computation, Spring 2015}.''
  \url{http://www-oldurls.inf.ethz.ch/personal/fukudak/lect/pclect/notes2015/PolyComp2015.pdf}.

\bibitem{he:2022}
T.~He, S.~Hern\'andez~Cuenca, V.~E. Hubeny and M.~Rota, \emph{{Work in
  progress}}, .

\bibitem{ghh:minimax}
G.~Grimaldi, M.~Headrick and V.~E. Hubeny, \emph{{Work in progress}}, .

\bibitem{Dong:2016hjy}
X.~Dong, A.~Lewkowycz and M.~Rangamani, \emph{{Deriving covariant holographic
  entanglement}}, \href{https://doi.org/10.1007/JHEP11(2016)028}{\emph{JHEP}
  {\bfseries 11} (2016) 028}
  [\href{https://arxiv.org/abs/1607.07506}{{\ttfamily 1607.07506}}].

\bibitem{Bao:2020zgx}
N.~Bao, N.~Cheng, S.~Hern\'andez-Cuenca and V.~P. Su, \emph{{The Quantum
  Entropy Cone of Hypergraphs}},
  \href{https://doi.org/10.21468/SciPostPhys.9.5.067}{\emph{SciPost Phys.}
  {\bfseries 9} (2020) 067} [\href{https://arxiv.org/abs/2002.05317}{{\ttfamily
  2002.05317}}].

\bibitem{Walter:2020zvt}
M.~Walter and F.~Witteveen, \emph{{Hypergraph min-cuts from quantum
  entropies}}, \href{https://doi.org/10.1063/5.0043993}{\emph{J. Math. Phys.}
  {\bfseries 62} (2021) 092203}
  [\href{https://arxiv.org/abs/2002.12397}{{\ttfamily 2002.12397}}].

\bibitem{Bao:2021gzu}
N.~Bao, N.~Cheng, S.~Hern\'andez-Cuenca and V.~P. Su, \emph{{Topological Link
  Models of Multipartite Entanglement}},
  \href{https://arxiv.org/abs/2109.01150}{{\ttfamily 2109.01150}}.

\bibitem{Akers:2021lms}
C.~Akers, S.~Hern\'andez-Cuenca and P.~Rath, \emph{{Quantum Extremal Surfaces
  and the Holographic Entropy Cone}},
  \href{https://doi.org/10.1007/JHEP11(2021)177}{\emph{JHEP} {\bfseries 11}
  (2021) 177} [\href{https://arxiv.org/abs/2108.07280}{{\ttfamily
  2108.07280}}].

\end{thebibliography}
\end{document}